\documentclass[12pt]{article}

\usepackage[margin = 1in]{geometry}
\usepackage{amsgen,amsmath,amstext,amsbsy,amsopn,amssymb,amsthm}
\usepackage[dvips]{graphicx}
\usepackage[colorlinks=true,linkcolor=blue,citecolor=blue,urlcolor=blue]{hyperref}
\usepackage{color}
\usepackage{xcolor}
\usepackage[round]{natbib}
\usepackage[normalem]{ulem}
\usepackage{dsfont}
\usepackage{float}
\usepackage{subcaption}
\usepackage{setspace}
\usepackage[titlenumbered,ruled,linesnumbered]{algorithm2e}
\usepackage{algpseudocode}
\usepackage{booktabs}
\usepackage{multirow}
\usepackage{makecell} 

\usepackage{mdframed} 
\usepackage{enumitem} 
\usepackage{letltxmacro}

\floatstyle{plain}
\newfloat{guidelinebox}{htbp}{loc}
\floatname{guidelinebox}{Box}

\newcommand{\one}{\mathds{1}}
\newcommand{\cee}{\text{CEE}}
\newcommand{\mee}{\text{MEE}}
\newcommand{\ate}{\text{ATE}}
\newcommand{\rate}{\text{rATE}}

\newcommand{\rank}{\text{rank}}
\newcommand{\aaa}{\text{AA}}
\newcommand{\bA}{\bar{A}}
\newcommand{\ba}{\bar{a}}
\newcommand{\RR}{\mathbb{R}}
\newcommand{\PP}{\mathbb{P}}
\newcommand{\II}{\mathbb{I}}
\newcommand{\tp}{\tilde{p}}
\newcommand{\pto}{\stackrel{p}{\to}}
\newcommand{\dto}{\stackrel{d}{\to}}
\newcommand{\w}{\text{w}}

\newcommand{\tD}{\tilde{D}}
\newcommand{\cH}{\mathcal{H}}
\newcommand{\tcH}{\tilde{\mathcal{H}}}
\newcommand{\tL}{\tilde{L}}

\newcommand{\tc}{\tilde{c}}
\newcommand{\cT}{\mathcal{T}}
\newcommand{\cP}{\mathcal{P}}

\newcommand{\aspn}{\text{ASPN}}
\newcommand{\spnc}{\text{SPNC}}


\newcommand{\revisionadd}[1]{#1}
\newenvironment{revisionaddblock}{}{}
\newcommand{\revisiondel}[1]{}

\usepackage[capitalize, nameinlink, sort]{cleveref}
\crefname{thm}{Theorem}{Theorems}
\Crefname{thm}{Theorem}{Theorems}
\crefname{cor}{Corollary}{Corollary}
\Crefname{cor}{Corollary}{Corollary}
\crefname{lem}{Lemma}{Lemmas}
\Crefname{lem}{Lemma}{Lemmas}
\crefname{asu}{Assumption}{Assumptions}
\Crefname{asu}{Assumption}{Assumptions}
\crefname{rmk}{Remark}{Remarks}
\Crefname{rmk}{Remark}{Remarks}
\crefname{defn}{Definition}{Definitions}
\Crefname{defn}{Definition}{Definitions}
\crefname{thmlisti}{Theorem}{Theorems}
\Crefname{thmlisti}{Theorem}{Theorems}
\crefname{asulisti}{Assumption}{Assumptions}
\Crefname{asulisti}{Assumption}{Assumptions}
\crefname{guidelinebox}{box}{boxes} 
\Crefname{guidelinebox}{Box}{Boxes} 

\newlist{thmlist}{enumerate}{1}
\setlist[thmlist]{label=(\roman*), ref=\thethm\,(\roman*)}

\newlist{asulist}{enumerate}{1}
\setlist[asulist]{label=(\roman*), ref=\theasu\,(\roman*)}

\usepackage{thmtools}
\declaretheorem[
    name=Theorem,
    Refname={Theorem,Theorems}]{thm}
\declaretheorem[
    name=Lemma,
    Refname={Lemma,Lemmas},
    sibling=thm]{lem}

\declaretheorem[
    name=Assumption,
    Refname={Assumption,Assumptions}]{asu}
\declaretheorem[
    name=Remark,
    Refname={Remark,Remarks}]{rmk}

\usepackage{xpatch}
\xapptocmd\normalsize{%
 \abovedisplayskip=4pt plus 2pt minus 2pt
 \belowdisplayskip=4pt plus 2pt minus 2pt
}{}{}

\title{Micro-randomized Trials with Categorical Treatments and Binary Proximal Outcome: Causal Effect Estimation and Sample Size Calculation}
\author{Jeremy Lin, Tianchen Qian\thanks{University of California, Irvine. Email: t.qian@uci.edu}}
\date{Department of Statistics, University of California, Irvine}

\def\spacingset#1{\renewcommand{\baselinestretch}%
{#1}\small\normalsize} \spacingset{1}

\begin{document}

\maketitle

\spacingset{1.2}

\begin{abstract}

Micro-randomized trials (MRTs) provide a framework for evaluating the marginal and moderated effects of mobile health (mHealth) interventions. In many applications, treatments take the form of categorical variables with multiple levels, such as different message contents or delivery strategies. Many scientifically meaningful longitudinal outcomes in mHealth studies are binary, such as whether a participant opens an app, engages with content, or completes a target behavior following a decision point at which treatment is randomized. This paper focuses on MRTs with categorical treatments and binary proximal outcomes. We define the causal excursion effect, propose an estimator called EMEE-catA, and derive a sample size formula for comparing categorical treatment levels that controls the type I error rate and guarantees power under working assumptions. We conduct extensive simulation studies to evaluate the operating characteristics of the proposed sample size formula, including robustness to violations of these assumptions. We further provide practical guidance for implementing the proposed approach to ensure adequate power in real-world MRTs. The methods are illustrated using data from the Drink Less MRT.

\noindent \textbf{Keywords:} availability, binary outcome, categorical treatment, causal excursion effect, micro-randomized trial, sample size calculation
\end{abstract}

\newpage

\tableofcontents

\newpage

\spacingset{1.9}


\section{Introduction}
\label{sec:introduction}

Support for healthy behavior change can be delivered through mobile technologies, such as smartphones and wearable devices, as part of mobile health (mHealth) interventions. These interventions often consist of reminders in the form of push notifications, text messages, or audible alerts delivered at randomized decision points to facilitate behavior change. Micro-randomized trials (MRTs) are a common experimental design for developing and optimizing mHealth interventions \citep{dempsey2015randomised, liao2016sample,qian2020micro}. A defining feature of an MRT is the repeated randomization of participants among intervention options at multiple decision points; this repeated randomization enables causal inference for marginal and moderated treatment effects, allowing researchers to learn whether, when, and in what contexts an intervention option is effective. The key estimand in such analyses is the causal excursion effect (CEE) \citep{boruvka2018assessing, qian2021estimating}, which is defined on the proximal outcome, a time-varying outcome measured shortly after each micro-randomization. For a binary treatment, such as message versus no message, the CEE at a decision point contrasts the proximal outcome under delivering the message with the proximal outcome under withholding it. CEEs are used to inform which intervention components to keep, modify, or drop in the next iteration of the intervention \citep{klasnja2019efficacy,nahum2021translating,phillips2026pilot,golbus2026impact}.

In an MRT, intervention options denote the full set of actions that may be delivered at a given decision point. In practice, the number of intervention options often exceeds two, reflecting the use of different behavioral theories as well as variations in message framing or content. For example, in the Drink Less MRT, an mHealth intervention aimed at reducing harmful alcohol consumption, the intervention options consisted of messages framed in multiple ways, as well as the option of no push notification \citep{bell2020notifications}. In the MARS MRT, participants were micro-randomized among three options: no prompt, a prompt recommending a brief self-regulatory strategy, and a prompt recommending a more effortful strategy \citep{nahumshani2021mars}. In the MiWaves MRT, which evaluated a cannabis reduction app for emerging adults, participants were micro-randomized \revisiondel{to no message or to one of six message types}\revisionadd{among seven options: no message and six message types} that varied in length and in the response the message asked for \citep{coughlin2024miwaves}. In this paper, we refer to such intervention options as categorical treatments. Evaluating differential effects among these options is critical for optimizing intervention design.

Our previous work proposed an estimator for the causal effects of categorical treatments in MRTs with a continuous proximal outcome \citep{lin2025micro}. However, binary proximal outcomes are also very common in MRTs and often the primary focus, because they naturally capture engagement or adherence, the target behavior of many mHealth interventions. For example, in the Drink Less MRT, researchers were primarily interested in whether the participant opened the app following a prompt. In MARS, the primary proximal outcome is whether the individual engaged in a self-regulatory strategy following a decision point \citep{nahumshani2021mars}.

In this paper, we focus on MRTs with categorical treatments and a binary proximal outcome. We extend the Estimator for Marginal Excursion Effects (EMEE) method proposed by \citet{qian2021estimating} to accommodate categorical treatments for MRTs with binary outcomes. We call the proposed estimator EMEE-catA, where catA indicates categorical treatments. We establish the asymptotic normality of the estimator even when the nuisance outcome regression model is misspecified. This robustness is particularly important in MRTs, where the high-dimensional nature of the observed history can make correct model specification challenging.

Beyond estimation, we develop a sample size formula for detecting a pre-specified differential effect among treatment options, such as a difference between the effects of two active options. More generally, the formula accommodates any hypothesis that can be written as a linear combination of the marginal causal excursion effects of the active options, including a simultaneous test that several options have no effect. We show that, under a set of working assumptions, the proposed sample size formula achieves the desired power while controlling the type I error rate. Through extensive simulation studies, we demonstrate that violations of certain working assumptions do not affect power; for assumptions whose violations do impact power, we carefully examine and discuss the resulting implications. We also provide practical guidance for applying the sample size formula in practice. This work extends the existing literature on sample size calculations for MRTs, which previously focused on binary treatment options, by enabling more granular comparisons among multiple treatment levels \citep{liao2016sample, cohn2023sample}.

The remainder of the paper is organized as follows. \Cref{sec:definition} introduces notation and defines the causal excursion effect for categorical treatments. \revisiondel{section 3}\revisionadd{\Cref{sec:estimator}} presents the proposed estimator for the causal excursion effect. \Cref{sec:sample-size-formula} describes the sample size formula and provides practical guidance for its use. \Cref{sec:application} presents case studies illustrating the estimator and sample size methodology using data from the Drink Less MRT. \Cref{sec:discussion} concludes with a discussion. Due to space constraints, comprehensive simulation studies evaluating the estimator and sample size formula are provided in Sections~\ref{sec:simulation-estimator} and~\ref{sec:simulation-sample-size-formula} of the Supplementary Material.

\section{Definition and Assumptions}
\label{sec:definition}

\subsection{Notation}

Consider an MRT that includes data from $n$ individuals with $T$ decision points where treatments are randomized. Throughout this paper, omission of the subscript $i$ implies that variables correspond to observations from a typical individual. Let $A_t$ denote the treatment assignment at decision point $t$. $A_t$ can take categorical value $A_t \in \{0, 1, \cdots, K\}$, with 0 being the reference level (e.g., no treatment). Let $X_t$ denote the information collected for the individual after decision point $t-1$ and up to decision point $t$; in particular, $X_1$ includes baseline covariates. Let $I_t \in X_t$ be the availability indicator, where $A_t = 0$ deterministically when $I_t = 0$ (for example, if the individual is not available). The overbar notation is used to denote a sequence of variables up to a decision point; for instance, $\bA_t = (A_1, \ldots, A_t)$. Information observed up to decision point $t$ is denoted by $H_t = (X_1, A_1, X_2, A_2, \cdots, X_{t-1}, A_{t-1}, X_t) = (\bar{X}_t, \bA_{t-1})$.
The randomization probability for $A_t$ may depend on $H_t$ and is denoted by $p_t(k | H_t) = P(A_t = k \mid H_t)$ for $0 \leq k \leq K$. The observed data for a typical individual are represented by $O = (X_1, A_1, \cdots, X_T, A_T, X_{T + 1})$. We assume that the data from different individuals are independent and identically distributed draws from an unknown distribution $\mathcal{P}$. We consider settings where the proximal outcome $Y_{t,\Delta}$ for the treatment at decision point $t$ may be measured over a time window of subsequent $\Delta$ decision points; here $\Delta$ is a fixed positive integer. Specifically, let $Y_{t,\Delta} = y(X_{t+1}, A_{t+1}, \cdots, X_{t + \Delta -1}, A_{t + \Delta -1}, X_{t + \Delta})$ for a given function $y(\cdot)$. We will use $Y_t$ to denote $Y_{t,\Delta = 1}$, the most commonly used proximal outcome in practice. Throughout this paper, all expected values are taken with respect to $\mathcal{P}$, unless stated otherwise.

For an arbitrary function $f(\cdot)$ of the generic observed data $O$, we define $\mathbb{P}_n f(O)$ as the sample mean $\frac{1}{n} \sum_{i = 1}^n f(O_i)$, where $O_i$ denotes the $i$-th individual's data. We use $\bar{0}_m$ to denote the $1 \times m$ row vector with all entries equal to 0. We use $\RR^{m \times n}$ to denote the space of $m\times n$ real matrices and $\RR^m$ to denote the space of $m \times 1$ real column vectors. We use $\one(\cdot)$ to denote the indicator function. We use $\II_p$ to denote the $p\times p$ identity matrix. For a positive integer $n$, we use $[n]$ to denote the set $\{1,2,\ldots,n\}$. For a vector $\alpha$ and a vector-valued function $f(\alpha)$, let $\dot{f}(\alpha) := \partial f(\alpha) / \partial \alpha^T$ denote the matrix where the $(i,j)$-th entry is the partial derivative of the $i$-th entry of $f$ with respect to the $j$-th entry of $\alpha$. We use $\otimes$ to denote the Kronecker product between two matrices (or vectors): for $B \in \RR^{d_1 \times d_2}$ and $C \in \RR^{d_3 \times d_4}$, $B \otimes C \in \RR^{d_1d_3 \times d_2d_4}$ is defined as
\begin{align*}
        B \otimes C = \begin{bmatrix}
        b_{1, 1}C &b_{1, 2}C  &\cdots &b_{1, d_2}C\\
        b_{2, 1}C &b_{2, 2}C  &\cdots &b_{2, d_2}C\\
        \vdots  &\vdots   &\vdots &\vdots\\
        b_{d_1, 1}C &b_{d_1, 2}C  &\cdots &b_{d_1, d_2}C\\
    \end{bmatrix}, 
\end{align*}
where $b_{i,j}$ is the $(i,j)$-th entry of $B$. Finally we define $0/0 = 0$.

\subsection{Potential outcomes and causal excursion effect}
To define treatment effects, we adopt the potential outcomes framework \citep{rubin1974estimating, robins1986new}. For an individual, let $X_t(\ba_{t-1})$ be the information that would have been observed had that individual been assigned to treatment sequence $\ba_{t-1}$. The potential outcomes for a typical individual are $\{X_1, A_1, X_2(a_1), A_2, \cdots, A_T, X_{T+1}(\ba_T): \text{for all } 0 \leq a_s \leq K, s \in [T] \}$.
Denote by $Y_{t, \Delta}(\ba_{t + \Delta -1})$ the potential outcome of $Y_{t, \Delta}$. The potential outcome of the history information $H_t$ is denoted by $H_t(\ba_{t-1}) = \{X_1, A_1, X_2(A_1), A_2, X_3(\ba_2), \cdots X_t(\ba_{t-1})\}$. For any $k \in [K]$ and $t \in [T]$, we define the Causal Excursion Effect (CEE) of treatment $k$ at time $t$ using the log relative risk scale:
\begin{align}
    \cee_{tk}(S_t) & = \log\frac{E\{Y_{t,\Delta}(\bA_{t-1}, k, \bar{0}_{\Delta - 1})\mid S_t(\bA_{t-1}), I_t(\bA_{t-1}) = 1\}}{E\{Y_{t,\Delta}(\bA_{t-1}, 0, \bar{0}_{\Delta - 1})\mid S_t(\bA_{t-1}), I_t(\bA_{t-1}) = 1\}} \label{eq:cee-def}
\end{align}
where $\bar{0}_{\Delta - 1}$ represents a length $(\Delta - 1)$ vector of zeros. Expression \eqref{eq:cee-def} represents the contrast between the expected outcome under two excursions: getting treatment $k$ at time $t$ and no treatment for the next $\Delta -1$ time points, versus no treatment at time $t$ and no treatment for the next $\Delta - 1$ time points. In both excursions, the treatment assignment up to time $t$ ($\bA_{t-1}$) is random and follows the treatment protocol of the micro-randomized trial. In \eqref{eq:cee-def}, $S_t \subset H_t$ is a set of effect modifiers of interest, and researchers' different choice of $S_t$ leads to different interpretation of CEE. For example, setting $S_t = \emptyset$ assesses a marginal effect that averages over all possible moderators. Setting $S_t = X_t$ assesses effect moderation by current covariates $X_t$.

\begin{rmk}[Interpretation and practical relevance of the CEE]
    \label{rmk:cee-interpretation}
    \normalfont
    We comment on two aspects of definition \eqref{eq:cee-def}. First, the expectations marginalize over the treatment history $\bA_{t-1}$, which follows the randomization protocol of the trial; the CEE therefore depends on the MRT randomization policy \citep{boruvka2018assessing, dempsey2020stratified, huch2025data}. This dependence is useful for intervention development: the CEE measures the effect of an excursion from the trial policy beginning at decision point $t$, so its sign and magnitude at a given value of $S_t$ indicate whether delivering treatment $k$ more or less often in that context would improve the proximal outcome, relative to the policy the trial used \citep{klasnja2015micro, nahumshani2018jitai, golbus2021microrandomized, liu2023microrandomized}. Note that it is not comparing assigning treatment $k$ at all times vs. assigning treatment 0 at all times. Second, the expectations condition on the decision point being available ($I_t(\bA_{t-1}) = 1$). When an individual is unavailable, delivering treatment is infeasible or inappropriate, so a treatment contrast at unavailable decision points is neither nonparametrically identifiable nor of scientific interest; conditioning on availability restricts the estimand to the decision points at which intervening is a valid option \citep{klasnja2015micro, boruvka2018assessing, seewald2019practical, liao2016sample}.
\end{rmk}

\subsection{Identification of parameters}

To express the causal excursion effect using the observed data, we make the following assumptions.
\begin{asu}
    \label{asu:causal-assumptions} 
    \normalfont
    \begin{itemize}
        \item[(a)] (SUTVA.) The observed data are equal to the potential outcome under the observed treatment assignment, and one's potential outcomes are not affected by others' treatment assignments. Specifically, $X_t = X_t(\bA_{t-1})$ and $Y_{t, \Delta} = Y_{t,\Delta}(\bA_{t +\Delta-1})$ for all $t \in [T]$.
        \item[(b)] (Positivity.) $P(A_t = k \mid H_t, I_t = 1) > 0$ almost surely for all $0 \leq k \leq K$ and $t \in [T]$.
        \item[(c)] (Sequential ignorability.) For $t \in [T]$, the potential outcomes $\{ X_{t+1}(\ba_t), \cdots, X_{T+1}(\ba_T) : \text{for all } 0 \leq a_s \leq K, s \in [T] \}$ are independent of $A_t$ conditional on $H_t$.
    \end{itemize}
\end{asu}

In an MRT, $A_t$ is sequentially randomized according to known treatment probabilities, $\{p_{t}(k|H_t): 0 \leq k \leq K \}$, which guarantees positivity and sequential ignorability. SUTVA may be violated if the treatment assigned to one individual influences the response of others, such as in a social network. This paper does not consider such possibilities. Under \Cref{asu:causal-assumptions}, CEE in \eqref{eq:cee-def} can be expressed in terms of the observed data distribution as
\begin{align}
    \cee_{tk}(S_t) & = \log \frac{E\bigg[E\bigg\{Y_{t,\Delta} \prod_{j = t + 1}^{t + \Delta - 1} \frac{\one(A_j = 0)}{p_j(0|H_j)} \bigg| A_t = k, H_t\bigg\} \bigg| S_t, I_t = 1 \bigg]}{E\bigg[ E \bigg\{Y_{t,\Delta} \prod_{j = t + 1}^{t + \Delta - 1} \frac{\one(A_j = 0)}{p_j(0|H_j)} \bigg| A_t = 0, H_t \bigg\} \bigg| S_t, I_t = 1 \bigg]} \label{eq:cee-identification}
\end{align}
Equation \eqref{eq:cee-identification} follows immediately from Lemma A.1 in \citet{qian2021estimating}.


\section{Estimating CEE for Categorical Treatment}
\label{sec:estimator}

We consider estimating a finite-dimensional unknown parameter $\beta$ in a parametric model for CEE. Specifically, suppose
\begin{align}
    \cee_{tk}(S_t) = f_t(S_t)^T \beta_k \text{ for all } t \in [T] \text{ and } k \in [K], \label{eq:CEE-parametric-model}
\end{align}
where $f_t(\cdot)$ is a prespecified vector-valued function and $\beta_k \in \RR^p$. The $p\times T$ matrix $(f_1,f_2,\ldots,f_T)$ is required to be of rank $p$ in order for $\beta_k$ to be identified. Let $\beta = (\beta_1^T, \ldots, \beta_K^T)^T \in \RR^{Kp}$. This model allows for nonlinear effects: e.g., $f_t(S_t)$ could include basis functions of $t$.

We propose an estimator for $\beta$ that generalizes the estimator of marginalized causal excursion effect (EMEE) in \citet{qian2021estimating} to the categorical treatment setting. We refer to this estimator as EMEE-catA. For each $t$, let $\tp_t(k | S_t)$ be a function of $(k,S_t)$ such that $\tp_t(k | S_t) > 0$ for $0 \leq k \leq K$ and $\sum_{k=0}^K \tp_t(k | S_t) = 1$; we later discuss the choice of $\tp_t(k | S_t)$. Define $J_t := \frac{\tp_t(A_t|S_t)}{p_t(A_t|H_t)} \prod_{j = t + 1}^{t + \Delta - 1} \frac{\one(A_j = 0)}{p_j(0|H_j)}$. Let $C_k(A_t) := \one(A_t = k) - \tp_t(k | S_t)$. Let $g_t(H_t)^T\alpha$ be a working model for $E(Y_{t,\Delta} \mid H_t, I_t = 1, A_t = 0)$.
The proposed estimator $\hat\beta$ is obtained by solving for $(\hat\alpha,\hat\beta)$ that satisfies $\PP_n m(\hat\alpha,\hat\beta) = 0$, where
\begin{align}
    m(\alpha,\beta) &:= 
    \sum_t^T I_t J_t \exp\bigg\{-\sum_{k = 1}^K C_k(A_t) f_t(S_t)^T\beta_k \bigg\} \nonumber\\
    &\times\bigg[Y_{t,\Delta} - \exp \bigg\{\sum_{k = 1}^K C_k(A_t) f_t(S_t)^T \beta_k + g_t(H_t)^T \alpha\bigg\} \bigg] 
    \begin{bmatrix}
        g_t(H_t) \\
        C_1(A_t) f_t(S_t)\\
        \vdots\\
        C_K(A_t) f_t(S_t)
    \end{bmatrix}. \label{eq:ee}
\end{align}
The first term in $J_t$ is similar to a stabilized inverse probability weight in marginal structural models \citep{robins2000marginal}, which is included because of the marginal aspect of \eqref{eq:cee-def}. The numerator probabilities, $\tp_t(k|S_t)$, do not affect the consistency of $\hat\beta$ but should be chosen as close to $p_t(k|H_t)$ as possible for efficiency. For example, one may fit a multinomial logistic regression with response $A_t$ and predictor $S_t$ to predict $\tp_t(k|S_t)$. The second term in $J_t$ is an inverse probability weight due to the $\bar{0}_{\Delta-1}$ in the potential outcomes in \eqref{eq:cee-def}. The centering of $A_t$ in $C_k(A_t)$ produces orthogonality between the estimation of $\beta$ and the estimation of the nuisance parameter $\alpha$, and this leads to the robustness of $\hat\beta$ even when $g_t(H_t)^T\alpha$ is misspecified (see \Cref{thm:CAN}). This robustness property is desirable because $H_t$ can be high-dimensional in an MRT with a large number of time points, making it very difficult to model $E \{Y_{t,\Delta} \mid H_t, I_t =1, A_t = 0 \}$ correctly. It is useful to note when the weight is trivial: if $\Delta = 1$ and the randomization probability depends on the history only through $S_t$ (for example, when the randomization probability is constant), then we may set $\tp_t(k|S_t) = p_t(k|H_t)$, which gives $J_t \equiv 1$ so that no weighting is needed. The weight cannot be made equal to 1 in two situations. When the randomization probability depends on elements of $H_t$ beyond $S_t$, the first term of $J_t$ ensures that $\hat\beta$ targets the effect conditional on the lower-dimensional $S_t$ rather than on the full history $H_t$. When $\Delta > 1$, the second term in $J_t$ accounts for the requirement of no treatment during the subsequent $\Delta - 1$ decision points in the excursions being contrasted. We establish the asymptotic property of $\hat\beta$ in \Cref{thm:CAN}, which is proven in Section~\ref{sec:proof-thm-CAN} of the Supplementary Material.

\begin{thm}[Asymptotic normality.]
    \label{thm:CAN}
    Suppose the parametric CEE model \eqref{eq:CEE-parametric-model} and \Cref{asu:causal-assumptions} hold.
    Suppose that randomization probabilities $\{p_t(k | H_t)\}_{0 \leq k \leq K, t \in [T]}$ are known. Suppose $\beta^0$ is the true value of $\beta$ under the data-generating distribution $\mathcal{P}$. Let $\dot{m}$ be the derivative matrix of $m(\alpha, \beta)$ with respect to $(\alpha, \beta)$. Let $(\hat\alpha, \hat\beta)$ be the solution to $\mathbb{P}_n m(\alpha, \beta) = 0$. Under regularity conditions, we have the following.
    \begin{itemize}
        \item[(i)] There exists $\alpha'\in\RR^q$ such that $(\hat\alpha,\hat\beta) \pto (\alpha',\beta^0)$ and $\sqrt{n}(\hat\beta - \beta^0) \dto N(0, M^{-1}\Sigma M^{-1,T})$, where
        \begin{align*}
            M & := \sum_{t=1}^T E ( I_t J_t U_t Y_{t,\Delta} D_t D_t^T ), \qquad U_t := e^{-\sum_{k = 1}^K C_k(A_t) f_t(S_t)^T \beta_k^0}, \\
            \Sigma & := \sum_{t=1}^T \sum_{s=1}^T E \{ I_t I_s J_t J_s r_t(\alpha',\beta^0) r_s(\alpha',\beta^0) D_t D_s^T \}, \\
            r_t(\alpha,\beta) & := e^{-\sum_{k = 1}^K C_k(A_t) f_t(S_t)^T\beta_k }Y_{t,\Delta} - e^{ g_t(H_t)^T \alpha}, \\
            D_t & := (C_1(A_t)f_t(S_t)^T, C_2(A_t)f_t(S_t)^T, \ldots, C_K(A_t)f_t(S_t)^T )^T.
        \end{align*}
        \item[(ii)] $\hat{M}^{-1} \hat\Sigma \hat{M}^{-1, T}$ is a consistent estimator for the asymptotic variance $M^{-1}\Sigma M^{-1,T}$, with $\hat{M} := \sum_{t=1}^T \PP_n ( I_t J_t \hat U_t Y_{t,\Delta} D_t D_t^T )$ and $\hat\Sigma := \sum_{t=1}^T \sum_{s=1}^T \PP_n \{ I_t I_s J_t J_s r_t(\hat\alpha,\hat\beta) r_s(\hat\alpha,\hat\beta) D_t D_s^T \}$, where $\hat U_t := e^{-\sum_{k = 1}^K C_k(A_t) f_t(S_t)^T \hat\beta_k}$.
        \item[(iii)] When $\{\tp_t(k | S_t)\}_{0 \leq k \leq K, t \in [T]}$ is estimated either parametrically or nonparametrically, the theorem conclusion still holds with the following modification: $\tp_t(k | S_t)$ replaced by its estimator $\hat\tp_t(k | S_t)$ in $\hat{M}^{-1} \hat\Sigma \hat{M}^{-1, T}$, and $\tp_t(k | S_t)$ replaced by the $L_2$-limit of $\hat\tp_t(k | S_t)$ in $M^{-1}\Sigma M^{-1,T}$.
    \end{itemize}
\end{thm}

\begin{rmk}
    \normalfont
    When the parametric model for CEE \eqref{eq:CEE-parametric-model} is correct, the choice of $\{\tp_t(k | S_t)\}_{0 \leq k \leq K, t \in [T]}$ only affects the asymptotic variance of $\hat\beta$ and does not affect the consistency of $\hat\beta$. However, when model \eqref{eq:CEE-parametric-model} is misspecified, different choices of $\{\tp_t(k | S_t)\}_{0 \leq k \leq K, t \in [T]}$ lead to different probability limits of $\hat\beta$. For example, for the immediate effect ($\Delta = 1$), suppose that for each $k \in [K]$ the analysis model asserts a constant marginal effect $\cee_{tk} (\emptyset) = \beta_k$ for all $t \in [T]$, while the true $\cee_{tk}(\emptyset)$ is not constant in $t$. In this case, if for each $0 \leq k \leq K$ we set $\tp_t(k|S_t)$ to be a constant over $t$, then each $\hat\beta_k$ converges in probability to (recall $Y_t := Y_{t,\Delta = 1}$)
    \begin{align*}
        \beta_k'  = \log \frac{\sum_{t = 1}^T E\{E(Y_t \mid H_t, A_t = k)\mid I_t = 1 \} E(I_t)}{\sum_{t = 1}^T E\{E(Y_t \mid H_t, A_t = 0)\mid I_t = 1 \} E(I_t)},
    \end{align*}
    which further simplifies to 
    \begin{align*}
        \beta_k'  = \log \frac{\sum_{t = 1}^T E(Y_t \mid I_t = 1, A_t = k) E(I_t)}{\sum_{t = 1}^T E(Y_t \mid I_t = 1, A_t = 0) E(I_t)}
    \end{align*}
    if the randomization probability does not depend on $H_t$, i.e., $p_t(k|H_t) = p_t(k)$, so that $A_t$ is independent of $H_t$ given $I_t = 1$.
\end{rmk}

\begin{rmk}[Connection to GEE and M-estimation]
    \label{rmk:gee-mestimation}
    \normalfont
    The estimator $(\hat\alpha, \hat\beta)$ solving $\PP_n m(\alpha,\beta) = 0$ is an M-estimator and \Cref{thm:CAN} follows from standard M-estimator theory \citep{stefanski2002calculus}. The estimating function \eqref{eq:ee} sums contributions across decision points without modeling the dependence among them, analogous to a GEE under working independence \citep{zeger1986longitudinal}, and analogously the repeated measures are accounted for in the standard error estimate so no within-individual correlation structure needs to be modeled. Two features distinguish \eqref{eq:ee} from a GEE. First, the centered treatment indicators $C_k(A_t)$ ensure that consistency of $\hat\beta$ requires only the CEE model \eqref{eq:CEE-parametric-model} to be correct, while the nuisance model $g_t(H_t)^T\alpha$ may be arbitrarily misspecified. A GEE fit does not have this robustness property as it would in general require the conditional mean model to be correctly specified. Second, the weight $J_t$ accounts for the marginal, excursion-based definition of the estimand.
\end{rmk}



\section{Sample Size Formula with Categorical Treatment}
\label{sec:sample-size-formula}

\subsection[Hypothesis, Test Statistic, and Rejection Region that Controls the Type I Error Rate]{Hypothesis, Test Statistic, and Rejection Region that Controls the Type I Error Rate}
\label{subsec:hypothesis-test-statistic-rejection-region}

For sample size calculation, we focus on the most common setting in planning an MRT, where the interest is in the immediate, marginal CEE (i.e., $\Delta = 1$ and $S_t = \emptyset$ in \eqref{eq:cee-def}), and the randomization probability at each decision point is either constant or dependent only on $t$. In this setting, we define $p_t(k) := P(A_t = k \mid I_t = 1) \equiv P(A_t = k \mid H_t, I_t = 1)$, and we define the marginal excursion effect (MEE) as the CEE with $\Delta = 1$ and $S_t = \emptyset$ (recall $Y_t := Y_{t,\Delta=1}$):
\begin{align}
    \mee_k(t) & = \log\frac{E\{Y_{t}(\bA_{t-1}, k)\mid  I_t(\bA_{t-1}) = 1\}}{E\{Y_{t}(\bA_{t-1}, 0)\mid I_t(\bA_{t-1}) = 1\}}\quad \text{for } t \in [T], k \in [K].\label{eq:mee-def}
\end{align}
A positive $\mee_k(t)$ indicates that the treatment level $k$ at time $t$ is effective compared to no treatment (assuming that a larger $Y_t$ is better).

Let $\mee_{1:K}(t) = (\mee_1(t), \mee_2(t), \ldots, \mee_K(t))^T$. We consider testing for 
\begin{align}
    \cH_0: L \times \mee_{1:K}(t) = 0 \text{ for all } t \in [T] \text{ v.s. } \cH_1: L \times \mee_{1:K}(t) \neq 0 \text{ for some } t \in [T], \label{eq:null-and-alternative}
\end{align}
where $L \in \RR^{v \times K}$ for some integer $v$. The researcher can specify different $L$ to represent different scientific hypotheses. For example, if we set $v = 2$ and $L$ to be a $2 \times K$ matrix with $(1,1)$ and $(2,2)$ entries being $1$ and all other entries being 0, then this gives $\cH_0: \mee_1(t) = \mee_2(t) = 0$ for all $t \in [T]$, a simultaneous test for null effect (compared to the reference level) of both treatment options 1 and 2. If we set $v = 1$ and $L$ be a $1 \times K$ matrix with $(1,1)$-th entry being $1$, $(1,2)$-th entry being $-1$, and all other entries being 0, then this gives $\cH_0: \mee_1(t) = \mee_2(t)$ for all $t \in [T]$, a test for the equivalence of treatment levels 1 and 2.

Nonparametrically testing for all deviations from $L \times \mee_{1:K}(t) = 0$ for each $t$ will result in low power for detecting specific  deviations. Instead we consider a working model for $\mee_k(t)$ that is parametric in $t$. Specifically, we consider the working assumption 
\begin{align}
    \mee_k(t) = f_t^T\beta_k \text{ for } t \in [T] \text{ and } k \in [K], \label{eq:working-assumption-on-MEE}
\end{align}
where $\beta_k \in \RR^p$ and $f_t$ is a predefined $p$-dimensional vector-valued function that only depends on $t$. We focus on calculating sample size for detecting target alternatives that satisfy \eqref{eq:working-assumption-on-MEE}. The particular form of \eqref{eq:working-assumption-on-MEE} that is scientifically plausible is typically decided in consultation with the scientific team. For instance, if the scientific team anticipates that the intervention's effect will not substantially change over time, a constant-in-time $\mee_k(t)$ with $f_t = 1$ would be an appropriate choice. If the scientific team anticipates that the intervention's effect might start near zero, increase gradually early in the study, and potentially decrease later, a quadratic-in-time $\mee_k(t)$ with $f_t = (1, t, t^2)^T$ could be appropriate. We note that although the sample size formula will be derived under the parametric working assumption \eqref{eq:working-assumption-on-MEE}, we show via extensive simulation studies that the sample size formula can yield adequate power even under certain violations of the working assumption (see \Cref{subsec:practical-guideline} and Section~\ref{sec:simulation-sample-size-formula} of the Supplementary Material).

Under the working assumption \eqref{eq:working-assumption-on-MEE}, we derive in Section~\ref{subsec:proof-L-tilde-matrix} of the Supplementary Material that the null and the alternative hypotheses in \eqref{eq:null-and-alternative} are equivalent to
\begin{align}
    \tcH_0: \tL \beta = 0 \quad \text{v.s.} \quad \tcH_1: \tL \beta \neq 0, \label{eq:null-and-alternative-tilde}
\end{align}
where $\tL:= L \otimes \II_p$ and $\beta := (\beta_1^T, \ldots, \beta_K^T)^T$. We construct a Wald-type test statistic based on the estimator $\hat\beta$ proposed in \Cref{sec:estimator}. Specifically, let $g_t$ be a predefined $q$-dimensional vector-valued function that only depends on $t$, and let \revisiondel{$g_t^T\alpha$}\revisionadd{$\exp(g_t^T\alpha)$} be a working model for $E(Y_t \mid I_t = 1, A_t =0)$. The function $g_t$ should be chosen such that $p_t(k) f_t$ lies in the linear span of $g_t$ for each $k \in [K]$; for example, include $f_t$ as a subvector of $g_t$ when the randomization probabilities are constant in $t$. Let $(\hat\alpha, \hat\beta)$ be the estimator that solves $\PP_n m(\alpha,\beta) = 0$, with the definition of $m(\alpha,\beta)$ modified by replacing $f_t(S_t)$, $g_t(H_t)$, $\tp_t(k|S_t)$, $p_t(k|H_t)$, and $J_t$ by $f_t$, $g_t$, $p_t(k)$, $p_t(k)$, and 1, respectively.
The Wald-type test statistic is defined as
\begin{align}
        \cT = n (\tL\hat\beta)^T (\tL\hat{M}^{-1} \hat{\Sigma} \hat{M}^{-1, T} \tL^T)^{-1} (\tL \hat\beta), \label{eq:def-test-stat}
\end{align}
with $\hat{M}$ and $\hat\Sigma$ defined in \Cref{thm:CAN}(ii) but with the same modifications as those for $m(\alpha,\beta)$.

Under \eqref{eq:working-assumption-on-MEE} and $\tcH_0$, it follows from \Cref{thm:CAN} that the large sample distribution of $\cT$ is $\chi^2_l$, a chi-squared distribution with \revisiondel{$l = \rank(L)$}\revisionadd{$l = \rank(\tL)=p\cdot\rank(L)$} degrees of freedom. Thus, a hypothesis test that uses the critical value of the chi-squared distribution will asymptotically control the type I error rate at the nominal level. To correct the downward bias of the sandwich estimator $\hat{M}^{-1} \hat{\Sigma} \hat{M}^{-1, T}$ when the sample size $n$ is small, we instead use the critical value from a \revisiondel{scaled $F$-distributions}\revisionadd{scaled $F$ distribution} because $\frac{n - q - l}{l (n - q - 1)} \cT$ approximately follows $F_{l, n - q - l}$, an $F$-distribution with degrees of freedom $(l, n - q - l)$ \citep{pan2002small}. Therefore, the rejection region of a test with significance level $\eta$ is 
\begin{align}
    \bigg\{ \cT : \frac{n - q - l}{l (n - q - 1)} \cT > F^{-1}_{l, n - q -l} (1 - \eta) \bigg\}, \label{eq:rejection-region}
\end{align}
where $F^{-1}_{l, n - q - l}(1 - \eta)$ is the $(1 - \eta)$-quantile of $F_{l, n - q -l}$. We further incorporated another small sample correction in \citet{mancl2001covariance} by replacing $\hat{\Sigma}$ in \eqref{eq:def-test-stat} with an adjusted version using the ``hat'' matrix.


\subsection{Sample Size Formula that Guarantees Power Under Working Assumptions}
\label{subsec:sample-size-formula}

Suppose the goal is to find the sample size of the MRT to have at least $(1-b)$ power under some target alternative hypothesis.

Consider the target alternative with a parametric model on $\mee_k(t)$ that is compatible with \eqref{eq:working-assumption-on-MEE}: for given $\beta := (\beta_1^T,\ldots,\beta_K^T)$ where each $\beta_k \in \RR^p$, let 
\begin{align}
    \cH_1^*(\beta): \mee_k(t) = f_t^T \beta_k \text{ for } t \in [T] \text{ and } k \in [K]. \label{eq:target-alternative}
\end{align}

Under $\cH_1^*(\beta)$, it follows from \Cref{thm:CAN} that the test statistic $\cT$ follows approximately $\chi^2_l(\lambda(n))$, a non-central chi-squared distribution with \revisiondel{$l = \rank(L)$}\revisionadd{$l = \rank(\tL)=p\cdot\rank(L)$} degrees of freedom and non-centrality parameter $\lambda(n) := n(\tL\beta)^T (\tL M^{-1} \Sigma M^{-1, T} \tL^T)^{-1} (\tL\beta)$.

Here, $M$ and $\Sigma$ are defined in \Cref{thm:CAN}(i) with the same modifications as those for $m(\alpha,\beta)$ described above display \eqref{eq:def-test-stat}. To improve small sample performance, we again use the $F$-distribution approximation: the scaled test statistic $\frac{n - q - l}{l (n - q - 1)} \cT$ approximately follows $F_{l, n - q - l;\lambda(n)}$, a non central $F$-distribution with degrees of freedom $(l, n - q - l)$ and non-centrality parameter $\lambda(n)$ \citep{pan2002small}. In order to have at least $1-b$ power under $\cH_1^*(\beta)$, the sample size $n$ must satisfy 
\begin{align*}
    P\bigg\{\frac{n-q-l}{l(n-q-1)} \cT > F^{-1}_{l, n-q-l}(1 - \eta)\bigg\} \geq 1-b, \text{ where }
    \frac{n-q-l}{l(n - q - 1)} \cT \sim F_{l, n - q - l;\lambda(n)}.
\end{align*}
Therefore, the required sample size is the smallest integer $n$ such that 
\begin{align}
    1 - F_{l, n-q-l;\lambda(n)}\Big\{ F^{-1}_{l, n-q-l}(1 - \eta)\Big\}\geq 1 - b. \label{eq:ss_formula}
\end{align}

Calculating $n$ from \eqref{eq:ss_formula} requires knowing $\lambda(n)$, but $\lambda(n)$ depends on the data generating distribution $\cP$ through $M$ and $\Sigma$, which are typically unknown during trial planning. To make it feasible to compute $n$ from \eqref{eq:ss_formula}, we make the following working assumptions (WA) about $\cP$. For completeness, the working assumption \eqref{eq:working-assumption-on-MEE} is also included below as (WA-a).
\begin{itemize}
    \item[(WA-a)] (Parametric MEE.) For each $k\in[K]$, there exists $\beta_k \in \RR^p$ such that $\mee_k(t) = f_t^T\beta_k$ for $t \in [T]$, where $f_t \in \RR^p$ is known.
    \item[(WA-b)] (Known success probability under no treatment.) There exists $\alpha^0 \in \RR^q$ such that $E(Y_t \mid I_t = 1, A_t = 0) = e^{g_t^T \alpha^0}$ for $t \in [T]$, where $g_t \in \RR^q$ is known. 
    \item[(WA-c)] (Known availability probability.) Suppose $E(I_t)= \tau(t)$ for $t \in [T]$, where the value of $\tau(t)$ is known for $t \in [T]$.
    \item[(WA-d)] (No serial correlation in the outcome.) Suppose that for every $(t,s)$ pair with $1 \leq s < t \leq T$, $E(Y_t \mid  I_s = 1, I_t = 1, A_t, A_s, Y_{s} )$ is constant with respect to $Y_s$.  
    \item[(WA-e)] (Exogenous availability process.) Suppose $I_t$ is independent of prior treatment or prior outcomes, i.e $I_t \perp \{A_s, Y_s : 1 \leq s < t\}$ for all $t \in [T]$.
\end{itemize}

Under these working assumptions, we derived a tractable sample size formula presented in \cref{alg:ss-calculator}, where the non-centrality parameter $\lambda$ has an explicit form. \Cref{thm:sample-size-formula} establishes the type I error rate control and power guarantee of the sample size formula. The derivation of the tractable sample size formula and the proof of \Cref{thm:sample-size-formula} is in Section~\ref{sec:proof-thm-sample-size-formula} of the Supplementary Material.

\normalem 
\begin{algorithm}[htbp]
    \caption{Sample size calculator for MRT with categorical treatment}
    \label{alg:ss-calculator}
    \SetKwInOut{Input}{Input}
    \SetKwInOut{Output}{Output}
    \Input{$K$: total number of active treatment options (excluding the reference level $0$); \newline
    $T$: total number of decision points per participant; \newline
    $p_t(k) = P(A_t = k \mid I_t = 1)$ for $t \in [T], k \in [K]$: randomization probability; \newline
    $\tau(t) = E(I_t)$ for $t \in [T]$: probability of being available; \newline
    $f_t \in \RR^p$ for $t \in [T]$: the vector in the parametric working model for $\mee_k(t)$; \newline
    $\beta_k \in \RR^p$ for $k\in [K]$: the coefficients for $f_t$ in the target alternative \eqref{eq:target-alternative}; \newline
    $g_t \in \RR^q$ for $t \in [T]$: the vector in the parametric working model for $E(Y_t \mid A_t = 0,I_t = 1)$; \newline
    $\alpha \in \RR^q$: the coefficients for $g_t$ in that working model, i.e., $\spnc(t) = \exp(g_t^T\alpha)$; \newline
    $L \in \RR^{v \times K}$: the linear contrast matrix in defining $\cH_0$ and $\cH_1$ in \eqref{eq:null-and-alternative}; \newline
    $\eta$: desired type I error rate; \newline
    $1-b$: desired power.}
    \revisionadd{$\tL \gets L\otimes \II_p$} \\
    \revisionadd{$l \gets \text{rank}(\tL)$} \\
    $\beta \gets (\beta_1^T, \ldots, \beta_K^T)^T$ \\
    Compute $M$ and $\Sigma$ matrix \\
    $\text{finished} \gets \text{False}$ \\
    $n \gets 9$ \Comment{Some small $n$ to start with} \\
    \While{not finished}{
        $n \gets n+1$ \\
        \revisionadd{$\lambda \gets n (\tL\beta)^T\{\tL M^{-1}\Sigma M^{-1,T}\tL^T\}^{-1}(\tL\beta)$} \\
        \If{ $F_{l, n-q-l;\lambda}\{F_{l, n-q-l}^{-1}(1-\eta)\} \leq b$ }{$\text{finished} \gets \text{True}$ \Comment{Found smallest $n$ to satisfy \eqref{eq:ss_formula}}}
    }
    \Output{$n$}
\end{algorithm}
\ULforem 

\begin{thm}[Type I error rate control and power guarantee under working assumptions.]
    \label{thm:sample-size-formula}
    Suppose the randomization probability at each decision point is either constant or dependent only on $t$. Consider the testing procedure for $\cH_0$ vs. $\cH_1$ in \eqref{eq:null-and-alternative} based on test statistic \eqref{eq:def-test-stat} and rejection region \eqref{eq:rejection-region}.
    \begin{itemize}
        \item[(i)] Suppose (WA-a) holds. Then the testing procedure has type I error rate approximately $\eta$.
        \item[(ii)] Suppose (WA-a)--(WA-e) hold, and suppose that for each $k \in [K]$, $p_t(k) f_t$ lies in the linear span of $g_t$ viewed as vector-valued functions of $t$ (when $p_t(k)$ does not depend on $t$, a simple sufficient condition is that each coordinate of $f_t$ is a coordinate of $g_t$). With $n$ calculated from \cref{alg:ss-calculator}, the testing procedure has power at least $1-b$ approximately under the target alternative $\cH_1^*(\beta)$ in \eqref{eq:target-alternative}, assuming $\tL \beta \neq 0$.
    \end{itemize}
\end{thm}

We discuss the implications of each of these working assumptions.

(WA-a) assumes that the researcher knows the functional form of $\mee_k(t)$ as a function of $t$. This is equivalent to the working assumption \eqref{eq:working-assumption-on-MEE} or the target alternative \eqref{eq:target-alternative}.

(WA-b) assumes that the researcher knows the functional form of the  success probability under no treatment, i.e., the success probability null curve $\spnc(t) := E(Y_t \mid A_t =0, I_t = 1)$ as a function of $t$. Note that $E(Y_t \mid A_t =0, I_t = 1)$ averages over the past treatments ($A_1, A_2, \ldots, A_{t-1}$) and past outcomes ($Y_1, Y_2, \ldots, Y_{t-1}$), and thus $\spnc(t)$ depends on the magnitude of delayed effects and serial correlation in the outcome.

(WA-c) assumes that for each decision point, the researcher knows the probability of a participant being available. If (WA-e) is violated, $\tau(t)$ would then reflect the dependence of $I_t$ on previous treatments ($A_s$) and outcomes ($Y_s$)  for all $1 \leq s < t$.

(WA-d) assumes that the outcome is independent of previous outcomes when conditioned on previous treatments. This assumption is introduced to enable an analytic sample size formula. Although this assumption is generally unrealistic in most mobile health applications where outcomes for the same participant measured close in time are often correlated, our simulation studies will show that the sample size formula performs well even when (WA-d) is violated.

(WA-e) assumes that the availability indicator $I_t$ is independent of prior treatments and outcomes. The plausibility of this assumption depends on the study context. For example, (WA-e) is reasonable if a decision point is unavailable due to a technical glitch unrelated to the participant’s behavior or treatment. However, (WA-e) is violated if availability is influenced by participant burden, i.e., if a decision point is more likely to become unavailable under treatment delivery in recent decision points. Our simulation studies will demonstrate that the sample size formula performs well even when (WA-e) is violated. Note that (WA-e) is satisfied in MRTs without availability constraints because $I_t$ will always be 1.

\subsection{Performance of the Sample Size Formula under Working Assumption Violations and Practical Guidelines}
\label{subsec:practical-guideline}

We performed extensive simulation studies to assess the performance of the sample size formula under ideal conditions when all working assumptions hold, and under scenarios where some working assumptions are violated. In this section, we summarize the performance of the formula in the simulation studies and provide practical guidelines for using the sample size formula. Detailed simulation results can be found in Section~\ref{sec:simulation-sample-size-formula} of the Supplementary Material.

First, we provide definitions necessary for summarizing the simulation findings. For given $k\in[K]$, define the average treatment effect of treatment level $k$ as 

\begin{align}
    \ate_k := \frac{\sum_{t = 1}^T E(Y_t \mid I_t = 1, A_t=k) E(I_t)}{ \sum_{t = 1}^T E(Y_t \mid I_t = 1, A_t=0) E(I_t)}, \label{eq:ate_k}
\end{align}
the ratio between two $\ate$ of treatment level $j$ and $k$ as
\begin{align}
     \rate_{jk} := \frac{\ate_k}{\ate_j}, \label{eq:ratio-ate}
\end{align}
the average success probability under the null (ASPN) as 
\begin{align}
    \aspn := \frac{\sum_{t = 1}^T  \spnc(t) E(I_t)}{ \sum_{t = 1}^T E(I_t)} \equiv \frac{\sum_{t = 1}^T  E(Y_t \mid A_t = 0, I_t = 1) E(I_t)}{ \sum_{t = 1}^T E(I_t)}, \label{eq:aspn}
\end{align}
and the average availability (AA) as 
\begin{align}
    \aaa := \frac{1}{T}\sum_{t = 1}^T \tau(t) \equiv \frac{1}{T}\sum_{t = 1}^T E(I_t). \label{eq:AA}
\end{align}

\revisiondel{$\ate_k$ is the MEE of treatment level $k$ (compared to treatment level 0) averaged over time and weighted by availability.} \revisionadd{$\ate_k$ is the ratio of the availability-weighted mean outcome under treatment level $k$ to the corresponding mean outcome under treatment level 0.} $\rate$ denotes the ratio between average treatment effect ($\ate$s) between treatment levels $j$ and $k$, capturing on average how different the effects of the two active treatments are. $\aspn$ is $\spnc(t) := E(Y_t \mid A_t =0, I_t = 1)$ averaged over time and weighted by availability. $\aaa$ is $E(I_t)$ averaged over time. $\ate_k$, $\aspn$, and $\aaa$ summarize the magnitude of $\mee_k(t)$, $\spnc(t)$, and $\tau(t)$, respectively. We will use ``pattern'' to refer to how $\mee_k(t)$, $\spnc(t)$, and $\tau(t)$ vary over $t$ independently of their magnitude. As we will see, the magnitude and the pattern of $\mee_k(t)$, $\rate_{jk}$, $\spnc(t)$, and $\tau(t)$ impact the performance of the sample size formula in different ways.

We distinguish between two versions of each quantity: one corresponding to the true data-generating distribution (denoted by a superscript $*$) and the other corresponding to the input used in the sample size formula (denoted by a superscript ``w'' for ``working''). For instance, $\ate_k^*$ represents \eqref{eq:ate_k} with expectations calculated according to the true data-generating distribution $\cP$ (which is unknown outside of simulations), and $\ate_k^\w$ represents \eqref{eq:ate_k} with expectations calculated using the input to the sample size formula, assuming all working assumptions hold. Using this notation, (WA-a) corresponds to $\mee_k^\w(t)$ and $\mee_k^*(t)$ being of the same functional form, (WA-b) corresponds to $\spnc^\w(t)$ and $\spnc^*(t)$ being of the same functional form, and (WA-c) corresponds to $\tau^\w(t) = \tau^*(t)$.

Through the simulation studies presented in Section~\ref{sec:simulation-sample-size-formula} of the Supplementary Material, we found that the type I error rate is always at the desired level under arbitrary working assumption violations. The sample size formula yields adequately powered MRTs when all the working assumptions hold. The power performance under various working assumption violations is listed in \Cref{table:sample-size-performance}.

We discuss practical guidelines for setting the inputs to the sample size calculator based on the simulation results. In order to ensure adequate power, it is recommended that the researcher (i) correctly specify the magnitude of MEE (i.e., $\ate_k^\w = \ate_k^*$) and the average availability (i.e., $\aaa^\w = \aaa^*$); (ii) use a constant $\mee_k(t)$ input (i.e., set $f_t = 1$); (iii) use a constant $\spnc(t)$ (e.g., set $g_t = 1$). It helps to be conservative about $\ate$, $\aspn$, and $\aaa$, i.e., use the lower end of a range of conjectured values, if such values are either based on prior studies or from domain knowledge. Once these are satisfied, violating the working assumptions in other ways does not hurt the power. We further summarize the practical guidelines in \Cref{box:practical-guideline} for ease of reference.


\begin{table}[htbp]
    \caption{Sample size formula performance when working assumptions (WA) are violated. $\downarrow$ under-powered; $\uparrow$ over-powered; $\rightarrow$ adequately powered; $\downarrow^*$ under-powered for some generative models; $\uparrow^*$ over-powered for some generative models.}
    \label{table:sample-size-performance}
    \centering
    \begin{tabular}{@{}lllc@{}}
    \toprule
    WA violated                 & \multicolumn{2}{l}{Detail about the violation}                                & Power             \\
    \midrule
    \multirow{4}{*}{(WA-a)}     & \multirow{2}{*}{$\mee_k(t)$ magnitude incorrect}   & $\rate^\w > \rate^*$      & $\downarrow$      \\
                                &                                                   & $\rate^\w < \rate^*$      & $\uparrow$        \\
    \cmidrule(l){2-4}
                                & \multirow{2}{*}{$\mee_k(t)$ pattern incorrect}     & $\mee_k^\w(t)$ constant   & $\rightarrow$     \\
                                &                                                   & $\mee_k^\w(t)$ linear     & $\downarrow^*$    \\
    \midrule
    \multirow{4}{*}{(WA-b)}     & \multirow{2}{*}{$\spnc(t)$ magnitude incorrect}    & $\aspn^\w > \aspn^*$      & $\downarrow$     \\
                                &                                                   & $\aspn^\w < \aspn^*$      & $\uparrow$       \\
    \cmidrule(l){2-4}
                                & \multirow{2}{*}{$\spnc(t)$ pattern incorrect}      & $\spnc^\w(t)$ constant    & $\rightarrow$     \\
                                &                                                   & $\spnc^\w(t)$ non-constant & $\downarrow$  \\
    \midrule
    \multirow{4}{*}{(WA-c)}     & \multirow{2}{*}{$\tau(t)$ magnitude incorrect}     & $\aaa^\w > \aaa^*$        & $\downarrow$      \\
                                &                                                   & $\aaa^\w < \aaa^*$        & $\uparrow$        \\
    \cmidrule(l){2-4}
                                & \multirow{2}{*}{$\tau(t)$ pattern incorrect}       & $\tau^\w(t)$ constant     & $\rightarrow$     \\
                                &                                                   & $\tau^\w(t)$ non-constant  & $\rightarrow$   \\
    \midrule
    \multirow{3}{*}{(WA-d)}     & \multicolumn{2}{l}{serial correlation but is accounted for ($\aspn^\w = \aspn^*$)} & $\rightarrow$     \\
    \cmidrule(l){2-4}
                                & \multirow{2}{*}{($\aspn^\w \neq \aspn^*$)} & positive serial correlation & $\uparrow$    \\
                                &                                                   & negative serial correlation & $\downarrow$    \\
    \midrule
    \multirow{7}{*}{(WA-e)}     & \multirow{3}{*}{correct $\tau(t)$ magnitude}       & $Y_{t-1}$ decreases $\tau(t)$ & $\downarrow$  \\
                                &                                                   & $Y_{t-1}$ increases $\tau(t)$ & $\uparrow$    \\
                                &                                                   & $A_{t-1}$ affects $\tau(t)$   & $\rightarrow$ \\
    \cmidrule(l){2-4}
                                & \multirow{4}{*}{$\tau(t)$ magnitude incorrect}     & $Y_{t-1}$ decreases $\tau(t)$ & $\downarrow$  \\
                                &                                                   & $Y_{t-1}$ increases $\tau(t)$ & $\uparrow$    \\
                                &                                                   & $A_{t-1}$ decreases $\tau(t)$ & $\downarrow$ \\
                                &                                                   & $A_{t-1}$ increases $\tau(t)$ & $\uparrow^*$  \\
    \bottomrule
    \end{tabular}
\end{table}

\begin{guidelinebox}[htbp]
    \caption{Practical guidelines for using the sample size formula.}
    \label{box:practical-guideline}
    \begin{mdframed}[linewidth=1pt, roundcorner=10pt, backgroundcolor=gray!10]
    \spacingset{1}
    \noindent \textbf{For $1-b$, $\eta$, $T$, $p_t(k)$, $L$:}
    \begin{itemize}[leftmargin=2em]
      \item Specify according to the MRT design and the scientific question of interest.
    \end{itemize}

    \noindent \textbf{For $f_t$ and $\beta_k$:}
    \begin{itemize}[leftmargin=2em]
      \item Use a constant pattern ($f_t = 1$) unless strong prior knowledge about a specific form;
      \item \revisiondel{Set $\beta_k = \ate_k$ for a conjectured $\ate_k$ value;} \revisionadd{For a conjectured $\ate_k$ value, set $\beta_k = \log(\ate_k)$;}
      \item If having a range of conjectured $\ate_k$ values, use the lower bound to be conservative.
    \end{itemize}

    \noindent \textbf{For $g_t$:} (recall \revisiondel{$\spnc(t) = g_t^T\alpha$}\revisionadd{$\spnc(t) = \exp(g_t^T\alpha)$})
    \begin{itemize}[leftmargin=2em]
      \item Use a constant pattern ($g_t = 1$)
      \item \revisiondel{Set $\alpha = \aspn$ for a conjectured $\aspn$ value} \revisionadd{For a conjectured $\aspn$ value, set $\alpha = \log(\aspn)$}
      \item If having a range of conjectured $\aspn$ values, use the lower bound to be conservative.
    \end{itemize}

    \noindent \textbf{In addition, for MRT with availability considerations:}
    \begin{itemize}[leftmargin=2em]
      \item Use a constant $\tau(t)$ pattern unless strong prior knowledge about a specific form;
      \item Set $\tau(t) = \aaa$ for a conjectured $\aaa$ value;
      \item If having a range of conjectured $\aaa$ values, use the lower bound to be conservative.
    \end{itemize}

    \noindent \textbf{In addition, for MRT with imbalanced treatment assignment:} (i.e., if $p_t(0),p_t(1),\ldots,p_t(K)$ are not all equal)
    \begin{itemize}[leftmargin=2em]
      \item Increase the sample size by 10--20\% if the variability in the outcome may be affected by prior treatment.
    \end{itemize}
    \end{mdframed}
\end{guidelinebox}

\spacingset{1.9}


\section{Data Example}
\label{sec:application}

We illustrate the use of CEE estimator and the calculation of sample size using the Drink Less MRT, a 2021 study aimed at optimizing the Drink Less smartphone app to help people reduce alcohol consumption \citep{bell2023notifications}. In this study, 349 participants were randomized every day at 8 pm for 30 days, resulting in 10,470 decision points. At each decision point, participants were randomized among three options:  no notifications (with probability 0.4), receiving a standard message (with probability 0.3), or receiving notification randomly selected from a bank of new messages (with probability 0.3). All participants were always considered available at all decision points throughout the study, i.e., $I_t = 1$ for all $t$. The primary outcome of interest is whether the user opened the app (yes or no) in the hour following the 8 pm decision point. We have $\Delta = 1$ (recall that $\Delta$ is the time window length over which the proximal outcome is defined). A self-reported baseline of Alcohol Use Disorders Identification Test (AUDIT) is taken for each participant stratified into three different risk zones: hazardous (8-15), harmful (16-19), and at risk for alcohol dependence (20-40).

\subsection{Illustration of the CEE estimator}
\label{subsec:application-estimator}

We applied the proposed estimator to separately estimate (a) the marginal CEE (by setting $S_t = \emptyset$ and $f_t(S_t) = 1$ in \eqref{eq:CEE-parametric-model}), (b) the CEE moderated by the decision point index (by setting $S_t = \emptyset$ and $f_t(S_t) = (1,t)^T$), and (c) the CEE moderated by the stratified baseline AUDIT score (by setting \revisiondel{$S_t = \texttt{AUDIT}_t$}\revisionadd{$S_t = \texttt{AUDIT}$} and $f_t(S_t) = (1,\texttt{AUDIT\_DEPENDENT}, \texttt{AUDIT\_HARMFUL})^T$). Here,
$\texttt{AUDIT\_DEPENDENT}$ and $\texttt{AUDIT\_HARMFUL}$ are binary indicators taking the value 1 if the participant's AUDIT score 
falls within the at risk for alcohol dependence range (score within 20-40) or in the harmful range (score within 16-19), respectively.

\begin{align}
    & \text{(a) } \cee_{ta}(\emptyset) = \beta_a^{\text{marginal}} \quad \text{ for } a = 1,2; \nonumber\\
    & \text{(b) } \cee_{ta}(\emptyset) = \beta_a^{\text{t-int}} + \beta_a^{\text{t-slp}}\cdot t \quad \text{ for } a = 1,2; \label{eq:application-analysis}\\
    & \text{(c) } \cee_{ta}(\texttt{AUDIT}) = \beta_a^{\text{int}} + \beta_a^{\text{dependent}} \cdot \texttt{AUDIT\_DEPENDENT} \nonumber \\
    & \qquad \qquad + \beta_a^{\text{harmful}} \cdot \texttt{AUDIT\_HARMFUL}  \quad \text{ for } a = 1,2.  \nonumber
\end{align}

Here, ``int'' stands for intercept and ``slp'' stands for slope. $t$ starts at 0 instead of 1 for the interpretability of $\beta_a^{\text{t-int}}$.  For all analyses, we included the following control variables in $g_t(H_t)$: the decision point index $t$, age, gender, employment type, AUDIT score, and the lagged proximal outcome, that is, whether the app was opened following the preceding decision point. In analysis (c), the AUDIT score enters through the same category indicators that define the moderator. Because $\Delta = 1$ and the randomization probabilities are constant, we set $\tp_t(k|S_t) = p_t(k)$ for each $k$, so that $J_t \equiv 1$ and no weighting is needed.
\Cref{tab:application-Drink-Less-analysis} presents the estimated parameters, standard errors, 95\% confidence intervals (CIs), and $p$-values.

For analysis (a), on average, delivering notifications from the bank of newly constructed messages \revisiondel{increased the probability of opening the Drink Less app within the hour by 3.27 times}\revisionadd{made the probability of opening the Drink Less app within the hour 3.27 times as large} ($= e^{1.18}$, 95\% CI $[2.66, 4.02]$, $p < 0.001$), compared with providing no suggestions. Similarly, users who received standard notifications were \revisiondel{3.62 times more likely}\revisionadd{3.62 times as likely} to open the app within the hour ($= e^{1.29}$, 95\% CI $[2.93, 4.46]$, $p < 0.001$) compared with those who received no suggestions. We found no evidence of a marginal difference between delivering notifications from standard messages and from the bank of new messages ($p = 0.139$). As a benchmark, a log-link GEE fit of the same model with a working independence correlation structure gives point estimates that agree with these to within 0.001 on the log relative risk scale. Section~\ref{sec:simulation-estimator} of the Supplementary Material reports simulations in which GEE is biased for the CEE while the proposed estimator is not, and scenarios in which the two approaches agree.

For analysis (b), at the beginning of the study, delivering notifications from the bank of new messages \revisiondel{increased the probability of opening the Drink Less app by 3.38 times}\revisionadd{made the probability of opening the Drink Less app 3.38 times as large} ($= e^{1.22}$, 95\% CI $[2.44, 4.68]$, $p < 0.001$), compared with users who received no suggestions. Delivering notifications from standard messages initially \revisiondel{increased the probability by 4.02 times}\revisionadd{made the probability 4.02 times as large} ($= e^{1.39}$, 95\% CI $[2.88, 5.61]$, $p < 0.001$). We found no evidence of a difference in initial effect between the two notification types ($p = 0.109$). Furthermore, we found no evidence of a time trend in either effect ($p = 0.804$ for bank-of-new-message notifications and $p = 0.409$ for standard-message notifications): that is, no evidence that the effect of either notification type at a given decision point, marginalized over the treatment history up to that point, changed over the course of the study. This is distinct from the effect of remaining on a single notification type for the full trial duration, an estimand that the CEE framework does not target.

For analysis (c), we examined the effects across different AUDIT score categories. Among users in the Hazardous category, receiving notifications from the bank of new messages \revisiondel{increased the probability of opening the app by 3.27 times}\revisionadd{made the probability of opening the app 3.27 times as large} ($= e^{1.18}$, 95\% CI $[2.35, 4.54]$, $p < 0.001$), whereas standard notifications \revisiondel{increased the probability by 4.03 times}\revisionadd{made the probability 4.03 times as large} ($= e^{1.39}$, 95\% CI $[2.96, 5.50]$, $p < 0.001$). This within-category difference favoring standard messages is statistically significant: standard-message notifications \revisiondel{increased the probability of opening the app by 1.23 times}\revisionadd{made the probability of opening the app 1.23 times as large} ($= e^{0.210}$, 95\% CI $[1.00, 1.52]$, $p = 0.048$) relative to bank-of-new-message notifications. For users in the Harmful category, notifications from the bank of new messages \revisiondel{increased the probability by 3.40 times}\revisionadd{made the probability 3.40 times as large} ($= e^{1.22}$, 95\% CI $[2.25, 5.14]$, $p < 0.001$), which was greater than the 2.94-fold increase ($= e^{1.08}$, 95\% CI $[1.89, 4.57]$, $p < 0.001$) from standard notifications. For users in the \revisiondel{At-risk of alcohol dependent category}\revisionadd{at risk for alcohol dependence category}, standard notifications \revisiondel{increased the probability by 3.69 times}\revisionadd{made the probability 3.69 times as large} ($= e^{1.30}$, 95\% CI $[2.50, 5.43]$, $p < 0.001$), whereas notifications from the bank of new messages \revisiondel{increased it by 3.19 times}\revisionadd{made it 3.19 times as large} ($= e^{1.16}$, 95\% CI $[2.19, 4.63]$, $p < 0.001$). We found no evidence of a difference between the effects of standard-message and bank-of-new-message notifications for individuals with AUDIT scores in the at risk for alcohol dependence or harmful categories ($p = 0.247$ and $p = 0.268$, respectively).

\begin{table}[htbp]
    \caption{Drink Less analysis result for \Cref{subsec:application-estimator}. The estimates correspond to parameters defined in the three models in \eqref{eq:application-analysis}. The standard errors are the small-sample adjusted standard errors of \citet{mancl2001covariance}, and the confidence intervals and $p$-values are the corresponding $t$-based quantities. \label{tab:application-Drink-Less-analysis}}
    \begin{tabular}[t]{llrrcr}
        \toprule
        & Parameter & Estimate & SE & 95\% CI  & $p$-value\\
        \midrule
        \multirow{3}{*}{\makecell{\textbf{Model (a):} \\ marginal CEE}} & $\beta_1^{\text{marginal}}$    & 1.184 & 0.105 & (\phantom{$-$}0.977, \phantom{$-$}1.392) &$<0.001$\\
        & $\beta_2^{\text{marginal}}$     & 1.286 & 0.107 & (\phantom{$-$}1.076, \phantom{$-$}1.495) &$<0.001$\\
        & $\beta_2^{\text{marginal}} - \beta_1^{\text{marginal}}$     & 0.101 & 0.068 & ($-$0.033, \phantom{$-$}0.236) & 0.139\\
        \midrule
        \multirow{6}{*}{\makecell{\textbf{Model (b):} \\ CEE moderated \\ by the decision \\ point index}} & $\beta_1^{\text{t-int}}$   & 1.217 & 0.166 &(\phantom{$-$}0.892, \phantom{$-$}1.543) & $<0.001$\\
        & $\beta_1^{\text{t-slp}}$   & $-$0.003 & 0.011 &($-$0.024, \phantom{$-$}0.019) &0.804\\
        & $\beta_2^{\text{t-int}}$   & 1.392 & 0.170 &(\phantom{$-$}1.058, \phantom{$-$}1.725) &$<0.001$\\
        & $\beta_2^{\text{t-slp}}$   & $-$0.009 & 0.011 &($-$0.030, \phantom{$-$}0.012) & 0.409\\
        & $\beta_2^{\text{t-int}} - \beta_1^{\text{t-int}}$ &0.174 &0.108 &($-$0.039, \phantom{$-$}0.387) & 0.109\\
        & $\beta_2^{\text{t-slp}} - \beta_1^{\text{t-slp}}$ & $-$0.006 & 0.008 &($-$0.021, \phantom{$-$}$0.009$) & 0.412 \\
        \midrule
        \multirow{6}{*}{\makecell{\textbf{Model (c):} \\ CEE moderated \\ by AUDIT score}} & $\beta_1^{\text{int}}$                       & 1.184 & 0.167 &(\phantom{$-$}0.856, \phantom{$-$}1.513) & $<0.001$\\
        & $\beta_1^{\text{int}} + \beta_1^{\text{harmful}}$             & 1.225 & 0.210 &(\phantom{$-$}0.812, \phantom{$-$}1.637) & $<0.001$\\
        & $\beta_1^{\text{int}} + \beta_1^{\text{dependent}}$             &1.159 & 0.190 &(\phantom{$-$}0.785, \phantom{$-$}1.532) & $<0.001$\\
        & $\beta_2^{\text{int}}$                       & 1.395 & 0.158 &(\phantom{$-$}1.084, \phantom{$-$}1.706) &$<0.001$\\
        & $\beta_2^{\text{int}} + \beta_2^{\text{harmful}}$             & 1.078 & 0.224 &(\phantom{$-$}0.637, \phantom{$-$}1.520) &$<0.001$\\
        & $\beta_2^{\text{int}} + \beta_2^{\text{dependent}}$             & 1.305 & 0.197 &(\phantom{$-$}0.918, \phantom{$-$}1.692) &$<0.001$\\
        & $\beta_2^{\text{int}} - \beta_1^{\text{int}}$ & 0.210 & 0.106 &(\phantom{$-$}0.002, \phantom{$-$}0.419) & 0.048\\
        &\makecell{$\{(\beta_2^{\text{int}} + \beta_2^{\text{harmful}})$ \\ $-(\beta_1^{\text{int}} + \beta_1^{\text{harmful}})\}$} & $-$0.147 & 0.132 &($-$0.407, \phantom{$-$}0.113) & 0.268 \\
        &\makecell{$\{(\beta_2^{\text{int}} + \beta_2^{\text{dependent}})$ \\ $-(\beta_1^{\text{int}} + \beta_1^{\text{dependent}})\}$} & 0.146 & 0.126 &($-$0.102, \phantom{$-$}0.395) & 0.247 \\
        \bottomrule
    \end{tabular}
\end{table}

\subsection{Illustration of the Sample Size Formula}
\label{subsec:application-sample-size}

We illustrate the use of the proposed sample size formula by determining the required sample size for a hypothetical micro-randomized trial (MRT). Suppose the hypothetical MRT resembles the Drink Less study in both intervention and trial design. The MRT includes two active treatment levels and a no-treatment level, with constant randomization probabilities $p_{0t} = 0.4$ and $p_{1t} = p_{2t} = 0.3$ for all $t \in [T]$. Each participant is enrolled for $T = 30$ decision points. For this illustration, we initially assume that participants are always available for randomization, i.e., $\tau(t) = 1$; this assumption will be relaxed later. In practice, these design parameters should be determined in consultation with domain experts.

We aim to detect a statistically significant difference between the two active treatment levels. According to the MRT design, we set $K = 2$, $T = 30$, $p_{0t} = 0.4$, $p_{1t} = p_{2t} = 0.3$, and $\tau(t) = 1$. We specify $L = (1, -1)$, $\eta = 0.05$, and $1 - b = 0.8$. \Cref{fig:DL1} shows the contour plot of the required sample size $n$ as a function of $\text{ATE}_1$ and $\text{ATE}_2$ (with $g_t = 1$, $\aspn = 0.036$, and all other parameters fixed). The sample size depends on both $\text{ATE}_1$ and $\text{ATE}_2$, rather than solely on their difference, as the slopes of the contour lines are different than 1. The marked point reads the contour at $\ate_1 = 3.27$ and $\ate_2 = 3.62$, the average treatment effects estimated in \Cref{subsec:application-estimator}, and gives $n = 1197$; this is the Drink Less design illustration developed at the end of this subsection. The required sample size is also sensitive to $\aspn$ (\Cref{fig:DL4}).

Next, we examine how $n$ varies under different patterns of $\mee(t)$. We use $\theta_{f1}$ and $\theta_{f2}$ to parameterize the relationship between two linear MEEs. Specifically, $\theta_{f1}$ parameterizes the slope of $\mee_1(t)$, while $\theta_{f2}$ parameterizes the difference between the slopes of $\mee_1(t)$ and $\mee_2(t)$. Two scenarios are considered in this setting: the first involves parallel linear MEEs with different slopes, parameterized by $\theta_{f1}$ and setting $\theta_{f2} = 0$; the second involves non-parallel MEEs, where $\theta_{f2} \neq 0$. Detailed parameterizations of these cases are provided in \Cref{box:detail-gm-0-part2} of the Supplementary Material. \Cref{fig:DL2} and \Cref{fig:DL3} illustrate how these parameterizations influence the required sample size $n$.

In the first scenario with parallel MEEs, specifying a linear MEE (i.e., setting $f_t = (1, t)^T$ rather than $f_t = 1$) leads to a larger required sample size than the constant-MEE setting: over the entire range of $\theta_{f1}$, the curve in \Cref{fig:DL2} lies above the dashed horizontal line. In the second scenario with non-parallel MEEs, the required sample size is largest when the two MEEs are close to parallel. Each of the three curves in \Cref{fig:DL3} peaks near $\theta_{f2} = 0$ and decreases \revisiondel{monotonically }as $|\theta_{f2}|$ grows. Because $\ate_1$ and $\ate_2$ are held fixed across the panel, $\theta_{f2}$ changes only how the difference between the two treatment effects is distributed over the decision points, and a difference that varies over time is detected with fewer participants than one that is constant. The required sample size nonetheless stays above the constant-MEE value over most of the range, falling below it only once $|\theta_{f2}|$ exceeds about $0.8$.

We further investigate how $\text{SPNC}(t)$ and the magnitude of $\aspn$ influence the required sample size $n$. Under a constant $\text{SPNC}(t)$, the required sample size decreases monotonically as the magnitude of $\aspn$ increases (\Cref{fig:DL4}). When $\text{SPNC}(t)$ follows a linear or quadratic trajectory rather than remaining constant, the required sample size is further reduced (\Cref{fig:DL5}). These results indicate that assuming a constant $\text{SPNC}(t)$ leads to a more conservative (i.e., larger) estimate of the required sample size.

Finally, we investigate how the expected availability trajectory $\tau(t)$ influences the required sample size. Under constant $\tau(t)$ (\Cref{fig:DL6}), the required $n$ decreases as the average availability (AA) increases as expected. When assuming linear or periodic availability patterns, the required $n$ remains nearly unchanged (\Cref{fig:DL7}), indicating robustness of the sample size to the specific temporal pattern of $\tau(t)$.

To provide a concrete interpretation of the proposed sample size calculator, we consider its application to the hypothetical Drink Less MRT discussed in the previous section, using the parameter estimates from \Cref{tab:application-Drink-Less-analysis}. According to the results from Model (b), we found no evidence that either treatment effect varies over time. Therefore, we assume a constant MEE pattern. The average treatment effects (ATEs) for the two interventions are obtained by exponentiating the corresponding coefficient $\beta$ from Model (a). In our application, we set $\ate_1 = 3.27$ $(= e^{1.18})$ and $\ate_2 = 3.62$ $(= e^{1.29})$. Following the discussion in the previous paragraph, we assume a constant success probability under the null, $\spnc(t)$, which yields a more conservative sample size requirement. We set the average success probability under the null to $\aspn = 0.036$, the observed proportion of decision points at which the app was opened when no notification was delivered. Because a smaller $\aspn$ requires a larger sample size, this is also the more conservative of the plausible choices, in line with the guidance in \Cref{box:practical-guideline}. By design, each participant is assumed to be available at every decision point; therefore, we assume a constant availability pattern with magnitude 1 ($\aaa = 1$) and constant $\tau(t)$. Substituting these quantities into the proposed sample size formula, we conclude that using our sample size calculator, the Drink Less study would require a total of 1197 participants to achieve the prespecified target power under the assumed constant-MEE specification for detecting the difference between the two treatment effects. This result provides a concrete illustration of how the proposed sample size calculator translates model-based effect size estimates into actionable design recommendations and guides the planning of future MRTs with time-varying treatment effects.

\begin{figure}
    \centering
    \begin{subfigure}[b]{0.3\textwidth}
        \centering
        \caption{}
        \label{fig:DL1}
        \includegraphics[width=\textwidth]{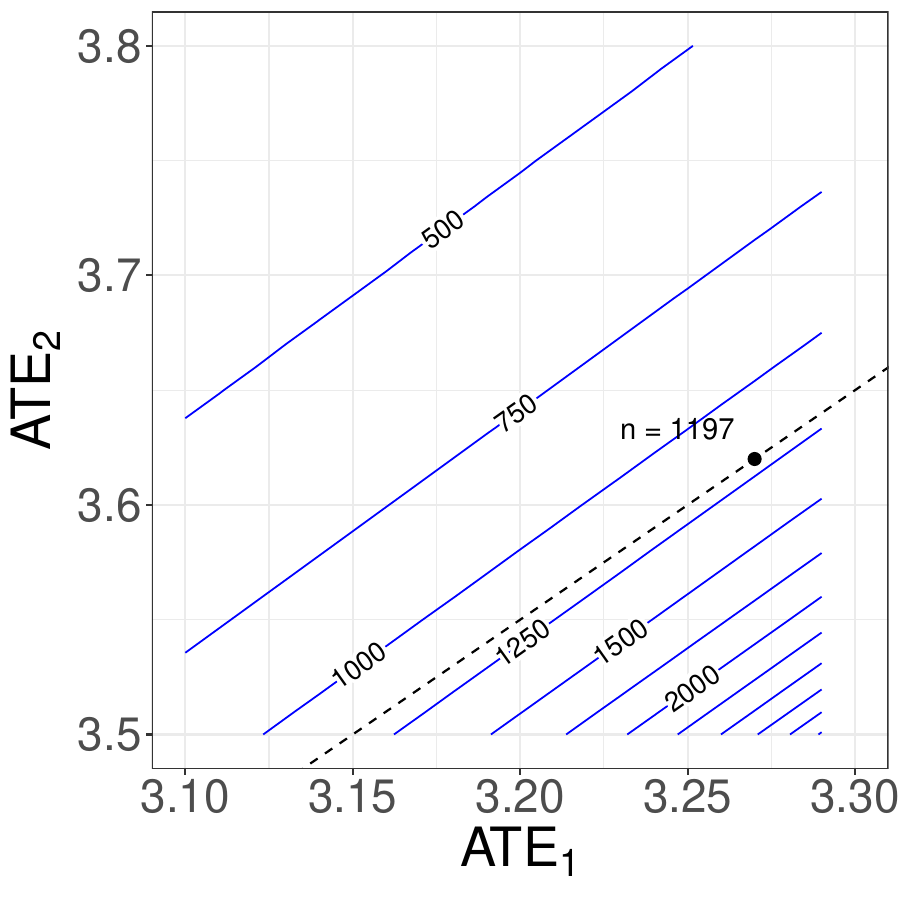}
    \end{subfigure}
    \begin{subfigure}[b]{0.3\textwidth}
        \centering
        \caption{}
        \label{fig:DL2}
        \includegraphics[width=\textwidth]{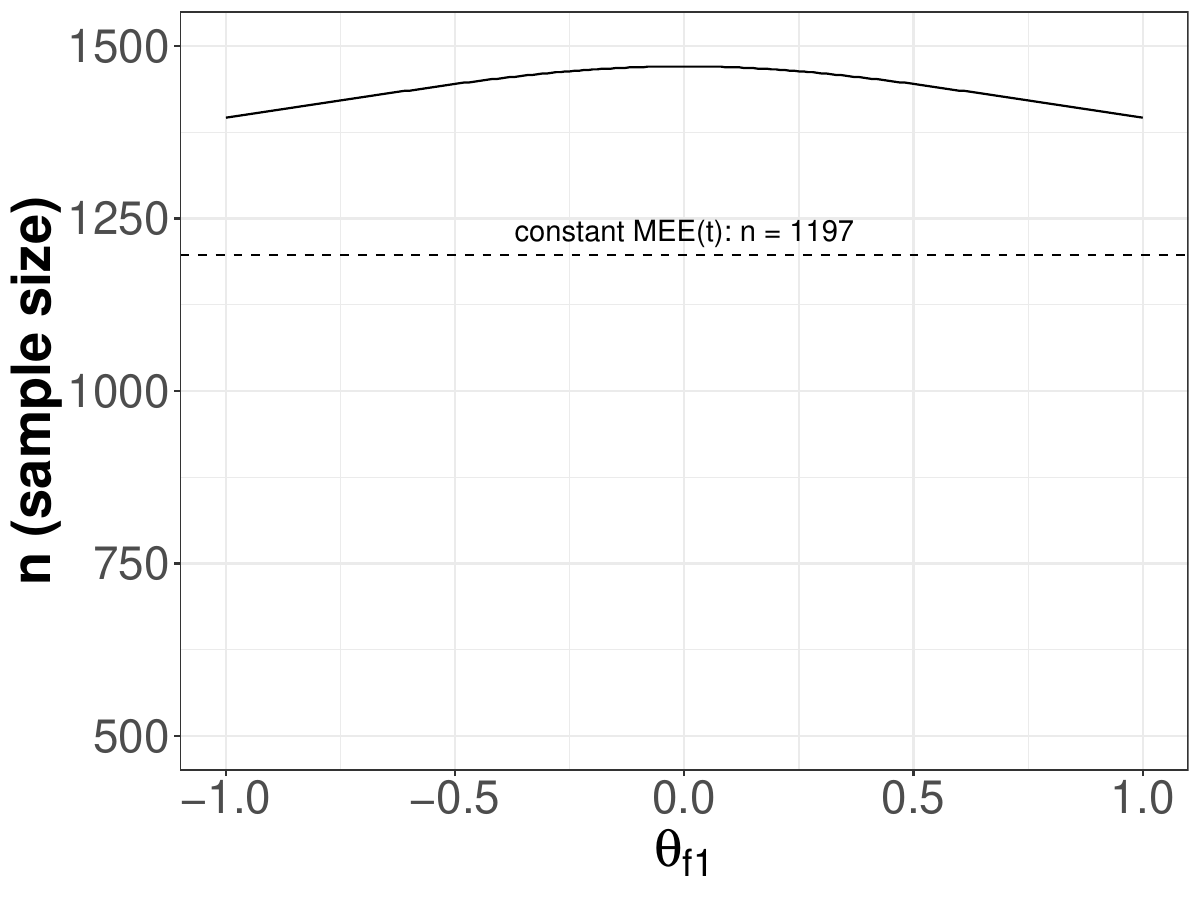}
    \end{subfigure}
    \begin{subfigure}[b]{0.33\textwidth}
        \centering
        \caption{}
        \label{fig:DL3}
        \includegraphics[width=\textwidth]{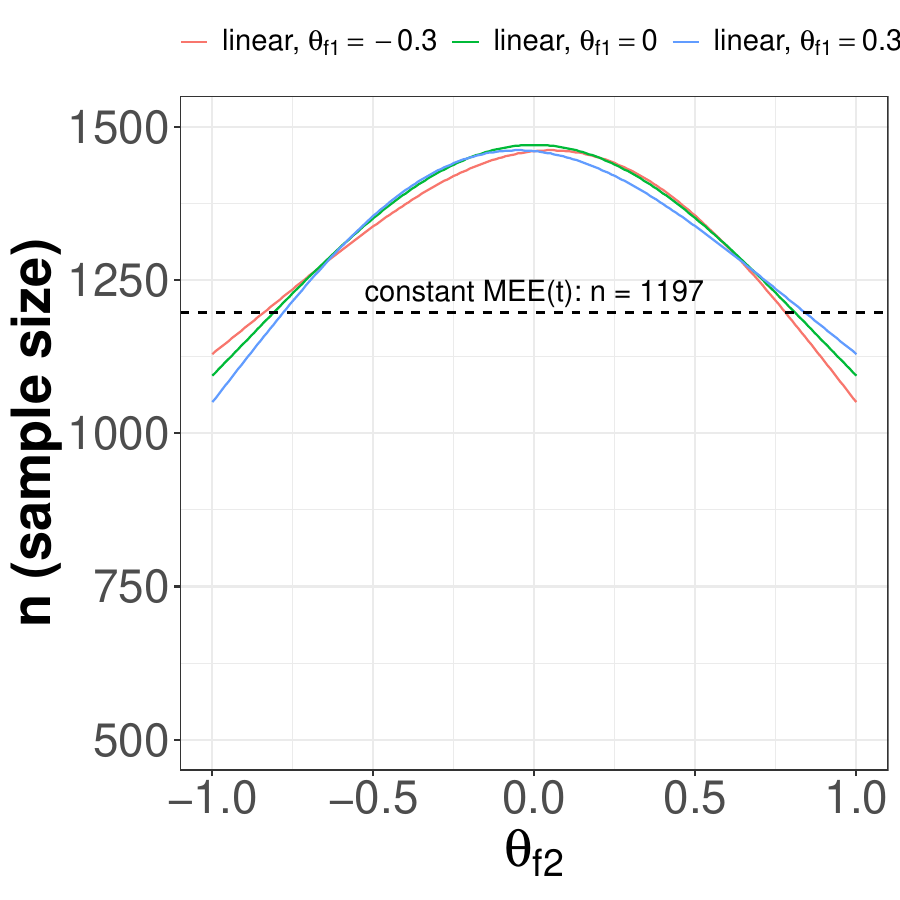}
    \end{subfigure}
    \caption{The dependence of the output sample size $n$ on various inputs of $\mee$ for the data example illustration in \cref{subsec:application-sample-size}. \textbf{Panel (a):} $n$ varies between a pair of $\ate_1$ and $\ate_2$. The point marks the estimated $(\ate_1, \ate_2) = (3.27, 3.62)$ and the dashed line the values with the same difference between the two effects. \textbf{Panel (b):} $n$ depends on the slope of two parallel $\mee$ parameterized by $\theta_{f1}$. \textbf{Panel (c):} $n$ depends on the difference in the slopes of $\mee_1(t)$ and $\mee_2(t)$, parameterized by $\theta_{f2}$}
\end{figure}

\begin{figure}
    \centering
    \begin{subfigure}[b]{0.4\textwidth}
        \centering
        \caption{}
        \label{fig:DL4}
        \includegraphics[width=\textwidth]{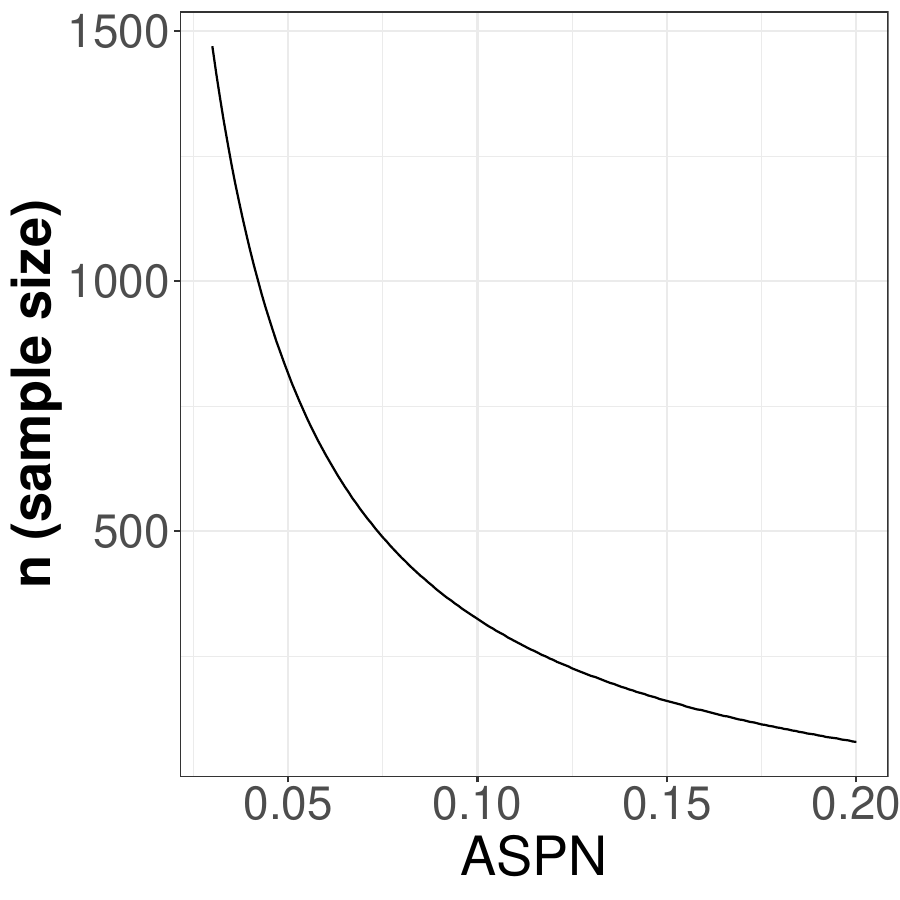}
    \end{subfigure}
    \begin{subfigure}[b]{0.45\textwidth}
        \centering
        \caption{}
        \label{fig:DL5}
        \includegraphics[width=\textwidth]{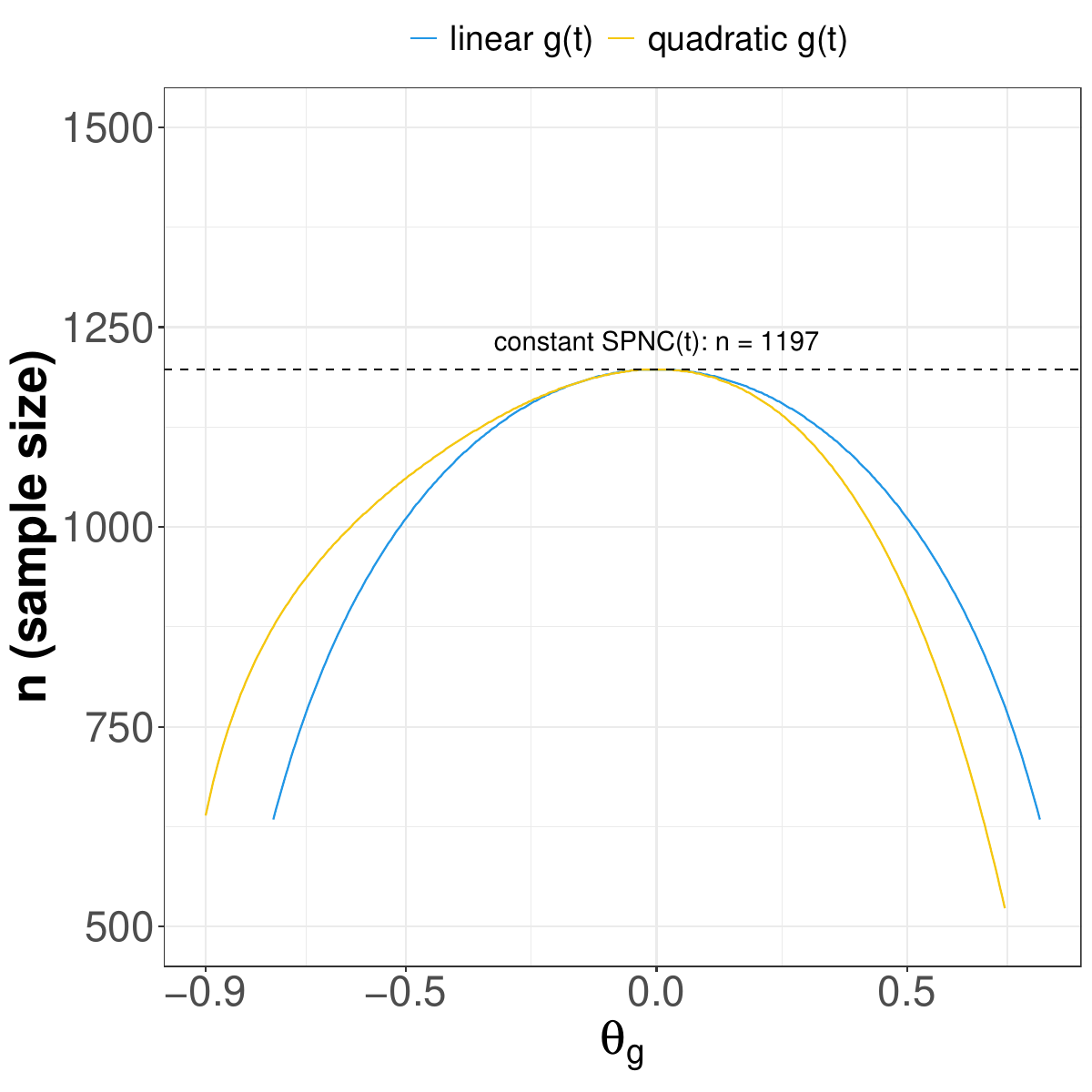}
    \end{subfigure}

    \caption{The dependence of the output sample size $n$ on various inputs of \spnc{} and \aspn{} for the data example illustration in \cref{subsec:application-sample-size}. \textbf{Panel (a):} $n$ varies monotonically with $\aspn$. \textbf{Panel (b):} $n$ depends on the shape of \spnc(t) parameterized by $\theta_g$. Constant $\spnc(t)$ provides the most conservative sample size.}  
\end{figure}

\begin{figure}
    \centering
    \begin{subfigure}[b]{0.35\textwidth}
        \centering
        \caption{}
        \label{fig:DL6}
        \includegraphics[width=\textwidth]{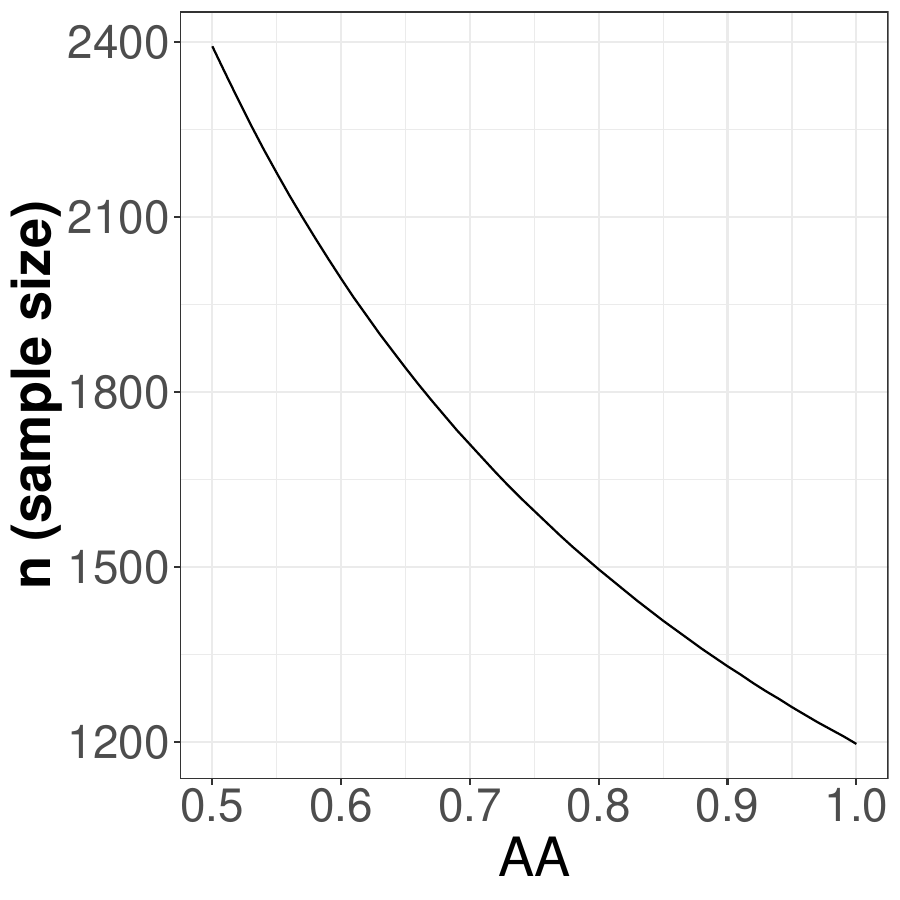}
    \end{subfigure}
    \hfill
    \begin{subfigure}[b]{0.6\textwidth}
        \centering
        \caption{}
        \label{fig:DL7}
        \includegraphics[width=\textwidth]{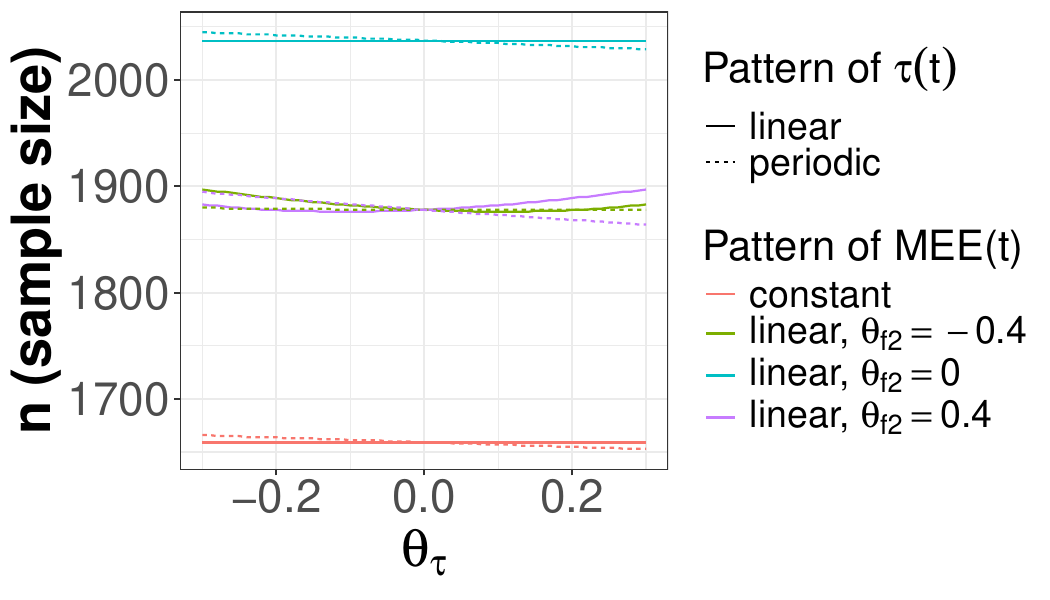}
    \end{subfigure}

    \caption{The dependence of the output sample size $n$ on various inputs of availability for the data example illustration in \cref{subsec:application-sample-size}. \textbf{Panel (a):} $n$ varies monotonically with the average expected availability, $\aaa$. \textbf{Panel (b):} $n$ essentially does not depend on the pattern of the availability over time, $\tau(t)$.}
\end{figure}

\section{Discussion}
\label{sec:discussion}

We considered micro-randomized trials in which, at each decision point, participants are assigned to one of several categorical treatment options, with a binary outcome of interest. For this setting, we generalized the causal excursion effect framework and extended the corresponding estimator of marginal excursion effect for binary outcomes \citep{qian2021estimating} to obtain EMEE-catA. In addition, we developed a sample size formula that supports a broad class of contrasts among categorical treatment levels. Under a set of working assumptions, the proposed formula is shown to control the type I error rate and guarantee power. Using extensive simulation experiments that intentionally violate these working assumptions, we found that the type I error rate is consistently controlled, while power is typically but not uniformly maintained. To support practical implementation, we summarized in \Cref{box:practical-guideline} guidance for selecting inputs to the sample size formula that yield adequate power in most applied settings, even when the working assumptions do not strictly hold. The proposed estimator and sample size methodology were demonstrated using data from the Drink Less micro-randomized trial.

In formulating the sample size expression, we focused on null hypotheses of the form \(\tL \beta = 0\). Although the framework naturally extends to null hypotheses of the form \(\tL \beta = \tc\) for a fixed vector \(\tc\), interpreting such hypotheses can be nontrivial, depending on the specific structure of \(\tL\) and \(\tc\). In particular, certain choices of \(\tc\) do not correspond to a parsimonious null hypothesis on the original scale of the time-varying marginal excursion effects \(\mee_{1:K}(t)\); that is, there may be no vector $c$ such that \(\tL \beta = \tc\) is equivalent to \(L \times \mee_{1:K}(t) = c\). For this reason, we restricted attention in this paper to hypotheses of the form \(\tL \beta = 0\).

Several directions for future work remain. First, incorporating flexible nonparametric and machine learning–based approaches for modeling nuisance functions may improve the efficiency of causal effect estimation. Second, extending the methodology to support sample size calculations for detecting effect moderation represents an important direction for further research. Third, developing sample size tools that accommodate more complex probability assignment policies beyond simple randomization would further enhance the applicability of the proposed framework. Finally, extending the framework to other outcome and treatment types is of interest. For time-to-event proximal outcomes, excursion effects could be defined on a hazard or restricted-mean scale; for ordinal treatments, the contrast matrix $L$ could encode hypotheses that respect the ordering of treatment levels; and for continuous treatments such as dose, the CEE could be formulated as a dose-response curve, with corresponding sample size methods derived under working models for these effects.

\section*{Supplementary Material}
The Supplementary Material includes theoretical proofs of \cref{thm:CAN} and \cref{thm:sample-size-formula}, simulation studies evaluating the consistency and robustness of the proposed estimator, and simulation studies evaluating the performance of the sample size formula under working assumption violations. 

\section*{Conflict of Interest}
None declared.

\section*{Data Availability}
R code implementing the proposed estimator and the sample size calculator, together with the code that reproduces the simulation studies and the Drink Less analysis and figures of \Cref{sec:application}, is available at \url{https://github.com/JeremyJosephLin/causal_excursion_mult_trt_binary_outcome}. The Drink Less data are publicly available at \url{https://osf.io/w3szp}, and the analysis code builds the analysis data set from that file directly.

\section*{Funding}
Not applicable. 

\section*{Author Contribution}
J.L. developed the methodology and theory, conducted the simulation studies and the data analysis, and drafted the manuscript. T.Q. conceived and supervised the project, contributed to the methodology and theory, and revised the manuscript.

\bibliography{references}
\bibliographystyle{plainnat}


\newpage

\appendix

\numberwithin{equation}{section}
\numberwithin{thm}{section}
\numberwithin{defn}{section}
\renewcommand{\thelem}{\thethm}
\renewcommand{\thecor}{\thethm}
\setcounter{table}{0}
\setcounter{figure}{0}
\setcounter{guidelinebox}{0}
\renewcommand{\thetable}{S.\arabic{table}}
\renewcommand{\thefigure}{S.\arabic{figure}}
\renewcommand{\theguidelinebox}{S.\arabic{guidelinebox}}
\def\theHtable{S.\arabic{table}}
\def\theHfigure{S.\arabic{figure}}
\def\theHguidelinebox{S.\arabic{guidelinebox}}
\def\theHthm{\thesection.\arabic{thm}}
\def\theHlem{\thesection.\arabic{thm}}
\def\theHcor{\thesection.\arabic{thm}}
\def\theHdefn{\thesection.\arabic{defn}}

\crefname{appendix}{Section}{Sections}
\Crefname{appendix}{Section}{Sections}
\crefname{subappendix}{Section}{Sections}
\Crefname{subappendix}{Section}{Sections}
\crefname{subsubappendix}{Section}{Sections}
\Crefname{subsubappendix}{Section}{Sections}

\phantomsection
\addcontentsline{toc}{section}{Supplementary Material}

\begin{center}
    {\Large \bf Supplementary Material} \\[0.8em]
    {\large for ``Micro-randomized Trials with Categorical Treatments and Binary
    Proximal Outcome: Causal Effect Estimation and Sample Size Calculation''} \\[0.8em]
    {Jeremy Lin, Tianchen Qian}
\end{center}

\spacingset{1.9}


\section{Technical Details for \texorpdfstring{\Cref{sec:estimator}}{Section 3} on Asymptotic Normality}
\label{sec:proof-thm-CAN}

\subsection{Lemmas}

\begin{lem}
    \label{lem:thm1-proofuse}
    Fix arbitrary $t \in [T]$. Let
    \begin{align*}
        W_t := \prod_{j = t + 1}^{t + \Delta - 1} \frac{\one(A_j = 0)}{p_j(0|H_j)}.
    \end{align*}
    For any $(\alpha,\beta)$, we have
    \begin{align}
        & \sum_{l=0}^K E\bigg[ e^{-\sum_{k = 1}^K C_k(A_t) f_t(S_t)^T\beta_k } \bigg(Y_{t,\Delta} - e^{\sum_{k = 1}^K C_k(A_t) f_t(S_t)^T \beta_k + g_t(H_t)^T \alpha}\bigg) \nonumber \\
        &~~~~~~~ \times J_t C_K(A_t) f_t(S_t) ~\Big|~ H_t, I_t = 1, A_t = l \bigg] ~ p_t(l | H_t) \nonumber \\
        = ~& e^{\sum_{k = 1}^K \tp_t(k|S_t) f_t(S_t)^T \beta_k } ~ \tp_t(K|S_t) f_t(S_t) \nonumber \\
        & \times \bigg[\Big\{1 - \tp_t(K|S_t)\Big\} \Big\{e^{-f_t(S_t)^T \beta_K}E(Y_{t,\Delta} W_t \mid H_t, I_t = 1, A_t = K) - E(Y_{t,\Delta} W_t \mid H_t, I_t = 1, A_t = 0) \Big\} \nonumber \\
        & \quad - \sum_{l=1}^{K-1} \tp_t(l|S_t) \Big\{e^{-f_t(S_t)^T \beta_l}E(Y_{t,\Delta} W_t \mid H_t, I_t = 1, A_t = l)- E(Y_{t,\Delta} W_t \mid H_t, I_t = 1, A_t = 0)\Big\} \bigg]. \label{eq:lem1}
    \end{align}
\end{lem}

\begin{proof}[Proof of \cref{lem:thm1-proofuse}]
    We prove the lemma by doing direct algebra similar to proof in \citep{cohn2023sample, lin2025micro}. To reduce notation burden, within this proof we are using these following short-hand notation: $Y := Y_{t,\Delta}$, $C_k := C_k(A_t)$, $f := f_t(S_t)$, $f\beta_k := f_t(S_t)^T\beta_k$, $g\alpha := g_t(H_t)^T\alpha$, $\tp(k) := \tp_t(k | S_t)$, $p(k) := p_t(k | H_t)$, $W := W_t$, $E_k(\cdot) := E(\cdot \mid H_t, I_t = 1, A_t = k)$. In the derivation below, we will repeatedly use the fact that $E_k(W) = 1$ for all $0 \leq k \leq K$, which follows from the law of iterated expectations. With the short-hand notation, the left-hand side of \eqref{eq:lem1} becomes
    \begin{align}
    \sum_{l=0}^K  E_l\bigg[ e^{-\sum_{k = 1}^K C_kf\beta_k }
    \bigg(Y - e^{\sum_{k = 1}^K C_k f\beta_k + g\alpha} \bigg) \frac{\tp(l)}{p(l)} W C_K f\bigg] p(l) = V_1 + V_2 + V_3, \label{eq:lem1-proofuse1}
    \end{align}
    where
    \begin{align}
        V_1 & :=  E_0\bigg[ e^{-\sum_{k = 1}^K C_kf\beta_k }
    \bigg(Y - e^ {\sum_{k = 1}^K C_k f\beta_k + g\alpha} \bigg) \frac{\tp(0)}{p(0)} W C_K f\bigg] p(0)\nonumber \\
        V_2 & := \sum_{l=1}^{K-1} E_l\bigg[ e^{-\sum_{k = 1}^K C_kf\beta_k }
    \bigg(Y - e^ {\sum_{k = 1}^K C_k f\beta_k + g\alpha} \bigg) \frac{\tp(l)}{p(l)} W C_K f\bigg] p(l) \nonumber \\
        V_3 & := E_K\bigg[ e^{-\sum_{k = 1}^K C_kf\beta_k}
    \bigg(Y - e^{\sum_{k = 1}^K C_k f\beta_k + g\alpha} \bigg) \frac{\tp(K)}{p(K)} W C_K f\bigg] p(K). \nonumber
    \end{align}
    By direct algebra we have
    \begin{align}
        V_1 & = E_0 \bigg[ e^{ \sum_{k = 1}^K \tp(k)f \beta_k}\bigg( Y -e^{ -\sum_{k = 1}^K \tp(k)f \beta_k + g\alpha}\bigg) \tp(0)\{-\tp(K)\} W f \bigg] \nonumber \\
        & = - \bigg[ e^{ \sum_{k = 1}^K \tp(k)f \beta_k} E_0(YW) - e^{ g\alpha} \bigg] \tp(0) \tp(K) f. \label{eq:lem1-proofuse-v1}
    \end{align}
    We also have
    \begin{align}
        V_2 & = \sum_{l=1}^{K-1} E_l\bigg[ e^{-\sum_{k = 1}^K C_kf\beta_k}
    \bigg(Y - e^{\sum_{k = 1}^K C_k f\beta_k + g\alpha} \bigg) \frac{\tp(l)}{p(l)} W C_K f\bigg] p(l) \nonumber \\
            &=\sum_{l=1}^{K-1} E_l\bigg[ \bigg(e^{-\sum_{k = 1}^K C_kf\beta_k } YW -  e^{g\alpha} W \bigg) \tp(l) \{-\tp(K)\} f\bigg] \nonumber \\
        & = -\bigg[e^{\sum_{k = 1}^K \tp(k)f\beta_k}
        \sum_{l=1 }^{K-1}\tp(l)E_l\bigg(e^{ -f\beta_l} YW\bigg)
        {}-{} e^{g\alpha} (1 - \tp(0) - \tp(K))\bigg]\tp(K)f
        \label{eq:lem1-proofuse-v2}
    \end{align}
Finally, we have
    \begin{align}
        V_3 = \bigg[e^{\sum_{k = 1}^K \tp(k)f\beta_k -f\beta_K}
        E_K\bigg(YW\bigg)
        - e^{g\alpha}\bigg]\tp(K) \left\{1 -\tp(K)\right\}f. \label{eq:lem1-proofuse-v3}
    \end{align}
Putting together \eqref{eq:lem1-proofuse-v1}, \eqref{eq:lem1-proofuse-v2}, and \eqref{eq:lem1-proofuse-v3}, we have that the left hand side of \eqref{eq:lem1} equals
    \begin{align}
        & ~~~~ V_1 + V_{2} + V_3 \nonumber \\
        & = e^{\sum_{k = 1}^K \tp(k)f\beta_k} \tp(K) f \bigg[-E_0(YW) \tp(0) - \sum_{l=1}^{K-1} E_l (e^{-f\beta_l}Y W) \tp(l) + E_K(e^{-f\beta_K}YW) \{1-\tp(K)\}\bigg] \nonumber \\
        & ~~~~ + \tp(K) f \bigg[ e^{g\alpha} \tp(0) - e^{g\alpha} \{1 - \tp(K)\} + e^{g\alpha} \{1 - \tp(0) - \tp(K)\} \bigg]. \label{eq:lem1-proofuse-4}
    \end{align}
    The first bracket in \eqref{eq:lem1-proofuse-4} is
    \begin{align}
        & ~~~~ -E_0(YW) \tp(0) - \sum_{l=1}^{K-1} E_l (e^{-f\beta_l}Y W) \tp(l) + E_K(e^{-f\beta_K}YW) \{1-\tp(K)\} \nonumber \\
        & = -E_0(YW) \left\{ 1 - \tp(K) - \sum_{l=1}^{K-1} \tp(l)\right\} - \sum_{l=1}^{K-1} E_l (e^{-f\beta_l}Y W) \tp(l) + E_K(e^{-f\beta_K}YW) \{1-\tp(K)\} \nonumber \\
        & = - \sum_{l=1}^{K-1} \tp(l) \{e^{-f\beta_l}E_l(YW) - E_0(YW)\} + \{1-\tp(K)\} \{e^{-f\beta_K}E_K(YW) - E_0(YW)\}.
        \label{eq:lem1-proofuse-5}
    \end{align}
    The second bracket in \eqref{eq:lem1-proofuse-4} is:
    \begin{align}
        e^{g\alpha} \tp(0) - e^{g\alpha} \{1 - \tp(K)\} + e^{g\alpha} \{1 - \tp(0) - \tp(K)\} = 0. \label{eq:lem1-proofuse-6}
    \end{align}
  
    Putting \eqref{eq:lem1-proofuse-5} and \eqref{eq:lem1-proofuse-6} into \eqref{eq:lem1-proofuse-4} yields \eqref{eq:lem1}. This completes the proof.
\end{proof}

\subsection{Proof of \texorpdfstring{\cref{thm:CAN}}{Theorem 1}}

The proof follows from a standard application of the estimating equation theory \citep[e.g., Section 5 of][]{van2000asymptotic}. The proof consists of five main steps. In Step 1, we show the unbiasedness of the estimating function $m(\alpha,\beta)$. In Step 2, we establish the asymptotic normality of $(\hat\alpha,\hat\beta)$. In Step 3, we derive the explicit form of the asymptotic variance for $\hat\beta$. In Step 4, we will show the consistency of the asymptotic variance estimator. Finally, in Step 5, we will show that the theorem conclusion holds when $\tp_t(k|S_t)$ is estimated. 

In the proof, we use the following definition.
\begin{align}
    r_t(\alpha,\beta) := e^{-\sum_{k = 1}^K C_k(A_t) f_t(S_t)^T\beta_k }Y_{t,\Delta} - e^{ g_t(H_t)^T \alpha}, \label{eq:rt-def-more-complex}
\end{align}
and
\begin{align}
    \tD_t :=
    \begin{bmatrix}
        g_t(H_t) \\
        C_1(A_t) f_t(S_t)\\
        \vdots\\
        C_K(A_t) f_t(S_t)
    \end{bmatrix}
    \quad \text{and} \quad
    D_t :=
    \begin{bmatrix}
        C_1(A_t) f_t(S_t)\\
        \vdots\\
        C_K(A_t) f_t(S_t)
    \end{bmatrix}. \nonumber
\end{align}

\textbf{Step 1:} We show that there exists some $\alpha' \in \RR^q$ such that $E \{m(\alpha',\beta^0)\} = 0$ with $m(\alpha,\beta)$ being the estimating function defined in \eqref{eq:ee} and $\beta^0$ being the true parameter value of $\beta$. We show that for arbitrary $\alpha$ and for all $t \in [T]$ the following holds
\begin{align}
    E \bigg[ I_t \bigg\{e^{-\sum_{k = 1}^K C_k(A_t) f_t(S_t)^T\beta_k } \bigg(Y_{t,\Delta} - e^{\sum_{k = 1}^K C_k(A_t) f_t(S_t)^T \beta_k + g_t(H_t)^T \alpha}\bigg)\bigg\}J_t C_K(A_t) f_t(S_t) \bigg] = 0. \label{eq:thm1-proofuse1}
\end{align}
Using iterated expectations we have
\begin{align}
    & ~~~~ E \bigg[ I_t \bigg\{e^{-\sum_{k = 1}^K C_k(A_t) f_t(S_t)^T\beta_k } \bigg(Y_{t,\Delta} - e^{\sum_{k = 1}^K C_k(A_t) f_t(S_t)^T \beta_k + g_t(H_t)^T \alpha}\bigg)\bigg\}J_t C_K(A_t) f_t(S_t) \bigg] \nonumber \\
    & = E \Bigg( I_t \sum_{l=0}^K E \bigg[\bigg\{e^{-\sum_{k = 1}^K C_k(A_t) f_t(S_t)^T\beta_k } \bigg(Y_{t,\Delta} - e^{\sum_{k = 1}^K C_k(A_t) f_t(S_t)^T \beta_k + g_t(H_t)^T \alpha}\bigg)\bigg\} \nonumber \\
    & \qquad \qquad  \qquad  \qquad \times J_t C_K(A_t) f_t(S_t) ~\Big|~ H_t, I_t = 1, A_t = l \bigg] p_t(l|H_t) \Bigg) \nonumber \\
    & = E \Bigg( I_t  e^{\sum_{k = 1}^K \tp_t(k|S_t) f_t(S_t)^T \beta_k }\tp_t(K|S_t) f_t(S_t) \nonumber \\
        & \qquad \times \bigg[\Big\{1 - \tp_t(K|S_t)\Big\} \Big\{e^{-f_t(S_t)^T \beta_K}E(Y_{t,\Delta} W_t \mid H_t, I_t = 1, A_t = K) - E(Y_{t,\Delta} W_t \mid H_t, I_t = 1, A_t = 0) \Big\} \nonumber \\
        & \qquad \qquad - \sum_{l=1}^{K-1} \tp_t(l|S_t) \Big\{e^{-f_t(S_t)^T \beta_l}E(Y_{t,\Delta} W_t \mid H_t, I_t = 1, A_t = l)- E(Y_{t,\Delta} W_t \mid H_t, I_t = 1, A_t = 0)\Big\} \bigg]\Bigg), \label{eq:thm1-proofuse2}
\end{align}
where \eqref{eq:thm1-proofuse2} follows from \cref{lem:thm1-proofuse}. Then because of the parametric CEE model \eqref{eq:CEE-parametric-model}, it follows immediately from an iterated expectation given $(S_t, I_t)$ that \eqref{eq:thm1-proofuse2} equals 0. Moreover, assuming that there exists a unique $\alpha'\in \RR^q$ such that the following holds
\begin{align}
    \sum_{t=1}^T E \bigg[  I_t \bigg\{e^{-\sum_{k = 1}^K C_k(A_t) f_t(S_t)^T\beta_k } \bigg(Y_{t,\Delta} - e^{\sum_{k = 1}^K C_k(A_t) f_t(S_t)^T \beta_k + g_t(H_t)^T \alpha}\bigg)\bigg\}J_t g_t(H_t)\bigg] = 0. \label{eq:thm1-proofuse3}
\end{align}
Then $E \{m(\alpha',\beta^0)\} = 0$, which completes Step 1. The same argument applied with $C_K(A_t)$ replaced by $C_k(A_t)$ for $k \in [K]$, relabeling which treatment level plays the role of $K$, shows that all $K$ blocks of the $\beta$-estimating equations are unbiased.

\textbf{Step 2:} The asymptotic normality of $(\hat\alpha, \hat\beta)$ follows immediately from Theorems 5.9 and 5.21 of \citet{van2000asymptotic}:
\begin{align}
    \sqrt{n} \left( \begin{bmatrix} \hat\alpha \\ \hat\beta \end{bmatrix} - \begin{bmatrix} \alpha' \\ \beta^0 \end{bmatrix} \right) \dto N\bigg(0, \Big[ E \big\{\dot{m}(\alpha', \beta^0) \big\} \Big]^{-1} \Big[ E \big\{ m(\alpha', \beta^0) m(\alpha', \beta^0)^T \big\} \Big] \Big[ E \big\{\dot{m}(\alpha', \beta^0) \big\} \Big]^{-1,T} \bigg). \label{eq:thm1-proofuse5}
\end{align}

\textbf{Step 3:} Derive the explicit form of the asymptotic variance for $\hat\beta$, i.e., the lower $Kp \times Kp$ principal submatrix of the asymptotic variance in \eqref{eq:thm1-proofuse5}. First, for simpler notation purposes, we define $U_t = e^{-\sum_{k = 1}^K C_k(A_t) f_t(S_t)^T \beta_k}$, and $\Lambda_t = \text{diag}(e^{g_t^T\alpha} 1_q , U_t Y_{t,\Delta} 1_{Kp})$.
By definition, $m(\alpha,\beta) = \sum_{t=1}^T I_t J_t\tD_t  r_t(\alpha,\beta)$, which implies that
\begin{align}
    \dot{m}(\alpha,\beta) = \sum_{t=1}^T I_t J_t \tD_t \dot{r}_t(\alpha,\beta) = - \sum_{t=1}^T I_t J_t \tD_t \tD_t^T \Lambda_t. \nonumber
\end{align}
Therefore,
\begin{align}
    E \big\{\dot{m}(\alpha', \beta^0) \big\} & = - \sum_{t=1}^T E(I_t J_t \tD_t \tD_t^T \Lambda_t)\nonumber \\
    & = - \sum_{t=1}^T E\left\{I_t J_t \begin{bmatrix} e^{g_t^T\alpha'} g_t(H_t)g_t(H_t)^T & U_t Y_{t,\Delta} g_t(H_t) D_t^T \\ e^{g_t^T\alpha'}D_t g_t(H_t)^T & U_t Y_{t,\Delta} D_t D_t^T \end{bmatrix} \right\}. \label{eq:thm1-proofuse6}
\end{align}
Using an argument similar to \eqref{eq:thm1-proofuse2}, we can show that the off-diagonal blocks of \eqref{eq:thm1-proofuse6} have expectation 0. This implies that 
\begin{align}
    E \big\{\dot{m}(\alpha', \beta^0) \big\} = - \begin{bmatrix} \sum_{t=1}^T E \{ I_t J_t e^{g_t(H_t)^T\alpha'} g_t(H_t)g_t(H_t)^T \} & 0 \\ 0 &  \sum_{t=1}^T E ( I_t J_t U_t Y_{t,\Delta}D_t D_t^T ) \end{bmatrix}, \label{eq:thm1-proofuse7}
\end{align}
which implies that only the lower $Kp \times Kp$ principal submatrix of $\{ m(\alpha', \beta^0) m(\alpha', \beta^0)^T \}$ matters when calculating the asymptotic variance of $\hat\beta$. By the definition of $m(\alpha,\beta)$, this submatrix has the form
\begin{align}
    \sum_{t=1}^T \sum_{s=1}^T E \{ I_t I_s J_t J_s r_t(\alpha',\beta^0) r_s(\alpha',\beta^0) D_t D_s^T \}. \label{eq:thm1-proofuse8}
\end{align}
Plug \eqref{eq:thm1-proofuse7} and \eqref{eq:thm1-proofuse8} into \eqref{eq:thm1-proofuse5} and we have the asymptotic variance for $\hat\beta$:
\begin{align}
    & \Big\{\sum_{t=1}^T E ( I_t J_t U_t Y_{t,\Delta} D_t D_t^T )\Big\}^{-1} \Big[\sum_{t=1}^T \sum_{s=1}^T E \{ I_t I_s J_t J_s r_t(\alpha',\beta^0) r_s(\alpha',\beta^0) D_t D_s^T \} \Big] \nonumber \\
    & \qquad \times \Big\{\sum_{t=1}^T E ( I_t J_t U_t Y_{t,\Delta} D_t D_t^T ) \Big\}^{-1, T}. \label{eq:thm1-proofuse9}
\end{align}

\textbf{Step 4:} We show that the asymptotic variance expression in equation \eqref{eq:thm1-proofuse9} can be consistently estimated by replacing all the $E$ with $\PP_n$ and replacing the in-probability limits of the estimators by the estimators themselves in \eqref{eq:thm1-proofuse9}. The consistency of this plug-in estimator follows from standard empirical process arguments; for example, Appendix D of \citet{bao2023estimating}.

\textbf{Step 5:} The two main conclusions of the theorem (asymptotic normality and consistent variance estimator) holds when $\tp_t(k|S_t)$ is estimated either parametrically or nonparametrically. This is because $\tp_t(k|S_t)$ is not used in identifying the parameter $\beta$ and is only used solely for improving efficiency. In particular, the estimating function for $\beta$ is always unbiased regardless of the estimation method of $\tp_t(k|S_t)$. When $\tp_t(k|S_t)$ is estimated parametrically, the results follows from  \citet[Section 4.3]{lok2021estimating}. For nonparametric estimation $\tp_t(k|S_t)$, the results follows from \citet{newey1994asymptotic} or \citet[Theorem 5.1]{cheng2023efficient}.
This completes the proof of \cref{thm:CAN}.

\section{Technical Details for \texorpdfstring{\Cref{sec:sample-size-formula}}{Section 4} on Sample Size Formula}
\label{sec:proof-thm-sample-size-formula}

\subsection{Derivation of \texorpdfstring{$\tL$}{L tilde} in \texorpdfstring{\eqref{eq:null-and-alternative-tilde}}{(7)}}
\label{subsec:proof-L-tilde-matrix}

We prove that when \eqref{eq:working-assumption-on-MEE} hold, i.e., when $\mee_k(t) = f_t^T\beta_k$ for $t \in [T]$ and $k \in [K]$, the following are equivalent:
\begin{align}
    L \times \mee_{1:K}(t) = 0 \text{ for } t \in [T] \quad \text{and} \quad \tL \beta = 0. \label{eq:L-tilde-proofuse-0}
\end{align}
Recall that $L \in \RR^{\nu \times K}$, $\mee_{1:K}(t) := (\mee_1(t), \mee_2(t), \ldots, \mee_K(t))^T$, and $\tL:= L \otimes \II_p$, where $\II_p$ denotes a $p \times p$ identity matrix.

Let $l_{ij}$ be the $(i,j)$-th element of $L$. The $i$-th row of $L \times \mee_{1:K}(t) = 0$ is
\begin{align}
    [l_{i1},\ldots,l_{iK}] \left[\begin{matrix} \mee_1(t) \\ \vdots \\ \mee_K(t)  \end{matrix}\right] = 0. \label{eq:L-tilde-proofuse-1}
\end{align}
Because $\mee_k(t) = f_t^T\beta_k$, \eqref{eq:L-tilde-proofuse-1} is equivalent to
\begin{align}
    [l_{i1},\ldots,l_{iK}] \left[\begin{matrix} f_t^T\beta_1 \\ \vdots \\ f_t^T\beta_K  \end{matrix}\right] = 0. \label{eq:L-tilde-proofuse-2}
\end{align}
The left hand side of \eqref{eq:L-tilde-proofuse-2} is equal to
\begin{align*}
    & ~~~~ [l_{i1},\ldots,l_{iK}] \left[\begin{matrix} f_t^T\beta_1 \\ \vdots \\ f_t^T\beta_K  \end{matrix}\right] \\
    & = l_{i1}f_t^T\beta_1 + \cdots +l_{iK}f_t^T\beta_K \\
    & = f_t^T (l_{i1}\beta_1 + \cdots + l_{iK}\beta_K) \\
    & = f_t^T \left[\begin{matrix} l_{i1}\beta_{11} + \cdots + l_{iK}\beta_{K1} \\ l_{i1}\beta_{12} + \ldots + l_{iK}\beta_{K2} \\ \vdots \\ l_{i1}\beta_{1p} + \cdots + l_{iK}\beta_{Kp} \end{matrix}\right] \\
    & = f_t^T ([l_{i1},\ldots,l_{iK}] \otimes \II_p) \beta. 
\end{align*}
\revisiondel{Therefore, the previous version concluded that $L \times \mee_{1:K}(t) = 0$ was equivalent to $f_t^T(L\otimes \II_p)\beta=0$, and then stacked the vectors $f_t$ to the left of $(L\otimes \II_p)\beta$. These products are not dimensionally conformable when $L$ has more than one row.}

\begin{revisionaddblock}
For each $i\in[\nu]$, define the $p$-vector
\begin{align*}
    \gamma_i:=\sum_{k=1}^K l_{ik}\beta_k
    =([l_{i1},\ldots,l_{iK}]\otimes\II_p)\beta.
\end{align*}
Then \eqref{eq:L-tilde-proofuse-2} is equivalent to $f_t^T\gamma_i=0$. Let $F$ denote the $T\times p$ matrix with $t$-th row $f_t^T$. Requiring \eqref{eq:L-tilde-proofuse-2} for every $t\in[T]$ is therefore equivalent to $F\gamma_i=0$. Because $F$ has full column rank, $F\gamma_i=0$ if and only if $\gamma_i=0$. Repeating this argument for every row $i\in[\nu]$ and stacking the resulting vectors gives
\begin{align}
    \begin{bmatrix}\gamma_1\\ \vdots\\ \gamma_\nu\end{bmatrix}
    =(L\otimes\II_p)\beta=\tL\beta=0. \label{eq:L-tilde-proofuse-3}
\end{align}
The reverse implication follows by reading these steps in reverse. Hence \eqref{eq:L-tilde-proofuse-0} holds. This completes the proof.
\end{revisionaddblock}




\subsection{Additional Definitions}

Recall that in \cref{thm:sample-size-formula}, we assume that the randomization probability depends at most on the decision point index and is thus written as $p_t(k) := P(A_t = k \mid I_t = 1) \equiv P(A_t = k \mid H_t, I_t = 1)$. We also assume that the proximal outcome window $\Delta = 1$ and thus $Y_{t,\Delta}$ is simply written as $Y_t$. We will use $\beta^0$ and $\alpha^0$ to denote the true parameter values corresponding to (WA-a) and (WA-b).

We introduce additional definitions to reduce notation burden. We will sometimes use $\PP$ to denote expectation. In particular, $\PP f := E(f)$. We will write $\sum_{t=1}^T$ as $\sum_t$, $\sum_{t=1}^T \sum_{s=1}^T$ as $\sum_{t,s}$, and $\sum_{1 \leq t, s \leq T: t \neq s}$ as $\sum_{t\neq s}$.

Define
\begin{align}
    r_t(\alpha,\beta) := e^{-\sum_{k = 1}^K C_k(A_t) f_t^T\beta_k }Y_{t} - e^{ g_t^T \alpha}, \label{eq:rt-def}
\end{align}
where $Y_t$ and $A_t$ denote random variables of a generic individual (and thus the subscript $i$ is omitted). We will omit $(\alpha,\beta)$ when they are evaluated at $(\alpha^0, \beta^0)$: i.e., $r_t := r_t(\alpha^0, \beta^0), m := m(\alpha^0, \beta^0)$. 

In addition, define
\begin{align}
    C(A_t) := \begin{bmatrix} C_1(A_t) \\ C_2(A_t) \\ \vdots \\ C_K(A_t) \end{bmatrix}, \quad
    \tD_t :=
    \begin{bmatrix}
        g_t \\
        C_1(A_t) f_t\\
        \vdots\\
        C_K(A_t) f_t
    \end{bmatrix}
    \quad \text{and} \quad
    D_t :=
    \begin{bmatrix}
        C_1(A_t) f_t\\
        \vdots\\
        C_K(A_t) f_t
    \end{bmatrix}. \nonumber
\end{align}

Define matrix $P_t$ as $K \times K$ matrix whose $(k_1,k_2)$-th entry equals $p_t(k_1)\{1-p_t(k_1)\}$ when $k_1 = k_2$ and $-p_t(k_1)p_t(k_2)$ when $k_1 \neq k_2$.  Lastly, we also define matrix $Q_t$ as $K \times K$ matrix whose element is defined as:
\begin{align*}
    (Q_t)_{h,j} = \begin{cases}
        \begin{aligned}
        p_t(h)\bigg( &\{1 - p_t(h)\} \left[p_t(h) + \left\{1 - p_t(h)\right\} e^{-f_t^T\beta_h} - e^{g_t^T\alpha^0}\right] \\
        &+ p_t(h)\left\{\sum_{l\neq 0,h}^K p_t(l) \left(e^{-f_t^T \beta_l} -1\right) \right\}\bigg)
        \end{aligned} &\text{,if $h = j$}\\
        p_t(h) p_t(j) \left[1 + e^{g_t^T\alpha^0} - e^{-f_t^T\beta_h} - e^{-f_t^T\beta_j}+ \sum_{l = 1}^K p_t(l)\left( e^{-f_t^T \beta_l} - 1\right)\right] &\text{,if $h\neq j$} 
    \end{cases}
\end{align*}

\subsection{A Weaker Working Assumption (WA-f)}

We present a working assumption (WA-f) that is an implication of (WA-a)--(WA-e), and we prove \cref{thm:sample-size-formula}(ii) under (WA-a), (WA-b), (WA-c), and (WA-f). (WA-f) is more technical and harder to interpret, so we presented its sufficient conditions (WA-d) and (WA-e) in the paper.

Consider the following working assumption:
\begin{itemize}
    \item[(WA-f)] \revisiondel{\(E\{r_t(\alpha^0,\beta^0)r_s(\alpha^0,\beta^0)\mid I_t=1,I_s=1,A_t,A_s\}\)} \revisionadd{\(E\{r_t(\alpha',\beta^0)r_s(\alpha',\beta^0)\mid I_t=1,I_s=1,A_t,A_s\}\)} is constant in $A_t,A_s$ for all $1\leq s<t\leq T$, where $\alpha'$ is the in-probability limit defined in \cref{lem:thm2-proofuse2a}.
\end{itemize}

(WA-f) is weaker than (WA-d) plus (WA-e) in the sense that when (WA-a) and (WA-b) hold, (WA-d) and (WA-e) imply (WA-f), which we establish in \cref{lem:waf-assump}.

\subsection{Lemmas}
\label{subsec:lemma}
Note that all the lemmas in this subsection concern the particular setting considered in the sample size calculation \Cref{sec:sample-size-formula}, where the interest is in the immediate, marginal CEE (i.e., $\Delta = 1$ and $S_t = \emptyset$ in \eqref{eq:cee-def}), and the randomization probability at each decision point is either constant or dependent only on $t$. We do not include this specification in each lemma statement to avoid repetition. Define $u_t = e^{\sum_{k = 1}^K p_t(k)f_t^T\beta_k}$ and $U_t = e^{-\sum_{k =1}^K C_k(A_t)f_t^T\beta_k}$.

\begin{lem}
    \label{lem:thm2-proofuse1}
    Under (WA-a), we have the following results (the first identity holds without (WA-a)):
    \begin{align}
        \PP I_t C_k(A_t) & = 0, \\
        \PP I_t U_t Y_t C_k(A_t) & = 0, \\
        \PP I_t U_t Y_t C_k(A_t) C_l(A_t) & = - E(I_t) E(Y_t \mid A_t=0, I_t = 1) u_t p_t(k) p_t(l), \\
        \PP I_t U_t  Y_t C_k(A_t)^2 & = E(I_t) E(Y_t \mid A_t = 0, I_t = 1) u_t p_t(k) \{1 - p_t(k)\}.
    \end{align}
\end{lem}

\begin{proof}[Proof of \cref{lem:thm2-proofuse1}]
    The first equation follow immediately by iterated expectation conditional on $I_t = 1$: 
    \begin{align}
        \PP I_t C_k(A_t) &= E(I_t)E[ \{\one(A_t = k) - p_t(k)\} \mid I_t = 1] = 0
    \end{align}

    The second equation follows from
    \begin{align}
        &\PP I_t e^{-\sum_{k = 1}^K C_k(A_t)f_t^T\beta_k} Y_t C_k(A_t)  \nonumber \\
        & = E(I_t) E[ e^{-\sum_{k = 1}^K C_k(A_t)ft^T\beta_k} Y_t \{\one(A_t = k) - p_t(k)\} \mid I_t = 1]  \nonumber \\
        & = E(I_t) p_t(k) u_t \bigg[\left\{e^{-f_t^T\beta_k} E(Y_t\mid A_t= k, I_t = 1) - E(Y_t \mid A_t = 0, I_t = 1)\right\} \nonumber \\ 
        &- \sum_{l \neq 0,k}^K p_t(l) \left\{e^{-f^T\beta_l} E(Y_t\mid A_t= l, I_t = 1) - E(Y_t \mid A_t = 0, I_t = 1)\right\} \nonumber \\
        &- p_t(k) \left\{e^{-f^T\beta_k} E(Y_t\mid A_t= k, I_t = 1) - E(Y_t \mid A_t = 0, I_t = 1)\right\}
        \bigg] \nonumber \\
        &=0 \nonumber,
    \end{align}
    where the last step of the equation is equal to zero because of the MEE assumption that $e^{-f_t^T\beta_k}E(Y_t\mid A_t= k, I_t = 1) - E(Y_t \mid A_t = 0, I_t = 1) = 0$ for all $k = 1, \cdots, K$.

The third equation follows from:
\begin{align}
    &\PP I_t e^{-\sum_{k = 1}^K C_k(A_t)ft^T\beta_k} Y_t C_k(A_t) C_l(A_t) \nonumber \\
    &= E(I_t) E\big(e^{-\sum_{k = 1}^K C_k(A_t)ft^T\beta_k} Y_t C_k(A_t) C_l(A_t) \mid I_t = 1\big)\nonumber \\
    &= - E(I_t) u_t p_t(k) p_t(l) \bigg(E(Y_t \mid A_t= 0 , I_t = 1) \nonumber \\
    &+ \{1 - p_t(l)\}\{e^{-f_t^T\beta_l} E(Y_t \mid A_t = l, I_t = 1) - E(Y_t \mid A_t = 0, I_t = 1)\} \nonumber \\
    &+  \{1 - p_t(k)\}\{e^{-f_t^T\beta_k} E(Y_t \mid A_t = k, I_t = 1) - E(Y_t \mid A_t = 0, I_t = 1)\} \nonumber \\
    &- \sum_{j \neq 0,l,k} p_t(j) \{e^{-f_t^T\beta_j} E(Y_t \mid A_t = j, I_t = 1) - E(Y_t \mid A_t = 0, I_t = 1)\} \bigg) \nonumber \\
    & = - E(I_t) E(Y_t \mid A_t=0, I_t = 1) u_t p_t(k) p_t(l), \nonumber 
\end{align}

Finally, fourth equation can be shown by following similar steps from the previous equation: 

\begin{align}
    &\PP I_t e^{-\sum_{k = 1}^K C_k(A_t)ft^T\beta_k} Y_t C_k(A_t)^2 \nonumber \\
    &= E(I_t) E\big(e^{-\sum_{k = 1}^K C_k(A_t)ft^T\beta_k}  Y_t C_k(A_t)^2  \mid I_t = 1\big)\nonumber \\
     &= E(I_t) u_t p_t(k) \bigg(\{1-p_t(k)\} e^{-f_t^T\beta_k} E(Y_t \mid A_t = k, I_t = 1) \nonumber \\
     &+p_t(k) \{1-p_t(k)\} \{e^{-f_t^T\beta_k} E(Y_t \mid A_t = k, I_t = 1) - E(Y_t \mid A_t = 0, I_t = 1)\} \bigg) \nonumber \\
    &= E(I_t) E(Y_t \mid A_t = 0, I_t = 1) u_t p_t(k) \{1 - p_t(k)\} \nonumber
\end{align}

\end{proof}

For \Cref{lem:thm2-proofuse2a} and the lemmas that follow, we additionally assume the span condition stated in \Cref{thm:sample-size-formula}(ii): for each $k \in [K]$, $p_t(k) f_t$ lies in the linear span of $g_t$ viewed as vector-valued functions of $t$.

\begin{lem}
\label{lem:thm2-proofuse2a}
Suppose $\alpha'$ is the in-probability limit of $\hat{\alpha}$, as defined in \eqref{eq:thm1-proofuse5}. Under assumptions (WA-a) and (WA-b), for all decision point $t$ we have:
\begin{align}
    e^{g_t^T\alpha'} = e^{\sum_{k =1}^K p_t(k) f_t^T\beta^0_k + g_t^T\alpha^0} = u_t e^{g_t^T\alpha^0}
\end{align}
\end{lem}

\begin{proof}
    The proof of \Cref{lem:thm2-proofuse2a} extends the proof in Lemma B.1 of \citet{cohn2023sample}. Because $p_t(k) f_t$ lies in the linear span of $g_t$, this implies that there exist a solution $\alpha^*$ such that $e^{g_t^T\alpha^*} = e^{\sum_{k=1}^K p_t(k) f_t^T\beta_k^0 + g_t^T\alpha^0}$ for all $t$. Display \eqref{eq:thm1-proofuse5} and assumption (WA-a) imply that $\beta' = \beta^0$, where $\beta'$ is the in-probability limit of $\hat{\beta}$ by Theorem 5.9 of \citep{van2000asymptotic}, and furthermore, $\alpha'$ uniquely satisfies
    \begin{align}
    \label{eq:lemmB2-proofuse}
        E\bigg[ \sum_t I_t \left\{ e^{\sum_{k=1}^K -{C_k(A_t)f_t^T\beta_k}} Y_t - e^{g_t^T\alpha'}\right\}g_t\bigg] = 0.
    \end{align}
    Following the same reasoning as in \eqref{eq:thm1-proofuse2}, assumption (WA-a) and (WA-b) guarantee that equation \eqref{eq:lemmB2-proofuse} holds when substituting $\alpha^*$ to $\alpha'$, i.e., replacing $e^{g_t^T\alpha'}$ by $e^{\sum_{k =1}^K p_t(k) f_t^T\beta^0_k + g_t^T\alpha^0}$. This concludes the proof using the uniqueness of $\alpha'$ that we just established.
\end{proof}

\begin{lem}
    \label{lem:thm2-proofuse2}
    Under (WA-a), (WA-b), and (WA-c), we have
    \begin{align}
        \sum_{t=1}^T \PP I_t r_t^2(\alpha',\beta^0) D_t D_t^T = \sum_{t=1}^T E(I_t) u_t^2 e^{g_t^T\alpha^0} Q_t \otimes (f_t f_t^T).
    \end{align}
\end{lem}

\begin{proof}[Proof of \cref{lem:thm2-proofuse2}]
    For an arbitrary decision point $t$ and treatment level $j$,
    \begin{align}
        &E\bigg\{I_t r_t^2(\alpha', \beta^0) C_j^2(A_t)\bigg\} \\
        &= E(I_t) E\bigg( p_t^2(j)  p_t(0) \bigg[u_t E(Y_t \mid A_t = 0, I_t = 1) \left\{ u_t - 2 e^{g^T\alpha'}\right\} + e^{2g_t^T\alpha'}\bigg] \nonumber \\
        &\quad + [1 - p_t(j)]^2 p_t(j) \bigg[u_t e^{-f_t^T\beta_j} E(Y_t \mid A_t = j, I_t = 1)\left\{ u_t e^{-f_t^T\beta_j} -2 e^{g_t^T \alpha'}\right\} + e^{2g_t^T\alpha'}\bigg] \nonumber\\
        &\quad+ \sum_{l\neq 0, j} p_t^2(j) p_t(l) \bigg[u_t e^{-f_t^T\beta_l} E(Y_t \mid A_t = l, I_t = 1)\left\{ u_t e^{-f_t^T\beta_l} -2 e^{g_t^T \alpha'}\right\} + e^{2g_t^T\alpha'}\bigg]
          \bigg) \nonumber \\
        &= E(I_t) E\bigg( p_t^2(j)  p_t(0) \bigg[u_t E(Y_t \mid A_t = 0, I_t = 1) \left\{ u_t - u_t e^{g_t^T\alpha^0} \right\} +  (u_t e^{g_t^T\alpha^0})^2 \bigg] \nonumber \\
        &\quad + [1 - p_t(j)]^2 p_t(j) \bigg[u_t e^{-f_t^T\beta_j} E(Y_t \mid A_t = j, I_t = 1)\left\{ u_t e^{-f_t^T\beta_j} -2  u_t e^{g_t^T\alpha^0} \right\} + ( u_t e^{g_t^T\alpha^0})^2\bigg] \nonumber\\ 
        &\quad + \sum_{l\neq 0, j} p_t^2(j) p_t(l) \bigg[u_t e^{-f_t^T\beta_l} E(Y_t \mid A_t = l, I_t = 1)\left\{ u_t e^{-f_t^T\beta_l} -2  u_t e^{g_t^T\alpha^0} \right\} + ( u_t e^{g_t^T\alpha^0})^2 \bigg]
          \bigg) \nonumber \\
        &= p_t(j) e^{g_t^T\alpha^0} u_t^2 \bigg( \left[1 - p_t(j)\right]\left[p_t(j) + \left\{ 1- p_t(j)\right\} e^{-f_t^T\beta_j} - e^{g_t^T\alpha^0}\right] + \sum_{l \neq 0,j} p_t(l)p_t(j) (e^{-f_t^T\beta_l} - 1)
        \bigg) \label{eq:lemmB3-proofuse1}
    \end{align}
The second step substitutes the result of \Cref{lem:thm2-proofuse2a}, and the last step uses assumption (WA-b).
For an arbitrary decision point $t$ and treatment levels $h$ and $j$, a similar argument to \eqref{eq:lemmB3-proofuse1} gives:
\begin{align}
    & E\left\{I_t r_t^2(\alpha', \beta^0)C_h(A_t) C_j(A_t)\right\} \nonumber \\
    &= E(I_t) E\bigg[ p_t(j) p_t(h) p_t(0) u_t^2 e^{g_t^T\alpha^0} (1 - e^{g_t^T\alpha^0}) + \{1 - p_t(h)\} p_t(h)p_t(j) u_t^2 e^{g_t^T\alpha^0} \left( e^{-f_t^T\beta_h} - e^{g_t^T\alpha^0}\right)\nonumber\\
    &\quad + \{1 - p_t(j)\} p_t(j)p_t(h) u_t^2 e^{g_t^T\alpha^0} \left( e^{-f_t^T\beta_j} - e^{g_t^T\alpha^0}\right) + \sum_{l\neq 0, h,j} p_t(h) p_t(j) p_t(l) u_t^2 e^{g_t^T\alpha^0} \left(e^{-f_t^T\beta_l} - e^{g_t^T\alpha^0}\right)\bigg] \nonumber \\
    &= E(I_t) p_t(h) p_t(j) u_t^2  e^{g_t^T\alpha^0} \left\{ 1 + e^{g_t^T\alpha^0} - e^{-f_t^T\beta_h} - e^{-f_t^T\beta_j} + \sum_{l =1}^K p_t(l) \left(e^{-f_t^T\beta_l} -1\right) \right\} \label{eq:lemmB3-proofuse2}
\end{align}

    Finally, it is easy to verify that
    \begin{align}
        D_t D_t^T = \begin{bmatrix} C_1^2(A_t) & \ldots & C_1(A_t) C_K(A_t) \\ \vdots & & \vdots \\ C_K(A_t)C_1(A_t) & \ldots & C_K^2(A_t) \end{bmatrix} \otimes (f_t f_t^T). \label{eq:lem2-thm2-proofuse5}
    \end{align}
    So putting together \eqref{eq:lemmB3-proofuse1}, \eqref{eq:lemmB3-proofuse2}, and \eqref{eq:lem2-thm2-proofuse5} yields
    \begin{align}
        \PP I_t r_t^2(\alpha',\beta^0) D_t D_t^T = E(I_t)u_t^2 e^{g_t^T \alpha^0} Q_t \otimes (f_t f_t^T). \nonumber
    \end{align}
    Thus the lemma statement follows. This completes the proof.
\end{proof}

\begin{lem} \label{lem:waf-assump}
    (WA-a), (WA-b), (WA-d), and (WA-e) imply (WA-f), moreover, one can also show that:
    \begin{align*}
        E\left\{ r_t(\alpha', \beta^0) r_s(\alpha', \beta^0) \mid I_s = 1, I_t = 1, A_s, A_t\right\} = 0
    \end{align*}
\end{lem}

\begin{proof}[Proof of \cref{lem:waf-assump}]
\revisiondel{This proofs extends the arguments presented in Lemma B.2} \revisionadd{This proof extends the argument in Lemma B.2 of \citet{cohn2023sample}.}
    First, we define:
    \begin{align}\label{eq:lem-waf:proofuse1}
        E(Y_t \mid I_s = 1, I_t = 1, A_s, A_t) = h_{ts}(A_t, A_s),
    \end{align} 
    for some function $h_{ts}$ that depends only on random variables $A_s$ and $A_t$. 

    \revisiondel{Second, by assumption (WA-e) that $I_t$ is independent of $Y_s$ which implies that:} \revisionadd{Second, (WA-e) implies that $I_t$ is independent of $Y_s$, so}
    \begin{align}\label{eq:lem-waf:proofuse2}
        E(Y_s \mid I_s = 1, I_t = 1, A_t, A_s) &= E(Y_s \mid I_s = 1, A_s) = e^{g_s^T\alpha^0 + \sum_{k = 1}^K \one_{(A_s = k )} f_s^T\beta_k},
    \end{align}
    \revisiondel{\(E\{E(Y_tY_s\mid I_t=1,I_s=1,A_t,A_s,Y_s)\mid I_t=1,I_s=1,A_t,A_s,Y_s\}\)}

    \revisionadd{Using iterated expectation, we obtain}
    \begin{align}
        &E(Y_s Y_t \mid I_s =1, I_t = 1, A_s, A_t) \nonumber\\
        &= E\left\{E(Y_t Y_s \mid I_t = 1, I_s = 1, A_t, A_s, Y_s) \mid I_t = 1, I_s = 1, A_t, A_s\right\} \nonumber\\
        &= E\left\{Y_s E(Y_t\mid I_t = 1, I_s = 1, A_t, A_s, Y_s) \mid I_t = 1, I_s = 1, A_t, A_s)\right\} \nonumber \\
        &= E\left\{Y_s E(Y_t\mid I_t = 1, I_s = 1, A_t, A_s) \mid I_t = 1, I_s = 1, A_t, A_s)\right\} \label{eq:lem-waf:proofuse3}\\
        &= h_{ts}(A_t, A_s) E(Y_s \mid I_s = 1, I_t = 1, A_t ,A_s)\nonumber \\
        &= h_{ts}(A_t, A_s) e^{g_s^T\alpha^0 + \sum_{k = 1}^K \one_{(A_s = k )} f_s^T\beta_k} \label{eq:lem-waf:proofuse4}.
    \end{align}
Equation \eqref{eq:lem-waf:proofuse3} follows from the no-serial-correlation assumption (WA-d), and \eqref{eq:lem-waf:proofuse4} follows from \eqref{eq:lem-waf:proofuse1} and \eqref{eq:lem-waf:proofuse2}. Recall from \Cref{subsec:lemma} that
\par\noindent\revisiondel{\(u_t=\sum_{k=1}^Kp_t(k)f_t^T\beta_k\).}
\par\noindent\revisiondel{The previous second identity omitted the treatment-level sum in its exponent.}
\begin{revisionaddblock}
\begin{align*}
u_t&=\exp\left\{\sum_{k=1}^Kp_t(k)f_t^T\beta_k\right\},\\
U_t&=\exp\left\{-\sum_{k=1}^KC_k(A_t)f_t^T\beta_k\right\}
=u_t\exp\left\{-\sum_{k=1}^K\one_{(A_t=k)}f_t^T\beta_k\right\}.
\end{align*}
\end{revisionaddblock}
\revisiondel{and Using a direct algebra calculation on assumption (WA-f);} \revisionadd{Direct calculation then gives}
\begin{align}
    &E\left\{ r_t(\alpha', \beta^0) r_s(\alpha', \beta^0) \mid I_s = 1, I_t = 1, A_s, A_t\right\} \nonumber\\
    &= E\left[\left\{U_t Y_t - e^{g_t^T\alpha'}\right\} \left\{ U_s Y_s - e^{g_s^T\alpha'}\right\} \mid I_t = 1, I_s = 1, A_s,A_t\right] \nonumber \\
    &=U_t U_s E(Y_t Y_s \mid I_t = 1, I_s = 1, A_s,A_t) - U_t e^{g_s^T\alpha'} E( Y_t \mid I_t = 1, I_s = 1, A_s,A_t) \nonumber \\
    &- U_s e^{g_t^T\alpha'} E(Y_s \mid I_t = 1, I_s = 1, A_s,A_t) + e^{g_t^T\alpha'} e^{g_s^T\alpha'} \nonumber\\
    &=U_t U_s h_{ts}(A_t, A_s) e^{g_s^T\alpha^0 + \sum_{k = 1}^K \one_{(A_s = k )} f_s^T\beta_k} - U_t e^{g_s^T\alpha'} h_{ts}(A_t, A_s) \nonumber \\
    &- U_s e^{g_t^T\alpha'}  e^{g_s^T\alpha^0 + \sum_{k = 1}^K \one_{(A_s = k )} f_s^T\beta_k} + e^{g_t^T\alpha'} e^{g_s^T\alpha'} \label{eq:lem-waf-proofuse5} \\
    &=U_t u_s h_{ts}(A_t, A_s) e^{g_s^T\alpha^0} - U_t u_s e^{g_s^T\alpha^0} h_{ts}(A_t, A_s)- u_s u_t e^{g_t^T\alpha^0}  e^{g_s^T\alpha^0} + u_t u_s e^{g_t^T\alpha^0} e^{g_s^T\alpha^0} \label{eq:lem-waf-proofuse6} \\
    &=0. \nonumber
\end{align}

Equation \eqref{eq:lem-waf-proofuse5} follows from \eqref{eq:lem-waf:proofuse1}, \eqref{eq:lem-waf:proofuse2}, and \eqref{eq:lem-waf:proofuse4}. Equation \eqref{eq:lem-waf-proofuse6} uses
\par\noindent\revisiondel{The previous identity used $f_t$ in the exponent.}
\begin{revisionaddblock}
\begin{align*}
U_s\exp\left\{\sum_{k=1}^K\one_{(A_s=k)}f_s^T\beta_k\right\}=u_s.
\end{align*}
\end{revisionaddblock}
This completes the proof.
\end{proof}

\begin{lem}
    \label{lem:thm2-proofuse3}
    Under (WA-a), (WA-b), and (WA-f), we have for any $1 \leq t \neq s \leq T$
    \revisiondel{\(\PP I_t r_t(\alpha^0,\beta^0)r_s(\alpha^0,\beta^0)D_tD_s^T=0\).}
    \begin{revisionaddblock}
    \begin{align}
        \PP I_t I_s r_t(\alpha',\beta^0) r_s(\alpha',\beta^0) D_t D_s^T = 0.
    \end{align}
    \end{revisionaddblock}
\end{lem}

\begin{proof}[Proof of \cref{lem:thm2-proofuse3}]
    For any $1 \leq k, l \leq K$ (including the case of $k = l$), we have
    \par\noindent\revisiondel{The previous display evaluated the residuals at $(\alpha^0,\beta^0)$.}
    \begin{revisionaddblock}
    \begin{align}
        &\PP I_t I_s r_t(\alpha',\beta^0)r_s(\alpha',\beta^0) C_k(A_t) C_l(A_s) \nonumber\\
        &\quad = E(I_t I_s) E\{E[r_t(\alpha',\beta^0)r_s(\alpha',\beta^0) \mid I_t = 1, I_s = 1, A_t, A_s] C_k(A_t) C_l(A_s) \mid I_t = 1, I_s = 1\} \nonumber \\
        & = 0, \label{eq:lem3-thm2-proofuse1}
    \end{align}
    Here \eqref{eq:lem3-thm2-proofuse1} follows because (WA-f) makes the inner conditional expectation constant in $(A_t,A_s)$, while sequential randomization with probabilities that depend at most on $t$ gives $E\{C_k(A_t)\mid I_t=1,I_s=1,A_s\}=0$.
    \end{revisionaddblock}

    Recall we have shown that:
    \revisiondel{\(D_tD_s^T=[C_k(A_t)C_l(A_s)]_{k,l=1}^K\otimes(f_tf_t^T)\).}
    \begin{revisionaddblock}
    \begin{align}
        D_t D_s^T = \begin{bmatrix} C_1(A_t) C_1(A_s) & \ldots & C_1(A_t) C_K(A_s) \\ \vdots & & \vdots \\ C_K(A_t)C_1(A_s) & \ldots & C_K(A_t) C_K(A_s) \end{bmatrix} \otimes (f_t f_s^T). \label{eq:lem3-thm2-proofuse3}
    \end{align}
    \end{revisionaddblock}
    Equations \eqref{eq:lem3-thm2-proofuse1} and \eqref{eq:lem3-thm2-proofuse3} imply the lemma. This completes the proof.
\end{proof}

\subsection{Proof of \texorpdfstring{\cref{thm:sample-size-formula}}{Theorem 2} Under Weaker Working Assumptions}

\cref{thm:sample-size-formula}(i) is already established in \cref{subsec:hypothesis-test-statistic-rejection-region}. We now prove \cref{thm:sample-size-formula}(ii) under working assumptions (WA-a), (WA-b), \revisiondel{(WA-d)}, \revisionadd{(WA-c)}, and (WA-f). Given the derivation in \cref{subsec:sample-size-formula} up to \cref{thm:sample-size-formula}, it suffices to show that under these working assumptions the asymptotic variance of $\hat\beta$, $M^{-1}\Sigma M^{-1,T}$, simplifies to the explicit form \eqref{eq:thm2-final-variance} derived below, in which $M = \sum_t E(I_t) u_t e^{g_t^T\alpha^0} P_t \otimes (f_t f_t^T)$ and $\Sigma = \sum_t E(I_t) u_t^2 e^{g_t^T\alpha^0} Q_t \otimes (f_t f_t^T)$; these are the matrices computed in the ``Compute $M$ and $\Sigma$ matrix'' step of \cref{alg:ss-calculator}.

\cref{sec:sample-size-formula} considers the setting where the randomization probability depends at most on $t$ and thus can be written as $p_t(k)$, and thus we can set $\tp_t(k) = p_t(k)$. Therefore, the estimating equation for $(\alpha,\beta)$ defined in \eqref{eq:ee} becomes
\begin{align}
    m(\alpha,\beta) &:= 
    \sum_t^T I_t 
    \times\bigg[\exp\left\{-\sum_{k = 1}^K C_k(A_t) f_t^T\beta_k \right\} Y_{t} - \exp \left\{g_t^T \alpha\right\} \bigg] 
    \begin{bmatrix}
        g_t \\
        C_1(A_t) f_t\\
        \vdots\\
        C_K(A_t) f_t
    \end{bmatrix}
    = \sum_{t=1}^T I_t r_t(\alpha,\beta) \tD_t. \nonumber
\end{align}

Now we calculate the asymptotic variance for $\hat\beta$ under the working assumptions. By \cref{thm:CAN}, the asymptotic variance for $(\hat\alpha,\hat\beta)$ has the form
\begin{align}
    \{\PP \dot{m}(\alpha', \beta^0)\}^{-1} ~ \{\PP m(\alpha', \beta^0) m(\alpha', \beta^0)^T\} ~ \{\PP \dot{m}(\alpha', \beta^0)\}^{-1,T}. \label{eq:thm2-proofuse1}
\end{align}

We compute each term in \eqref{eq:thm2-proofuse1}.

\textbf{Step 1: Compute $\{\PP \dot{m}(\alpha', \beta^0)\}^{-1}$.} We have $\dot{m}(\alpha,\beta) = \sum_t I_t \tD_t \dot{r}_t(\alpha, \beta) = - \sum_t I_t \tD_t \tD_t^T \Lambda_t$, where $\Lambda_t = \text{diag}(e^{g_t^T\alpha} 1_q, U_t Y_t 1_{Kp})$ as in \Cref{sec:proof-thm-CAN}. Therefore,
\begin{align}
    \PP\dot{m}(\alpha, \beta) 
    & = - \PP \sum_{t=1}^T I_t \begin{bmatrix} e^{g_t^T\alpha} g_t g_t^T & U_t Y_t g_t D_t^T \\ e^{g_t^T\alpha}D_t g_t^T & U_t Y_t D_t D_t^T \end{bmatrix} \nonumber \\
    & = -  \sum_{t=1}^T \begin{bmatrix} E(I_t)e^{g_t^T\alpha} g_tg_t^T & g_t E(I_t U_t Y_t D_t^T) \\ E(I_t e^{g_t^T\alpha}D_t) g_t^T & E(I_t U_t Y_t D_t D_t^T) \end{bmatrix}.\label{eq:thm2-proofuse2}
\end{align}

\cref{lem:thm2-proofuse1} implies that $g_t E(I_t U_t Y_t D_t^T) = E(I_t e^{g_t^T\alpha} D_t) g_t^T = 0$ and $E(I_t U_t Y_t D_t D_t^T) = E(I_t) u_t e^{g_t^T\alpha^0} P_t \otimes f_t f_t^T$, where $P_t \in \RR^{K \times K}$ whose $(k_1,k_2)$-th entry equals $p_t(k_1)\{1-p_t(k_1)\}$ when $k_1 = k_2$ and $-p_t(k_1)p_t(k_2)$ when $k_1 \neq k_2$. Therefore, for  $(\alpha',\beta^0)$ we have
\begin{align}
    \{\PP \dot{m}(\alpha', \beta^0)\}^{-1} = - \begin{bmatrix} \{\sum_t E(I_t) e^{g_t^T\alpha'} g_t g_t^T\}^{-1} & 0 \\ 0 & \{\sum_t E(I_t) u_t e^{g_t^T\alpha^0} P_t \otimes (f_t f_t^T)\}^{-1} \end{bmatrix}. \label{eq:thm2-proofuse3}
\end{align}

\textbf{\revisiondel{Step 2: Compute $\PP m(\alpha^0,\beta^0)m(\alpha^0,\beta^0)^T$.}}
\revisionadd{\textbf{Step 2: Compute $\PP m(\alpha',\beta^0)m(\alpha',\beta^0)^T$.}}
\begin{revisionaddblock}
Because the goal is to derive the asymptotic variance for $\hat\beta$ (the lower $Kp \times Kp$ principal submatrix of \eqref{eq:thm2-proofuse1}) and because the off-diagonal blocks of $\{\PP \dot{m}(\alpha',\beta^0)\}^{-1}$ are 0 by \eqref{eq:thm2-proofuse3}, it suffices to compute the lower $Kp \times Kp$ principal submatrix of $\PP m(\alpha',\beta^0)m(\alpha',\beta^0)^T$. We have
\begin{align}
    &\PP m(\alpha',\beta^0)m(\alpha',\beta^0)^T \nonumber\\
    &\quad = \PP \Big\{\sum_t I_t r_t(\alpha',\beta^0)\tD_t\Big\}
    \Big\{\sum_s I_s r_s(\alpha',\beta^0)\tD_s\Big\}^T \nonumber \\
    &\quad = \PP \sum_{t,s} I_t I_s r_t(\alpha',\beta^0)r_s(\alpha',\beta^0)
    \begin{bmatrix} g_t g_s^T & g_t D_s^T \\ D_t g_s^T & D_t D_s^T \end{bmatrix}. \nonumber
\end{align}
Therefore, the lower $Kp \times Kp$ principal submatrix of $\PP m(\alpha',\beta^0)m(\alpha',\beta^0)^T$ is
\begin{align}
    &\PP \sum_{t,s} I_t I_s r_t(\alpha',\beta^0)r_s(\alpha',\beta^0)D_tD_s^T \nonumber\\
    &\quad = \sum_t \PP I_t r_t^2(\alpha',\beta^0)D_tD_t^T
    + \sum_{t\neq s}\PP I_tI_s r_t(\alpha',\beta^0)r_s(\alpha',\beta^0)D_tD_s^T \nonumber\\
    & = \sum_{t=1}^T E(I_t) u_t^2 e^{g_t^T\alpha^0} Q_t \otimes (f_t f_t^T) \label{eq:thm2-proofuse4}
\end{align}
where \eqref{eq:thm2-proofuse4} follows from \cref{lem:thm2-proofuse2} and \cref{lem:thm2-proofuse3}. 
\end{revisionaddblock}

\textbf{Step 3: Computing the asymptotic variance of $\hat\beta$.} Plugging \eqref{eq:thm2-proofuse3} and \eqref{eq:thm2-proofuse4} into \eqref{eq:thm2-proofuse1} implies that the asymptotic variance of $\hat\beta$ is
\begin{align}
    \left\{\sum_t E(I_t) u_t e^{g_t^T\alpha^0} P_t \otimes (f_t f_t^T)\right\}^{-1} \left\{ \sum_{t=1}^T E(I_t) u_t^2 e^{g_t^T\alpha^0} Q_t \otimes (f_t f_t^T) \right\}    \left\{\sum_t E(I_t) u_t e^{g_t^T\alpha^0} P_t \otimes (f_t f_t^T)\right\}^{-1}. \label{eq:thm2-final-variance}
\end{align}

This completes the proof.

\section{Simulation Results on Consistency and Asymptotic Normality}
\label{sec:simulation-estimator}

We evaluate the consistency and asymptotic normality of the estimator $\hat\beta$ for the CEE model \eqref{eq:CEE-parametric-model}, proposed in \Cref{sec:estimator}. We focus on the setting where the length of the excursion $\Delta = 1$ and participants are always available ($I_t \equiv 1$). In the generative model, we set the total number of decision points for each individual $T = 30$. A time-varying covariate $Z_t$, which is independent of all variables and takes three values 0, 1, and 2 with equal probability, is generated.
The number of active treatment levels is $K = 2$, and the categorical treatment $A_t \in \{0,1,2\}$. The binary outcome $Y_t = Y_{t,\Delta = 1}$ is generated as $Y_t \sim \text{Bern}(E(Y_t \mid H_t, A_t))$, where
\begin{align}
E(Y_t \mid H_t, A_t) &= \{0.2 \one(Z_t = 0) + 0.5 \one(Z_t = 1) + 0.4 \one(Z_t = 2)\} \nonumber \\
&\quad \times \exp\{\one(A_t = 1)(0.1 + 0.3 Z_t) + \one(A_t = 2)(0.25 + 0.3 Z_t)\}.
\end{align}
We consider two randomization scenarios. In Scenario 1, the treatment probabilities are equal: $P(A_t = 0) = P(A_t = 1) = P(A_t = 2) = 1/3$. In Scenario 2, the treatment assignment probabilities depend on $Z_t$:
\begin{align*}
p_t(A_t \mid Z_t = 0) = (0.5,\, 0.3,\, 0.2), \quad p_t(A_t \mid Z_t = 1) = (0.2,\, 0.5,\, 0.3), \quad p_t(A_t \mid Z_t = 2) = (0.3,\, 0.2,\, 0.5),
\end{align*}
for $(A_t = 0, 1, 2)$. In both scenarios, the numerator probabilities $\tp_t(k \mid S_t)$ are set to the constant $1/3$, so $J_t \equiv 1$ in Scenario 1 and $J_t \neq 1$ in Scenario 2.

We consider two sets of estimands in the simulation. When setting $S_t = \emptyset$ in \eqref{eq:cee-def}, the fully marginal CEEs are
\begin{align*}
\cee_{t1}(\emptyset) = \log \frac{0.2 e^{0.1} + 0.5 e^{0.4} + 0.4 e^{0.7}}{1.1}, \qquad \cee_{t2}(\emptyset) = \log \frac{0.2 e^{0.25} + 0.5 e^{0.55} + 0.4 e^{0.85}}{1.1}.
\end{align*}
When setting $S_t = Z_t$, the moderated CEEs are $\cee_{t1}(Z_t) = 0.1 + 0.3 Z_t$ and $\cee_{t2}(Z_t) = 0.25 + 0.3 Z_t$. We use the estimating function \eqref{eq:ee} to estimate both CEEs by setting different $S_t$'s. The working model for $E(Y_{t,\Delta} \mid H_t, I_t = 1, A_t = 0)$ is misspecified as $\exp\{g_t(H_t)^\top\alpha\} = \exp\{\alpha_0 + \alpha_1 Z_t\}$.

We compare the proposed estimator, labeled EMEE-catA in the tables below, with two generalized estimating equations (GEE) fits of a log-linear outcome model. For the marginal CEEs the GEE mean model is
\begin{align*}
E(Y_t \mid Z_t, A_t) = \exp\{\gamma_0 + \gamma_1 Z_t + \gamma_2 \one(A_t = 1) + \gamma_3 \one(A_t = 2)\},
\end{align*}
and for the moderated CEEs it further includes the interactions $\one(A_t = 1) Z_t$ and $\one(A_t = 2) Z_t$. We fit this model by GEE with a log link and the Poisson variance function, with robust standard errors \citep{zeger1986longitudinal}, under a working independence correlation structure (GEE-ind) and an exchangeable working correlation structure (GEE-exch). The GEE coefficients on the treatment indicators, and on their interactions with $Z_t$, are compared with the corresponding components of the CEE.

The two approaches coincide in one special case. When $f_t = 1$ and $g_t = 1$, the randomization probability depends on neither $H_t$ nor $t$, and $\tp_t$ is set equal to that probability, both estimating equations are saturated in the treatment and are solved by the level-wise sample averages, so both give $\hat\beta_k = \log(\bar Y_k / \bar Y_0)$, where $\bar Y_k$ is the sample average of $Y_t$ over the person-decision points at which level $k$ was assigned. The same finite-sample identity holds within levels of a categorical moderator when $f_t(S_t)$ consists of indicators for the levels of $S_t$ and $g_t$ contains only those indicators. If $g_t$ also contains baseline or time-varying covariates, the two estimators generally solve different empirical equations and need not agree exactly. Under the same randomization conditions, treatment is independent of those covariates, so the two estimators nevertheless target the same population CEE parameters when the treatment-by-moderator portion of the GEE mean model is saturated. This is the situation in analysis (a) of \Cref{subsec:application-estimator}, which explains why the proposed estimator and the GEE have similar performance.

The three settings reported below start near that special case and change one feature at a time. In \Cref{tbl:est_marginal_scenario1} the estimand is the marginal CEE and the randomization does not depend on $H_t$, but $g_t$ is misspecified and $Z_t$ enters the GEE mean model. \Cref{tbl:est_moderator} changes the estimand to the moderated CEE, holding the randomization uniform; there $Z_t$ enters the GEE mean model linearly, both on its own and in the treatment interactions, while the true $E(Y_t \mid Z_t, I_t = 1, A_t = 0)$ is not log-linear in $Z_t$. \Cref{tbl:est_marginal} instead keeps the marginal estimand and lets the randomization depend on $Z_t$.

The bias, Monte Carlo standard deviation of the estimates, average estimated standard error, and 95\% confidence interval coverage probability are based on 1000 replicates.

In all three settings the proposed estimator has negligible bias and coverage close to the nominal level at every sample size, which is what \Cref{thm:CAN} leads one to expect: its consistency requires the CEE model to be correctly specified and places no restriction on the working model $g_t$. The GEE fits behave differently across the three. In \Cref{tbl:est_marginal_scenario1} they track the proposed estimator closely, with comparable bias, standard deviation, and coverage, even though their mean model is misspecified in $Z_t$. In \Cref{tbl:est_moderator} they carry a bias of about $+0.07$ on the intercept and $-0.06$ on the $Z_t$ coefficient at every sample size, and their coverage falls at $n = 100$ to about 0.88 and to between 0.82 and 0.84. In \Cref{tbl:est_marginal} their bias is about $+0.14$ for treatment level 1 and $+0.07$ for level 2, and coverage falls at $n = 100$ to 0.22 and 0.72. In the latter two settings the GEE bias is roughly constant in $n$ while the standard deviation shrinks, which is what drives the coverage down. The two GEE fits differ from each other only in the third or fourth decimal place throughout, because the generative model draws $Z_t$ independently across decision points and includes no random effect, leaving no within-person dependence for the exchangeable working correlation to exploit.

\begin{table}[htbp]
\caption{\label{tbl:est_marginal_scenario1} Simulation results for estimating the marginal CEEs ($S_t = \emptyset$) under Scenario 1, in which the randomization probabilities are equal and do not depend on $Z_t$, with 2 active treatment levels and 1000 replicates. EMEE-catA, the proposed estimator; GEE-ind and GEE-exch, log-link GEE with a working independence and an exchangeable working correlation structure, respectively; SD, Monte Carlo standard deviation of the estimates across replicates; SE, average of the estimated standard errors across replicates; CP, coverage probability of the 95\% confidence interval.}
\resizebox{0.8\textwidth}{!}{
\begin{tabular}{ccccccc}
\toprule

Trt & Method & Sample size & Bias & SD & SE & CP\\
\midrule
1 & EMEE-catA
 & 20  & 0.00615 & 0.10713 & 0.10937 & 0.971 \\
 & & 30  & -0.00164 & 0.08892 & 0.08874 & 0.958 \\
 & & 40  & 0.00024 & 0.07581 & 0.07644 & 0.956 \\
 & & 50  & -0.00319 & 0.06908 & 0.06796 & 0.943 \\
 & & 100 & 0.00102 & 0.04725 & 0.04783 & 0.957 \\
\cmidrule(l){2-7}
 & GEE-ind
 & 20  & 0.00641 & 0.10739 & 0.10344 & 0.939 \\
 & & 30  & -0.00156 & 0.08904 & 0.08558 & 0.933 \\
 & & 40  & 0.00032 & 0.07596 & 0.07443 & 0.944 \\
 & & 50  & -0.00314 & 0.06914 & 0.06656 & 0.932 \\
 & & 100 & 0.00106 & 0.04738 & 0.04738 & 0.949 \\
\cmidrule(l){2-7}
 & GEE-exch
 & 20  & 0.00657 & 0.10766 & 0.10324 & 0.938 \\
 & & 30  & -0.00131 & 0.08928 & 0.08551 & 0.931 \\
 & & 40  & 0.00021 & 0.07596 & 0.07434 & 0.944 \\
 & & 50  & -0.00303 & 0.06918 & 0.06651 & 0.932 \\
 & & 100 & 0.00106 & 0.04743 & 0.04736 & 0.947 \\
\midrule
2 & EMEE-catA
 & 20  & 0.00617 & 0.10163 & 0.10312 & 0.967 \\
 & & 30  & 0.00453 & 0.08199 & 0.08339 & 0.959 \\
 & & 40  & 0.00013 & 0.07200 & 0.07228 & 0.952 \\
 & & 50  & -0.00407 & 0.06390 & 0.06440 & 0.945 \\
 & & 100 & 0.00001 & 0.04472 & 0.04520 & 0.956 \\
\cmidrule(l){2-7}
 & GEE-ind
 & 20  & 0.00640 & 0.10193 & 0.09753 & 0.934 \\
 & & 30  & 0.00460 & 0.08222 & 0.08038 & 0.935 \\
 & & 40  & 0.00018 & 0.07208 & 0.07038 & 0.943 \\
 & & 50  & -0.00397 & 0.06411 & 0.06308 & 0.932 \\
 & & 100 & 0.00008 & 0.04478 & 0.04478 & 0.952 \\
\cmidrule(l){2-7}
 & GEE-exch
 & 20  & 0.00656 & 0.10211 & 0.09739 & 0.935 \\
 & & 30  & 0.00482 & 0.08244 & 0.08032 & 0.929 \\
 & & 40  & 0.00019 & 0.07210 & 0.07028 & 0.944 \\
 & & 50  & -0.00396 & 0.06418 & 0.06304 & 0.932 \\
 & & 100 & 0.00005 & 0.04483 & 0.04476 & 0.951 \\
\bottomrule
\end{tabular}}
\end{table}

\begin{table}[htbp]
\caption{\label{tbl:est_moderator} Simulation results for estimating the moderated CEEs ($S_t = Z_t$) under Scenario 1, with 2 active treatment levels and 1000 replicates. EMEE-catA, the proposed estimator; GEE-ind and GEE-exch, log-link GEE with a working independence and an exchangeable working correlation structure, respectively; SD, Monte Carlo standard deviation of the estimates across replicates; SE, average of the estimated standard errors across replicates; CP, coverage probability of the 95\% confidence interval.}
\resizebox{\textwidth}{!}{
\begin{tabular}{ccccccccccc}
\toprule

\multicolumn{3}{c}{ } & \multicolumn{4}{c}{Intercept} & \multicolumn{4}{c}{$Z_t$} \\

Trt & Method & Sample size & Bias & SD & SE & CP & Bias & SD & SE & CP\\
\midrule
1 & EMEE-catA
 & 20  & -0.00888 & 0.25744 & 0.26217 & 0.954 & 0.01039 & 0.17397 & 0.17861 & 0.967 \\
 & & 30  & -0.00087 & 0.20599 & 0.21046 & 0.966 & 0.00315 & 0.13736 & 0.14182 & 0.964 \\
 & & 40  & -0.00010 & 0.17982 & 0.18171 & 0.957 & 0.00265 & 0.12101 & 0.12251 & 0.954 \\
 & & 50  & 0.00353 & 0.15874 & 0.16116 & 0.965 & -0.00040 & 0.10922 & 0.10892 & 0.953 \\
 & & 100 & -0.00582 & 0.11264 & 0.11283 & 0.948 & 0.00474 & 0.07623 & 0.07637 & 0.952 \\
\cmidrule(l){2-11}
 & GEE-ind
 & 20  & 0.07266 & 0.21556 & 0.20540 & 0.908 & -0.05491 & 0.13565 & 0.13083 & 0.906 \\
 & & 30  & 0.07707 & 0.17427 & 0.16992 & 0.918 & -0.05898 & 0.10832 & 0.10697 & 0.905 \\
 & & 40  & 0.07805 & 0.15224 & 0.14844 & 0.903 & -0.05961 & 0.09560 & 0.09355 & 0.884 \\
 & & 50  & 0.08015 & 0.13403 & 0.13277 & 0.907 & -0.06144 & 0.08639 & 0.08394 & 0.864 \\
 & & 100 & 0.07101 & 0.09546 & 0.09438 & 0.888 & -0.05657 & 0.06059 & 0.05982 & 0.836 \\
\cmidrule(l){2-11}
 & GEE-exch
 & 20  & 0.07253 & 0.21605 & 0.20508 & 0.908 & -0.05462 & 0.13609 & 0.13066 & 0.908 \\
 & & 30  & 0.07692 & 0.17454 & 0.16960 & 0.918 & -0.05884 & 0.10859 & 0.10679 & 0.907 \\
 & & 40  & 0.07799 & 0.15228 & 0.14827 & 0.904 & -0.05959 & 0.09560 & 0.09344 & 0.884 \\
 & & 50  & 0.08030 & 0.13406 & 0.13267 & 0.907 & -0.06151 & 0.08634 & 0.08389 & 0.867 \\
 & & 100 & 0.07109 & 0.09546 & 0.09434 & 0.885 & -0.05659 & 0.06057 & 0.05981 & 0.836 \\
\midrule
2 & EMEE-catA
 & 20  & -0.00769 & 0.24335 & 0.24828 & 0.960 & 0.00957 & 0.16425 & 0.17002 & 0.966 \\
 & & 30  & -0.00117 & 0.19951 & 0.20245 & 0.959 & 0.00505 & 0.13340 & 0.13635 & 0.962 \\
 & & 40  & 0.00329 & 0.17291 & 0.17237 & 0.944 & 0.00190 & 0.11635 & 0.11610 & 0.945 \\
 & & 50  & -0.00021 & 0.15442 & 0.15376 & 0.942 & 0.00237 & 0.10598 & 0.10422 & 0.950 \\
 & & 100 & -0.00415 & 0.10582 & 0.10755 & 0.962 & 0.00351 & 0.07150 & 0.07287 & 0.961 \\
\cmidrule(l){2-11}
 & GEE-ind
 & 20  & 0.07370 & 0.20430 & 0.19472 & 0.917 & -0.05568 & 0.12855 & 0.12449 & 0.909 \\
 & & 30  & 0.07774 & 0.16945 & 0.16335 & 0.907 & -0.05786 & 0.10544 & 0.10264 & 0.904 \\
 & & 40  & 0.08110 & 0.14599 & 0.14100 & 0.911 & -0.06015 & 0.09150 & 0.08862 & 0.891 \\
 & & 50  & 0.07682 & 0.13037 & 0.12676 & 0.906 & -0.05910 & 0.08372 & 0.08031 & 0.880 \\
 & & 100 & 0.07227 & 0.09011 & 0.09010 & 0.878 & -0.05747 & 0.05699 & 0.05710 & 0.820 \\
\cmidrule(l){2-11}
 & GEE-exch
 & 20  & 0.07362 & 0.20530 & 0.19425 & 0.917 & -0.05552 & 0.12908 & 0.12420 & 0.908 \\
 & & 30  & 0.07752 & 0.16966 & 0.16311 & 0.908 & -0.05772 & 0.10576 & 0.10249 & 0.903 \\
 & & 40  & 0.08116 & 0.14618 & 0.14090 & 0.912 & -0.06022 & 0.09153 & 0.08855 & 0.890 \\
 & & 50  & 0.07677 & 0.13036 & 0.12666 & 0.905 & -0.05907 & 0.08376 & 0.08025 & 0.880 \\
 & & 100 & 0.07236 & 0.09022 & 0.09007 & 0.878 & -0.05749 & 0.05707 & 0.05710 & 0.820 \\
\bottomrule
\end{tabular}}
\end{table}

\begin{table}[htbp]
\caption{\label{tbl:est_marginal} Simulation results for estimating the marginal CEEs ($S_t = \emptyset$) under Scenario 2, in which the randomization probabilities depend on $Z_t$, with 2 active treatment levels and 1000 replicates. EMEE-catA, the proposed estimator; GEE-ind and GEE-exch, log-link GEE with a working independence and an exchangeable working correlation structure, respectively; SD, Monte Carlo standard deviation of the estimates across replicates; SE, average of the estimated standard errors across replicates; CP, coverage probability of the 95\% confidence interval.}
\resizebox{0.8\textwidth}{!}{
\begin{tabular}{ccccccc}
\toprule

Trt & Method & Sample size & Bias & SD & SE & CP\\
\midrule
1 & EMEE-catA
 & 20  & 0.00496 & 0.11906 & 0.12136 & 0.959 \\
 & & 30  & 0.00433 & 0.09847 & 0.09704 & 0.950 \\
 & & 40  & 0.00250 & 0.07955 & 0.08320 & 0.965 \\
 & & 50  & 0.00328 & 0.07375 & 0.07443 & 0.949 \\
 & & 100 & 0.00277 & 0.05208 & 0.05250 & 0.952 \\
\cmidrule(l){2-7}
 & GEE-ind
 & 20  & 0.14224 & 0.11703 & 0.11153 & 0.751 \\
 & & 30  & 0.13911 & 0.09501 & 0.09188 & 0.687 \\
 & & 40  & 0.13904 & 0.07782 & 0.07964 & 0.599 \\
 & & 50  & 0.13859 & 0.07248 & 0.07146 & 0.511 \\
 & & 100 & 0.13867 & 0.05230 & 0.05100 & 0.224 \\
\cmidrule(l){2-7}
 & GEE-exch
 & 20  & 0.14251 & 0.11724 & 0.11138 & 0.751 \\
 & & 30  & 0.13926 & 0.09522 & 0.09183 & 0.689 \\
 & & 40  & 0.13890 & 0.07783 & 0.07956 & 0.603 \\
 & & 50  & 0.13853 & 0.07243 & 0.07141 & 0.510 \\
 & & 100 & 0.13867 & 0.05235 & 0.05099 & 0.223 \\
\midrule
2 & EMEE-catA
 & 20  & 0.00219 & 0.11357 & 0.11712 & 0.959 \\
 & & 30  & 0.00483 & 0.09585 & 0.09305 & 0.951 \\
 & & 40  & 0.00261 & 0.07729 & 0.08045 & 0.960 \\
 & & 50  & 0.00287 & 0.07164 & 0.07190 & 0.956 \\
 & & 100 & 0.00228 & 0.05100 & 0.05051 & 0.944 \\
\cmidrule(l){2-7}
 & GEE-ind
 & 20  & 0.07037 & 0.10868 & 0.10551 & 0.912 \\
 & & 30  & 0.06987 & 0.09038 & 0.08611 & 0.868 \\
 & & 40  & 0.06892 & 0.07396 & 0.07506 & 0.869 \\
 & & 50  & 0.06922 & 0.06852 & 0.06755 & 0.841 \\
 & & 100 & 0.06794 & 0.04994 & 0.04799 & 0.723 \\
\cmidrule(l){2-7}
 & GEE-exch
 & 20  & 0.07047 & 0.10873 & 0.10534 & 0.910 \\
 & & 30  & 0.07008 & 0.09051 & 0.08608 & 0.870 \\
 & & 40  & 0.06879 & 0.07397 & 0.07499 & 0.868 \\
 & & 50  & 0.06921 & 0.06844 & 0.06752 & 0.839 \\
 & & 100 & 0.06795 & 0.04994 & 0.04798 & 0.723 \\
\bottomrule
\end{tabular}}
\end{table}

\section{Simulation Results on Performance of Sample Size Formula}
\label{sec:simulation-sample-size-formula}

\subsection{Generative Models}

We assess the performance of the sample size calculator and the testing procedure proposed in \Cref{sec:sample-size-formula} in terms of the type I error rate and power, when all working assumptions hold and when some working assumptions are violated. For all simulations, the desired type I error rate is 0.05 and the desired power is 0.8. In the generative models, the number of active treatment levels is $K = 2$ so $A_t \in \{0,1,2\}$, and we test for $\cH_0: \mee_1(t) = \mee_2(t)$ for all $t \in [T]$ vs. $\cH_1: \mee_1(t) \neq \mee_2(t)$ for some $t \in [T]$. Throughout this section we write $\rate$ for $\rate_{12} = \ate_2 / \ate_1$, defined in \eqref{eq:ratio-ate}.


We consider a simple generative model (GM-0) and two generative models with more complicated features: a generative model where the outcomes are serially correlated (GM-SC), and a generative model where the availability process depends on previous outcomes and previous treatments, i.e., is endogenous (GM-EA). We denote parameters associated with the true generative model by superscript ``*'', and the input of the sample size formula, i.e., the working model, by superscript ``w''.

GM-0 is characterized by a set of parameters : $\{p^*_{kt}: k=0,1,2\}, \alpha^*, g_t^*, \{\beta_k^*:k=1,2\}, f_t^*$, and $\tau^*(t)$ for $t \in [T]$. $p^*_{kt}$ is the randomization probability for treatment level $k$: $p^*_{kt} = P(A_t = k \mid I_t =1)$, and recall that $P(A_t = 0 \mid I_t =0) = 1$. $\alpha^*$ and $g_t^*$ determine the mean proximal outcome under no treatment. $\beta_k^*$ and $f_t^*$ determine the marginal excursion effect of treatment $k$. $\tau(t)^*$ is the time-varying availability pattern. For each individual, their observations are generated sequentially over $t \in [T]$:  $I_t \sim \text{Bernoulli}(\tau^*(t))$. $A_t = 0$ if $I_t = 0$. $A_t \mid I_t = 1$ is generated from a categorical distribution with support $\{0, 1, 2\}$ and probabilities $\{p^*_{kt}: k=0,1,2\}$. The success probability of $Y_t$ is generated by $\exp\{g_t^{*T} \alpha^* + \one(A_t = 1) f_t^{*T} \beta_1^* + \one(A_t = 2) f_t^{*T} \beta_2^*\}$. Details about the specifications of the parameters are given in \Cref{box:detail-gm-0-part1} and \Cref{box:detail-gm-0-part2}.

For GM-SC, $I_t$ and $A_t$ are generated the same way as in GM-0, and $Y_t$ is generated differently. In GM-SC, the success probability of $Y_t$ is generated by $\exp\{g_t^{*T} \alpha^* + \one(A_t = 1) f_t^{*T} \beta_1^* + \one(A_t = 2) f_t^{*T} \beta_2^* + \nu_1 Y_{t-1}\}$. We only consider the case where the causal effects $\mee_k(t)$ are constant and each individual is always available ($\tau^*(t) = 1$ for all $t$). Details about the specifications of $\spnc^*(t)$ are given in \Cref{box:detail-gm-SC}. 

For GM-EA, $I_t$ is generated from a Bernoulli distribution with success probability
\begin{align*}
    0.5 + \nu_2 \one_{(A_{t-1} = 1)} + \nu_3 \one_{(A_{t-1} = 2)} + \nu_4 Y_{t-1}.
\end{align*}
The $\nu_2$-term and $\nu_3$-term make $I_t$ depend on the previous treatment, while the $\nu_4$-term makes $I_t$ \revisiondel{depends on previous outcome}\revisionadd{depend on the previous outcome}, and $\nu_2$, $\nu_3$, and $\nu_4$ take values between $-0.2$ and $0.2$. $A_t$ and $Y_t$ are generated in the same way as in GM-0.

\begin{guidelinebox}[htbp]
    \caption{Details about GM-0: $\tau(t)$ and $\spnc(t)$.}
    \label{box:detail-gm-0-part1}
    \begin{mdframed}[linewidth=1pt, roundcorner=10pt, backgroundcolor=gray!10]
    \spacingset{1}

    \noindent \textbf{Parameterization of $\tau(t)$:}

    Given AA, we consider two \revisiondel{parameterization}\revisionadd{parameterizations} of $\tau(t)$: constant and linear.
    \begin{itemize}
        \item Constant $\tau(t)$: We set $\tau(t) = AA$ for all $t$.
        \item Linear $\tau(t)$ parameterized by $\theta_{\tau}$: We set $\tau^*(t)$ to vary linearly in $t$, with $\tau(1) = AA +\theta_\tau$ and $\tau(T) = AA - \theta_\tau$. In other words,
        \revisiondel{\(\tau(t)=-\frac{2}{T-1}\theta_\tau t+AA-\frac{T+1}{T-1}\theta_\tau\)}
        \revisionadd{\(\tau(t)=-\frac{2}{T-1}\theta_\tau t+AA+\frac{T+1}{T-1}\theta_\tau\).}
    \end{itemize}

    \noindent \textbf{Parameterization of $\spnc(t)$:}

    Given $\aspn$ and $\tau(t)$, we consider three parameterizations of $\spnc(t)$: constant, linear, and quadratic similar to parameterization in the binary treatment paper in \citep{cohn2023sample}. 
    \begin{itemize}
        \item Constant $\spnc(t)$: We set $\spnc(t) = \aspn$ for all t. 
        \item Linear $\spnc(t)$ parameterized by $\theta_g$: we set $\log\spnc(t) = \alpha_0 + \alpha_1 t$, where $(\alpha_0, \alpha_1)$ is determined as follows. A positive (negative) $\theta_g$ indicates a linearly increasing (decreasing) $\spnc(t)$. In particular, if $\theta_g = 0$ we set $\alpha_1 = 0$ and $\alpha_0 = \log \aspn$. If $\theta_g \in (-1, 0)\cup (0,1)$, we set 
        \begin{align} 
            \frac{\alpha_0 + \alpha_1}{\alpha_0 + T \alpha_1} = \frac{1 + \theta_g}{1 - \theta_g} \label{eq:theta-g}
        \end{align} 
        and thus 
        \begin{align*}
            \alpha_1 = \frac{2 \theta_g}{ 1- T - \theta_g(1 + T)} \alpha_0,
        \end{align*}
        and then we solve for $\alpha_0$ and $\alpha_1$ from equation \eqref{eq:theta-g} using the given $\aspn$ and $\tau(t)$.
        \item Log-quadratic $\spnc(t)$ parameterized by $\theta_g$: we set $\log \spnc(t) = \alpha_0 + \alpha_1 t + \alpha_2 t^2$ where $(\alpha_0, \alpha_1, \alpha_2)$ is determined as follows. We constrain $\spnc(1) = \spnc(T)$ and a positive (negative) $\theta_g$ indicates the parabola $\log \spnc(t)$ opens upward (downward). In particular for $\theta_g \in (-1, 1)$, we set: 
        \begin{align}
            \alpha_0 + \alpha_1 + \alpha_2 &= \alpha_0 + \alpha_1 T + \alpha_2 T^2 \nonumber \\
            \frac{\log \spnc\{(T+1)/2\}}{\log \spnc(1)} &= \frac{\alpha_0 + \frac{T+1}{2}\alpha_1 + \frac{(T+1)^2}{4} \alpha_2}{\alpha_0 + \alpha_1 + \alpha_2} = \frac{1 + \theta_g}{ 1 - \theta_g}
        \end{align}
    \end{itemize}
    \end{mdframed}
\end{guidelinebox}

\begin{guidelinebox}[htbp]
    \caption{Details about GM-0: $\mee_k(t)$.}
    \label{box:detail-gm-0-part2}
    \begin{mdframed}[linewidth=1pt, roundcorner=10pt, backgroundcolor=gray!10]
    \spacingset{1}
    \noindent \textbf{Parameterization of $\mee_k(t)$:}

    Given $\ate_1$, $\ate_2$, and $\tau(t)$, we consider two parameterizations of $\mee_1(t)$ and $\mee_2(t)$: \revisiondel{constants and linear,}\revisionadd{constant and linear.}
    \begin{itemize}
        \item Constant $\mee_k(t)$: \revisiondel{We set $\mee_k(t)=\ate_k$ for all $t\in[T]$.}\revisionadd{We set $\mee_k(t)=\log(\ate_k)$ for all $t\in[T]$.}
        \item Linear $\mee_k(t)$ parameterized by $\theta_{f1}$ and $\theta_{f2}$, where \revisiondel{$\theta_{f1}$ parameterize the slope of $\mee_1(t)$, while $\theta_{f2}$ parameterize the difference between the slope of two treatment level, i.e $\theta_{f2}=0$ indicates that the two treatment have a parallel slope.}\revisionadd{$\theta_{f1}$ parameterizes the slope of $\mee_1(t)$, while $\theta_{f2}$ parameterizes the difference between the slopes for the two active treatment levels. Thus, $\theta_{f2}=0$ indicates parallel slopes.} We set $\mee_1(t) = \beta_1 + \beta_2 t$, where $(\beta_1, \beta_2)$ is determined as follows: a positive (negative) $\theta_{f1}$ indicates a linearly increasing (decreasing) $\mee_1(t)$. In particular, if $\theta_{f1} = 0$ we set $\beta_2 = 0$ and \revisiondel{$\beta_1=\ate$}\revisionadd{$\beta_1=\log(\ate_1)$}. If $\theta_{f1} \in (-1, 0)\cup (0,1)$, we set 
        \begin{align} 
            \frac{\beta_1+ \beta_2}{\beta_1 + T \beta_2} = \frac{1 + \theta_{f1}}{1 - \theta_{f1}}.
        \end{align} 
        We set $\mee_2(t) = \mee_1(t) + \beta_3 + \beta_4 t$. When $\theta_{f2} = 0$, we set $\beta_4 = 0$ and \revisiondel{$\beta_3=\ate_2(t)$}\revisionadd{$\beta_3=\log(\ate_2)-\log(\ate_1)$}. If $\theta_{f2} \in (-1, 0)\cup (0,1)$, we set 
        
        \begin{align} 
            \frac{\beta_3 + \beta_4}{\beta_3 + T \beta_4} = \frac{1 + \theta_{f2}}{1 - \theta_{f2}}.
        \end{align} 
        and then we solve for $\beta_1$, $\beta_2$, $\beta_3$ and $\beta_4$ from equation \eqref{eq:ate_k} using the given $\ate_k$ and $\tau(t)$.
    \end{itemize}
    \end{mdframed}
\end{guidelinebox}

\begin{guidelinebox}[htbp]
    \caption{Details about GM-SC: Generative Model with Serial Correlation in Outcome}
    \label{box:detail-gm-SC}
    \begin{mdframed}[linewidth=1pt, roundcorner=10pt, backgroundcolor=gray!10]
    \spacingset{1}

    Here we describe the generative model with serial correlation (GM-SC), where the success probability for $Y_t$ is 
    \begin{align}
        e^{g_t^T\alpha + \sum_{k=1}^{K} \one(A_t = k) f_t^T\beta_k + \nu_1 Y_{t-1}}.
    \end{align}

     We consider the case where the individuals are always available ($\tau(t) = 1$ for all $t$) and constant causal effects $\mee_k(t)$ for simplicity purposes. Since serial correlation, indicated by parameter $\nu_1$, affects the calculation of $\spnc$, the parameterization of $\spnc(t)$ and $\aspn$ are different from GM-0.

    \noindent\textbf{Parameterization of $\spnc(t)$:}
    \begin{align*}
        \spnc(t) &:= E(Y_t | A_t = 0, I_t = 1) \\
                 &= e^{g_t^T\alpha} E(e^{\nu_1 Y_{t-1}}) \\
                 &= e^{g_t^T\alpha} \left[ 1 + (e^{\nu_1} - 1) \spnc(t-1)\left\{ p_{t-1}(0) + \sum_{k = 1}^K e^{f_{t-1}^T\beta_k} p_{t-1}(k)\right\}\right]
    \end{align*}

    with $\spnc(1) = e^{g_1^T\alpha}$, because there is no $Y_0$ term at $t = 1$. Therefore, one can compute $\spnc(t)$ for all $t$ recursively using the equation above. Since $\tau(t) = 1$ for all $t$ in GM-SC, the average availability is $\aaa = 1$.
    \end{mdframed}
\end{guidelinebox}

\begin{guidelinebox}[htbp]
    \caption{Details about GM-EA: parameterization of $\tau(t)$}
    \label{box:detail-gm-ea}
    \begin{mdframed}[linewidth=1pt, roundcorner=10pt, backgroundcolor=gray!10]
    \spacingset{1}

    \noindent\textbf{Parameterization of $\tau(t)$:}
    \begin{align*}
        E(I_t) &= 0.5 + \nu_2 E(\one_{(A_{t-1} = 1)}) + \nu_3 E(\one_{(A_{t-1} = 2)}) + \nu_4 E(Y_{t-1}) \\
        &\approx 0.5 + \nu_2 E(\one_{(A_{t-1} = 1)} \mid I_{t-1} = 1) E(I_{t-1}) +
        \nu_3 E(\one_{(A_{t-1} = 2)} \mid I_{t-1} = 1) E(I_{t-1}) \\
        &\quad + \nu_4 E(Y_{t-1}|I_{t-1}= 1) E(I_{t-1}) \\
        &= 0.5 + \nu_2 p_{t-1}(1) E(I_{t-1}) + \nu_3 p_{t-1}(2) E(I_{t-1}) \\
        &\quad + \nu_4 \left\{e^{g_{t-1}^T\alpha} \left(p_{t-1}(0) + \sum_{k=1}^{K} e^{f_{t-1}^T\beta_k} p_{t-1}(k)\right)\right\} E(I_{t-1}),
    \end{align*}

    and $E(I_1) = 0.5$. Therefore, one can compute $\tau(t) = E(I_t)$ for all $t$ using the equation above and thus compute $\aaa$ using \eqref{eq:AA}.

    \noindent The two terms involving $\nu_2$ and $\nu_3$ are exact, because $A_{t-1} = 0$ whenever $I_{t-1} = 0$. The term involving $\nu_4$ is an approximation: in GM-EA the proximal outcome is generated at every decision point, including the unavailable ones (at which $A_{t-1} = 0$), so the exact recursion carries the additional term $\nu_4 e^{g_{t-1}^T\alpha}\{1 - E(I_{t-1})\}$. We use the approximate recursion above to set the working availability $\tau^\w(t)$ in \Cref{subsec:simulation-violate-f}. Across the settings reported there it misstates $\aaa$ by at most $0.02$ (about $4\%$ in relative terms), so the settings described there as having a correctly specified $\aaa^\w$ are correct only up to this discrepancy. The conclusions drawn in \Cref{subsec:simulation-violate-f} concern how the power responds to $\nu_2$, $\nu_3$ and $\nu_4$, and are not affected by a shift of this size.
    \end{mdframed}
\end{guidelinebox}

\spacingset{1.9}

\subsection[Type I error rate when working assumptions do or do not hold]{Type I error rate when working assumptions do or do not hold}
\label{subsec:simulation-result-type-i-error}

To assess the type I error rate numerically, we considered three sets of generative models: (a) all working assumptions hold; (b) (WA-a) is violated; and (c) one or more of (WA-b)-(WA-f) is violated but (WA-a) hold. Note that \Cref{thm:sample-size-formula}(i) guarantees asymptotic type I error rate control for (a) and (c) but not (b). Simulation results show that the type I error rate is controlled at the 0.05 level in all three scenarios, with the empirical rate in each individual setting consistent with the nominal level up to Monte Carlo error (\Cref{fig:typeierror_hist}), including when (WA-a) is violated (middle panels of \Cref{fig:typeierror_hist}). We note that the small sample correction helped to improve the type I error rate control (the empirical type I error rate without small sample correction is not shown in the figure).

Specifically, for (a) we considered 1,617 settings (each corresponding to a specific parameterization of the data generating distribution) when all working assumptions hold using GM-0; \Cref{box:sim-all-wa-hold-1000-settings} lists the factors that generate these settings. For (b) we misspecified separately the magnitude of $\mee(t)$ or the pattern of $\mee(t)$ using GM-0; \Cref{subsec:simulation-violate-a} on power contains the specific form of misspecification. For (c) we used various generative models from GM-0, GM-SC, and GM-EA.

\begin{figure}[htbp]
\centering
\includegraphics[width=1\textwidth]{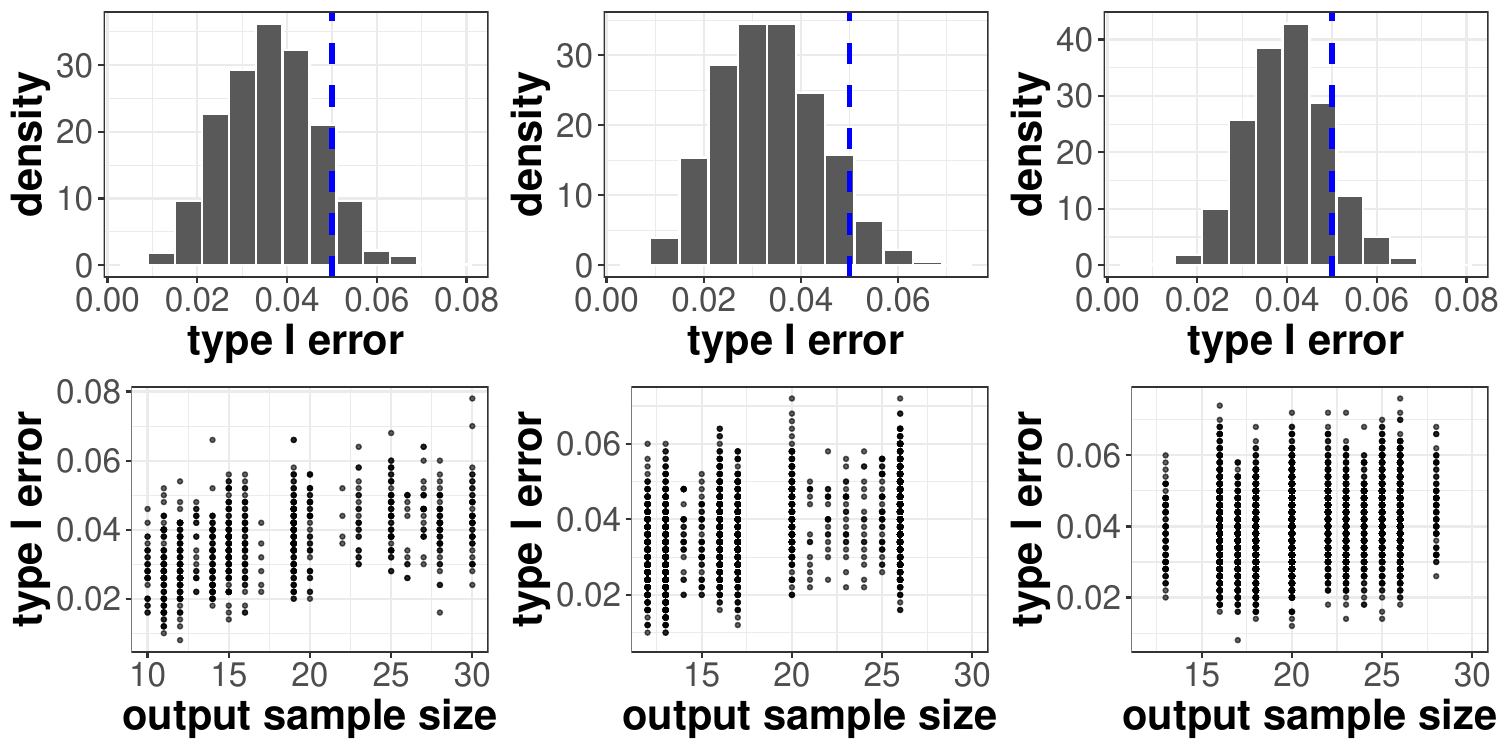}
  \caption{\textbf{Left plots:} type I error rate when all working assumptions hold. \textbf{Middle plots:} type I error rate when (WA-a) is violated. \textbf{Right plots:} type I error rate when other WAs are violated but (WA-a) holds. Each dot in the bottom plots represent the type I error rate under a specific generative model obtained from 1,000 simulation replicates. Histograms in the top plots capture the distribution of empirical type I error rates from simulations using the various generative models.}
  \label{fig:typeierror_hist}
\end{figure}

\begin{guidelinebox}[htbp]
    \caption{The 1,617 simulation settings when all working assumptions hold.}
    \label{box:sim-all-wa-hold-1000-settings}
    \begin{mdframed}[linewidth=1pt, roundcorner=10pt, backgroundcolor=gray!10]
    \spacingset{1}

    The simulation settings form the full factorial design of the six factors listed below, from which we remove the combinations in which $p_t(k) f_t$ does not lie in the linear span of $g_t$. This leaves 1,620 settings, of which 1,617 generate success probabilities inside $[0,1]$ and are the settings reported here. Because the grid of $\theta_f$ values repeats the value 0, settings with a flat $\mee_k(t)$ enter the design twice; 972 of the 1,617 settings are distinct. This affects only the weight each setting receives in the histograms of \Cref{fig:typeierror_hist} and \Cref{fig:power_working_assumption}, not the range of settings covered.

    \begin{itemize}
        \item Factor 1: number of decision points, $T = 15$ or $T = 30$ (2 levels).
        \item Factor 2: patterns of $f_t$, $g_t$, and $p_t$ (18 levels), each ensuring that $f_t$ is a subset of $g_t$. The randomization probability is constant and equal across treatment options, $p_t = (1/3, 1/3, 1/3)$, in every level:
        \begin{enumerate}
            \item (3 levels) constant $f_t$; linear $g_t$ with $\theta_g = -0.3, 0, 0.3$
            \item (3 levels) constant $f_t$; quadratic $g_t$ with $\theta_g = -0.3, 0, 0.3$
            \item (6 levels) linear $f_t$ with $\theta_f = -0.3, 0$; linear $g_t$ with $\theta_g = -0.3, 0, 0.3$
            \item (6 levels) linear $f_t$ with $\theta_f = -0.3, 0$; quadratic $g_t$ with $\theta_g = -0.3, 0, 0.3$
        \end{enumerate}
        \item Factor 3: $\ate^*$ of the two treatment options (3 levels).
        \begin{enumerate}
            \item $\ate_1^* = 1.2$
            \item $\ate_2^* = 1.35$, $1.4$, or $1.6$
        \end{enumerate}
        \item Factor 4: $\aspn^* = 0.2$, $0.3$, or $0.4$ (3 levels).
        \item Factor 5: value of $\aaa^*$ (3 levels): $\aaa^* = 0.5$, $0.8$, or $1$. In every setting of this study $\tau(t)$ is constant in $t$, so $\tau(t) = \aaa^*$ for all $t$; settings with a non-constant $\tau(t)$ are studied separately in \Cref{subsec:simulation-violate-d}.
        \item Factor 6: desired type I error rate $\eta = 0.05$ and desired power $1 - b = 0.8$ (1 level).
    \end{itemize}
    \end{mdframed}
\end{guidelinebox}

\spacingset{1.9}

\subsection{Power when all working assumptions hold}
\label{subsec:simulation-all-hold}

Using GM-0, we considered 1,617 simulation settings under a variety of settings with different patterns and magnitudes of $\mee^*_k(t)$, $\ate^*_k$, $\aspn^*$, and $\aaa^*$, assuming all working assumptions hold. That is, when $\mee^*_k(t) = \mee^\w_k(t)$, $\ate^*_k = \ate^\w_k$, $\aspn^* = \aspn^\w$, and $\aaa^* = \aaa^\w$ (see \Cref{box:sim-all-wa-hold-1000-settings} for more detailed simulation settings). In all 1,617 settings we considered, the power is consistently around the desired level (\Cref{fig:power_working_assumption}). The power is slightly greater than 0.8 before applying the small sample correction and slightly below the desired 0.8 after the correction is applied. 

\begin{figure}[htbp]
\centering
\includegraphics[width=1\textwidth]{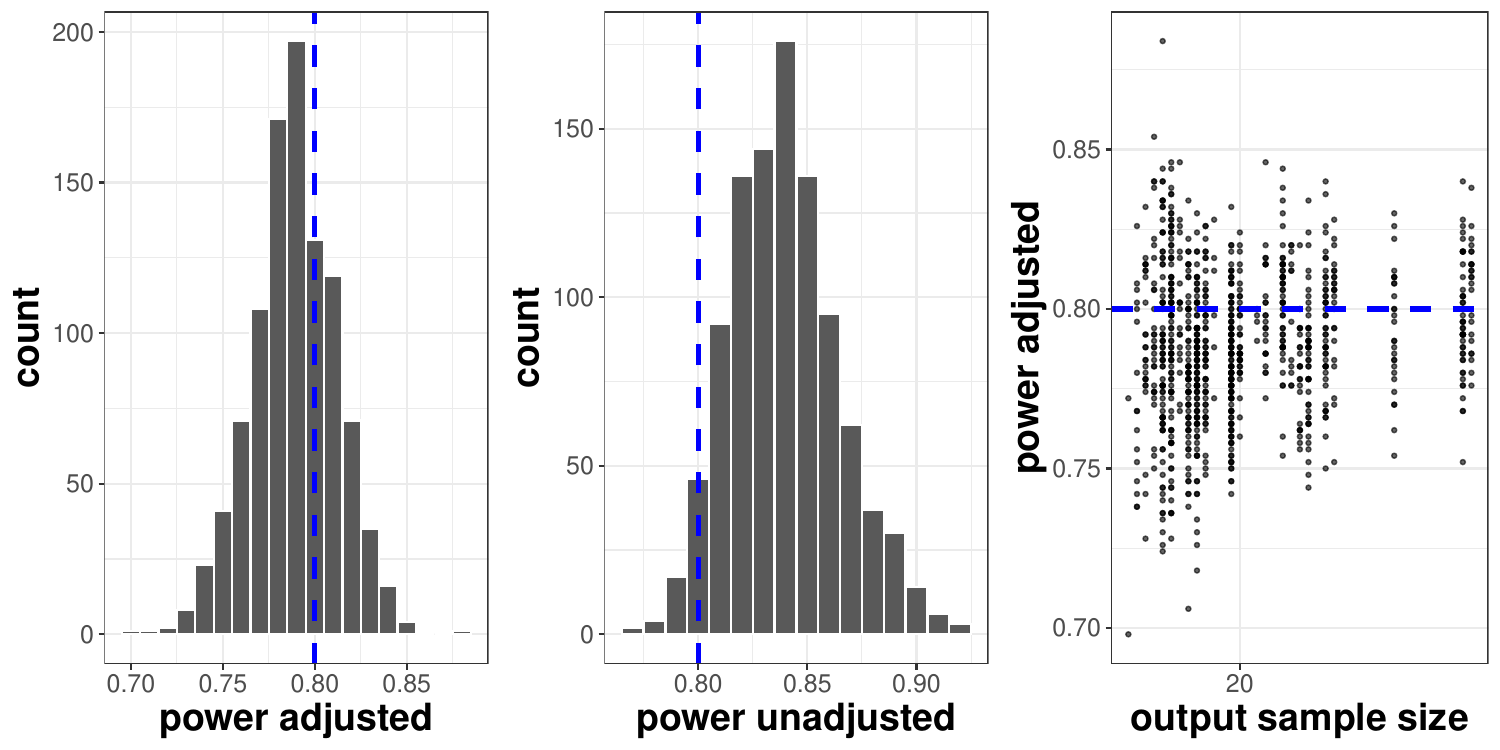}
  \caption{Adjusted and unadjusted power when all working assumptions hold. The first two figures show the histogram of the adjusted and unadjusted power under 1,000 settings. The last figure shows the adjusted power under each setting against the sample size $n$ obtained from the sample size formula.}
  \label{fig:power_working_assumption}
\end{figure}

\subsection{Power when (WA-a) is violated}
\label{subsec:simulation-violate-a}

Using GM-0, we consider two scenarios where (WA-a) can be violated: incorrect magnitude of $\mee_k(t)$ and incorrect time-varying pattern of $\mee_k(t)$. 

In the first scenario, where the ratio between two treatment effects ($\rate := \frac{\ate_2}{\ate_1}$) is misspecified while the working $\mee_k^\w(t)$ is specified to be of the same time-varying pattern as $\mee_k^*(t)$ (constant or linear), if $ \rate^\w >  \rate^*$ then the MRT is under-powered, and if $ \rate^\w <  \rate^*$ then the MRT is over-powered (\Cref{fig:viol miss ATE}).

\begin{figure}[htbp]
\centering
    \includegraphics[width=1\textwidth]{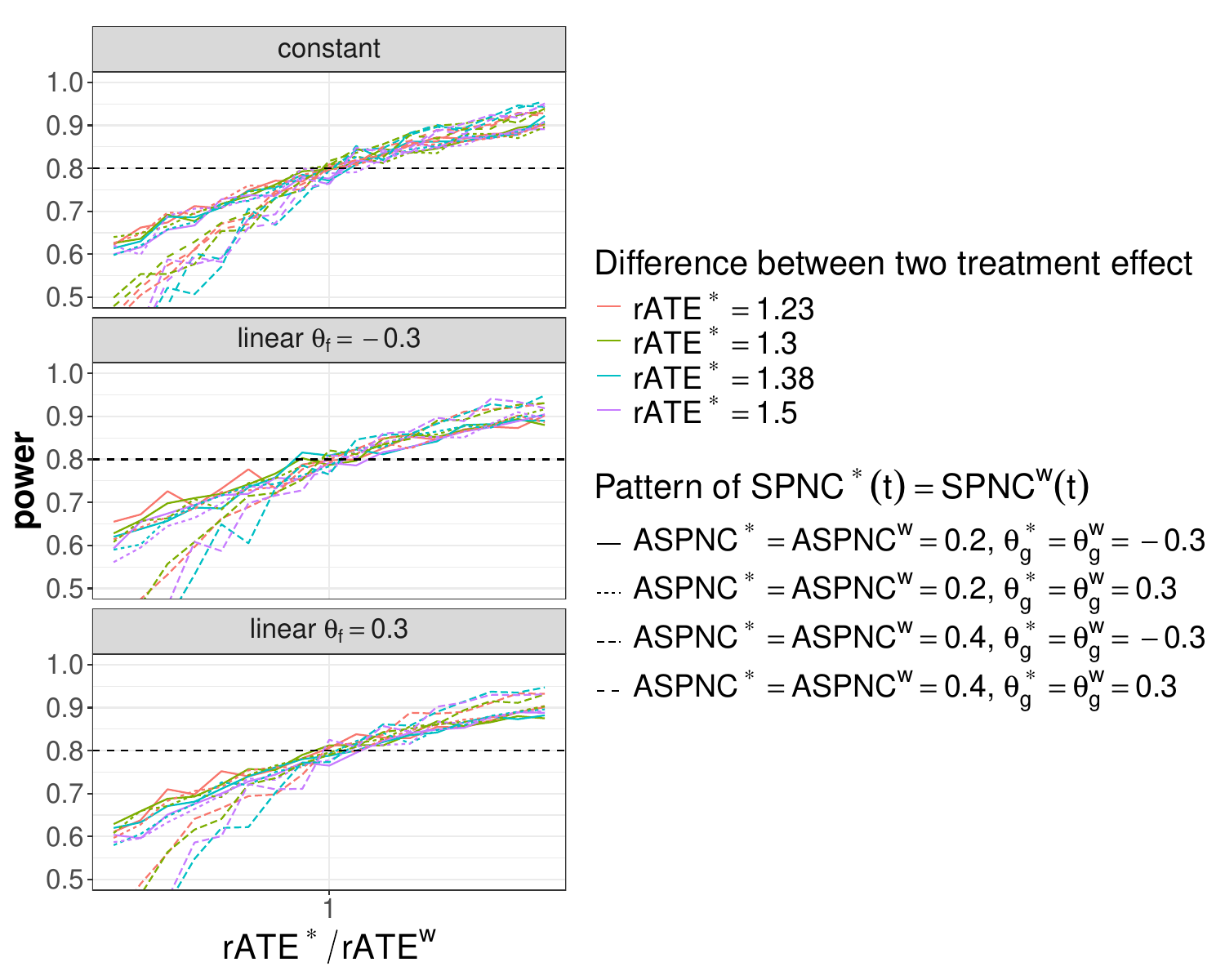}
    \caption{Power when (WA-a) is violated in that $ \rate^\w \neq \text{rATE}^*$ while the working $\mee_k^\w(t)$ is specified to be of the same time-varying pattern as $\mee_k^*(t)$. $\theta_f$ captures the slope of $\mee_k^*(t)$ and the details can be found in \Cref{box:detail-gm-0-part2}. ``ASPNC'' in the figure legend denotes $\aspn$.}
    \label{fig:viol miss ATE}
\end{figure}

In the second scenario, where the pattern of $\mee_k(t)$ is misspecified (i.e., the working $\mee_k^\w(t)$ is specified to be of a different time-varying pattern from $\mee_k^*(t)$) while the magnitude is correct (i.e., $\ate_k^\w = \ate_k^*$), we considered linear and constant patterns for both working and true $\mee_k(t)$. The take-away is that specifying constant working $\mee_k^\w(t)$ always yields adequate power, regardless of whether the true $\mee_k^*(t)$ is constant or not; specifying a linear working $\mee_k^\w(t)$ sometimes leads to inadequate power. The detailed result is as follows.

When the true $\mee_k^*(t)$ is linear in $t$ but the working $\mee_k^\w(t)$ is assumed to be constant, and the ATE for each treatment level is correctly specified ($\ate_i^\w = \ate_i^*$ for all $i \in \{1, 2, \ldots, K\}$), using a constant working $\mee^\w$ yields adequate power (\Cref{fig:mee_working_constant}). We also examine a more relaxed setting where the $\ate_k$ is misspecified but the effect difference ($\rate$) is correctly specified (i.e., $\ate^\w_k \neq \ate_k^*$ but $\rate^\w = \rate^*$). Our results show that correctly specifying the effect difference between the two treatment effects alone is not a sufficient criterion to yield adequate power in binary outcome MRT and one needs to specify $\ate$ further correctly. When $\ate_i^\w > \ate_i^*$ the MRT is over-powered, whereas when $\ate_i^\w < \ate_i^*$ the MRT is under-powered (\Cref{fig:mee_working_constant_mis_ATE}). This result differs from the continuous outcome case, where the magnitude of the ATE does not need to be correctly specified as long as the effect difference is correctly specified to yield adequate power.

When the true $\mee_k^*(t)$ is constant, but the working $\mee_k^\w(t)$ is linear in $t$, we considered two cases: one where $\mee_2^\w(t) - \mee_1^\w(t) = \mee_2^*(t) - \mee_1^*(t)$ (specifically, $\mee_1^\w(t)$ and $\mee_2^\w(t)$ are parallel, i.e., $\theta_{f2} = 0$ with $\theta_{f2}$ defined in \Cref{box:detail-gm-0-part2}), and a second case where $\mee_2^\w(t) - \mee_1^\w(t) \neq \mee_2^*(t) - \mee_1^*(t)$ yet $\rate^\w = \rate^*$ (specifically, $\mee_1^\w(t)$ and $\mee_2^\w(t)$ are not parallel, i.e., $\theta_{f2} \neq 0$). In both cases the MRT will be under-powered (\Cref{fig:mee_working_linear_parallel slope}, \Cref{fig:mee_working_mis_slope}). Using a parallel slope input might alleviate the under-powering problem by a bit, but the power is decreasing as the $\ate$ increases. 

When both true $\mee^*_k(t)$ and working $\mee^\w_k(t)$ are linear in $t$, we consider the two cases: one where the true $\mee^*_1(t)$ and $\mee^*_2(t)$ are parallel but the working $\mee^\w_1(t)$ and $\mee^\w_2(t)$ are not parallel, and a second case where the true $\mee^*_1(t)$ and $\mee^*_2(t)$ are not parallel but the working $\mee^\w_1(t)$ and $\mee^\w_2(t)$ are parallel. In the first case, the MRT is always under-powered (\Cref{fig:mee_true_parallel_slope}). In the second case, the MRT is always over-powered (\Cref{fig:mee_working_parallel slope}).

\begin{figure}[htbp]
    \centering
    \begin{subfigure}[b]{0.49\textwidth}
        \centering
        \includegraphics[width=\textwidth]{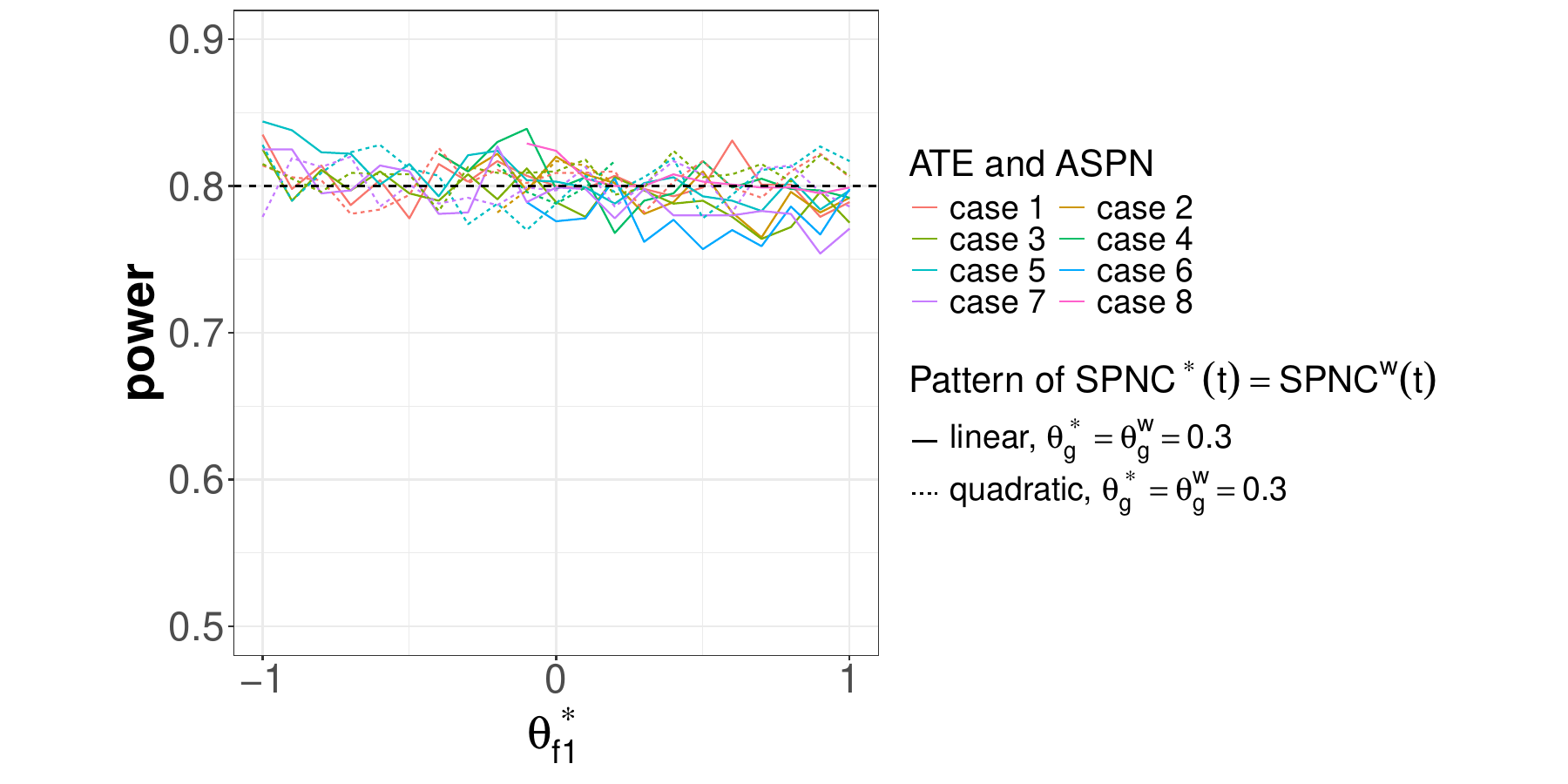}
        \caption{}
        \label{fig:mee_working_constant}
    \end{subfigure}
    \hfill
    \begin{subfigure}[b]{0.49\textwidth}  
        \centering 
        \includegraphics[width=\textwidth]{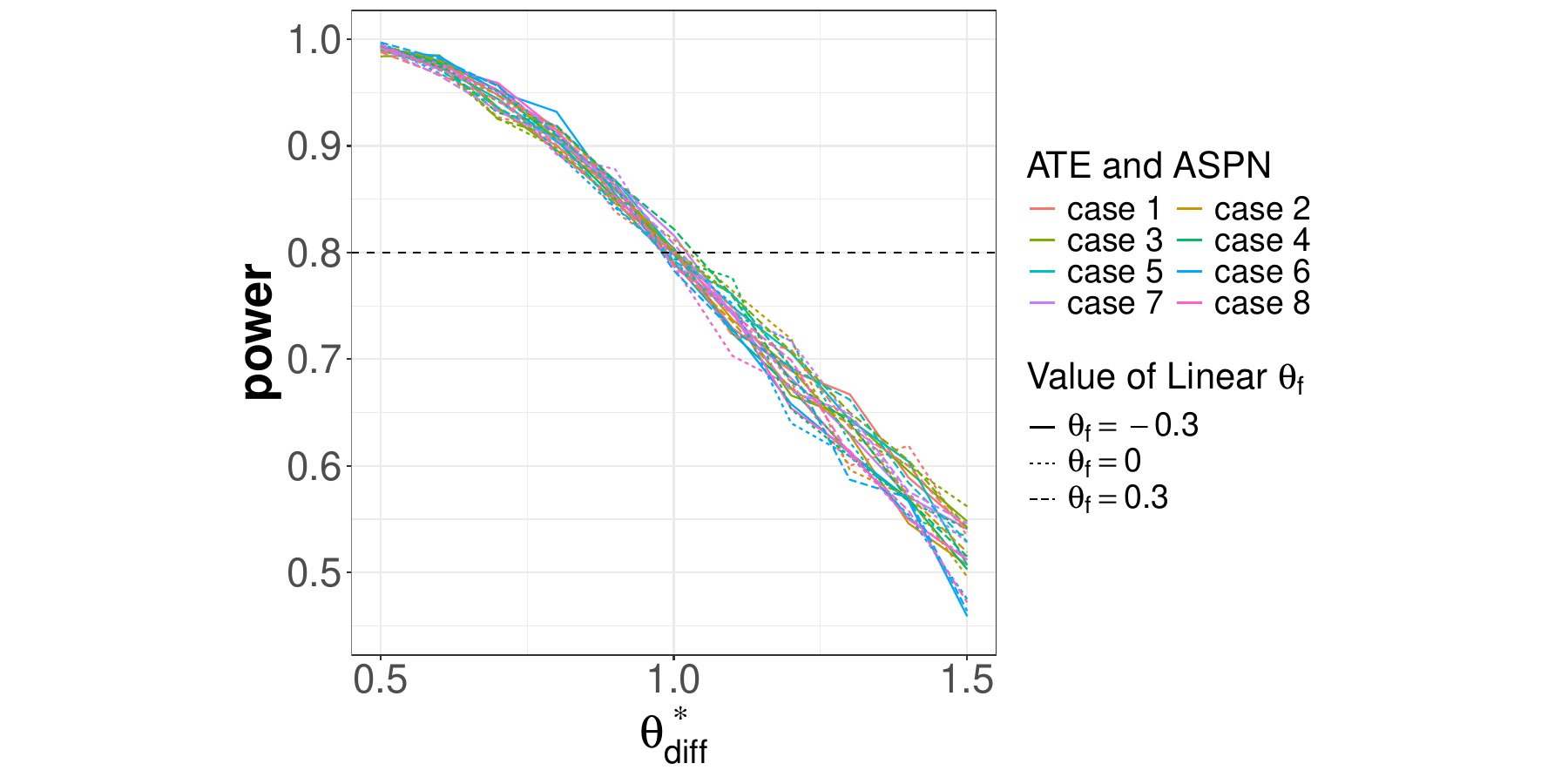}
        \caption{}  
        \label{fig:mee_working_constant_mis_ATE}
    \end{subfigure}
    
    \begin{subfigure}[b]{0.49\textwidth}   
        \centering 
        \includegraphics[width=\textwidth]{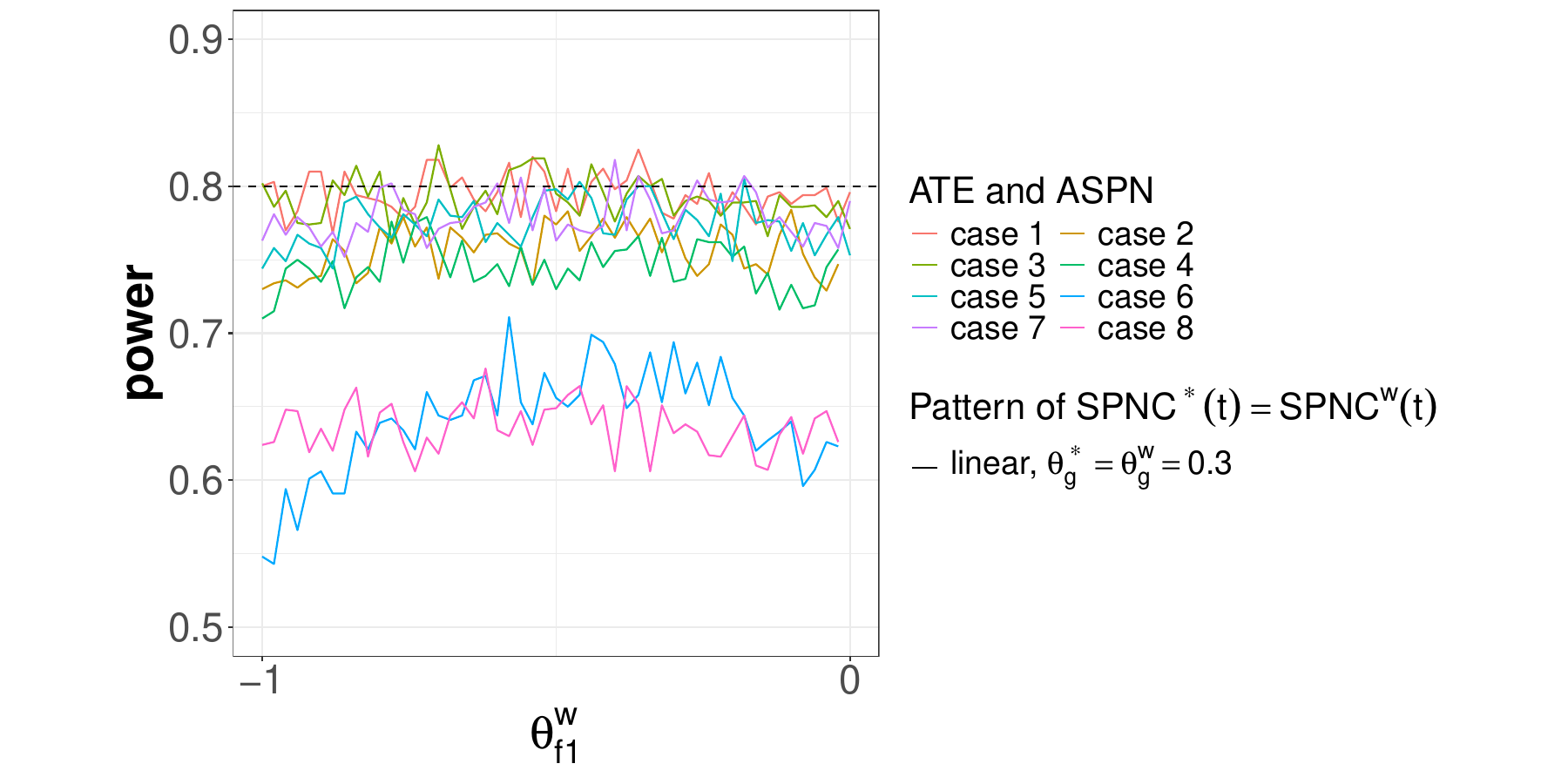}
        \caption{}
        \label{fig:mee_working_linear_parallel slope}
    \end{subfigure}
    \hfill
    \begin{subfigure}[b]{0.49\textwidth}   
        \centering 
        \includegraphics[width=\textwidth]{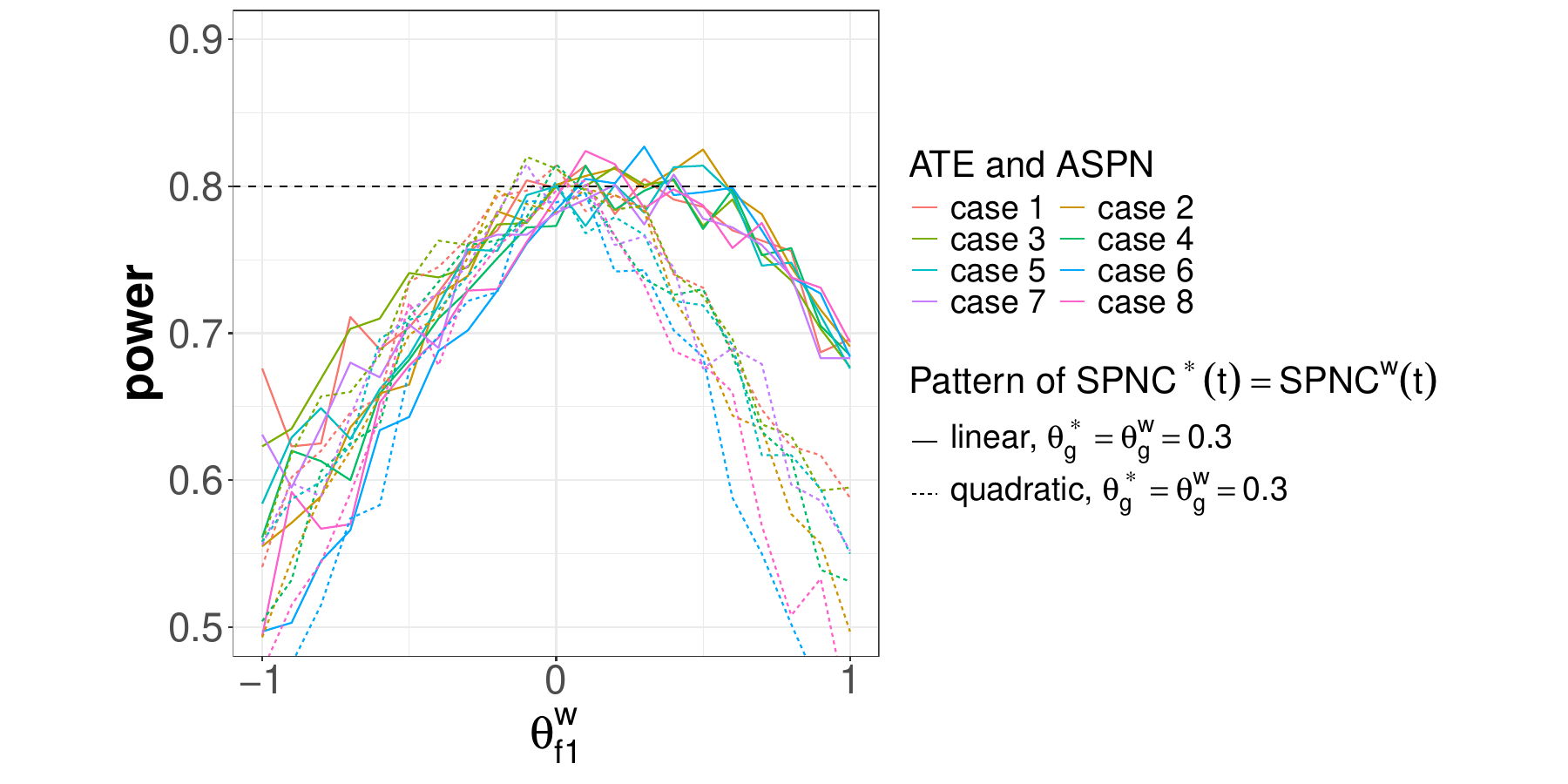}
        \caption{}   
        \label{fig:mee_working_mis_slope}
    \end{subfigure}
    \begin{subfigure}[b]{0.49\textwidth}   
        \centering 
        \includegraphics[width=1\textwidth]{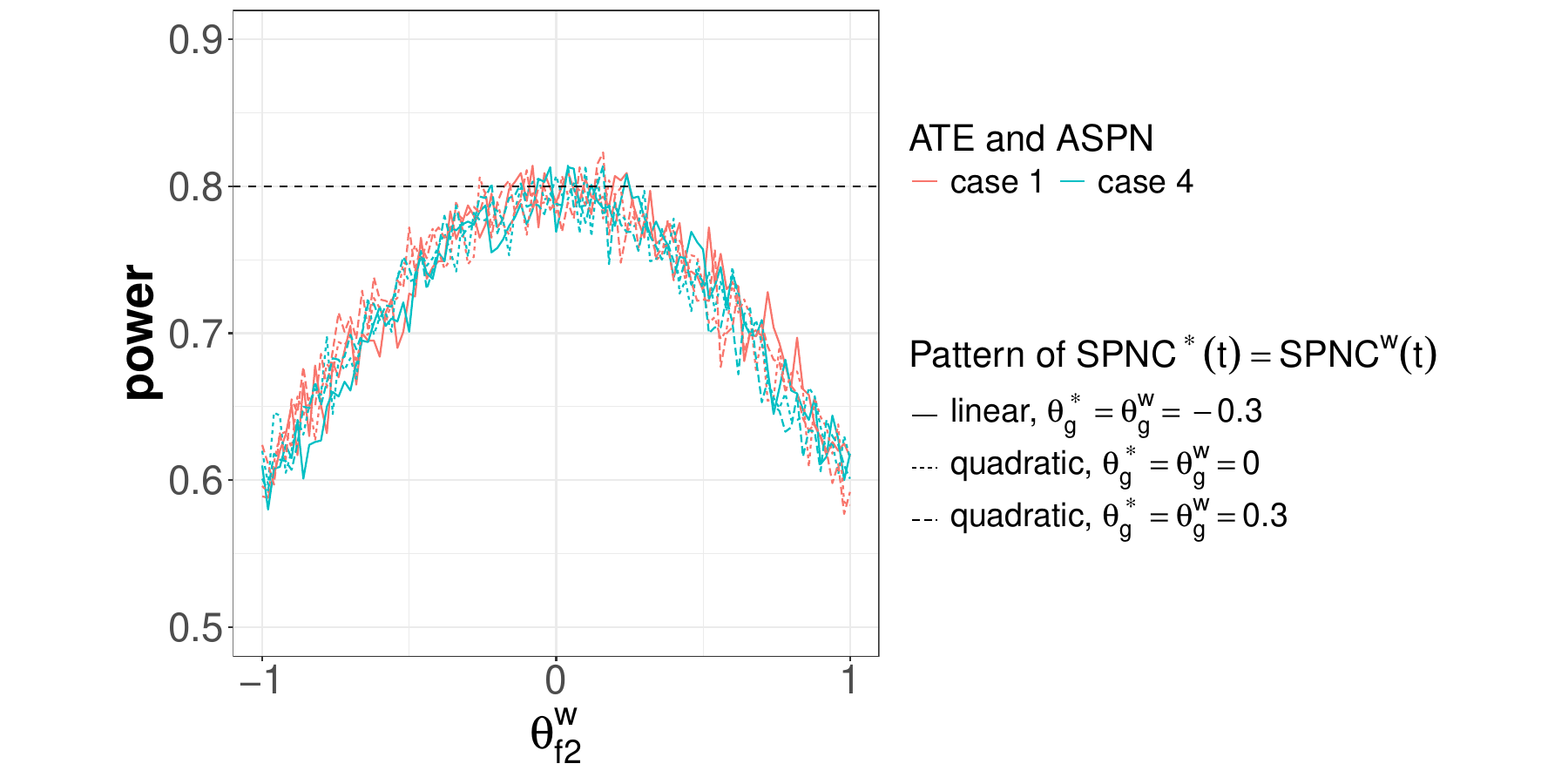}
        \caption{}
        \label{fig:mee_true_parallel_slope}
    \end{subfigure}
    \hfill
    \begin{subfigure}[b]{0.49\textwidth}   
        \centering 
        \includegraphics[width=1\textwidth]{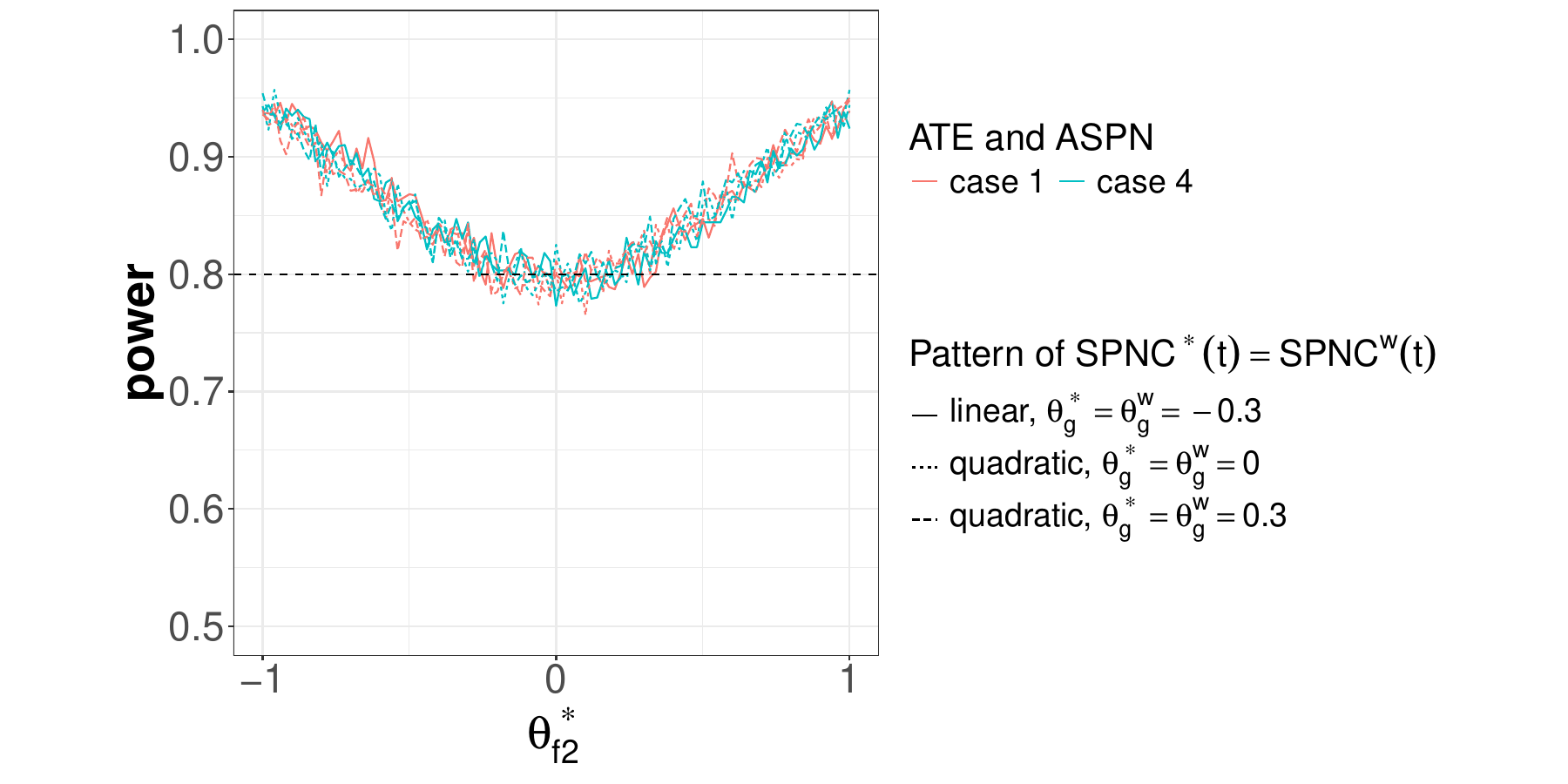}
        \caption{}
        \label{fig:mee_working_parallel slope}
    \end{subfigure}
    
    {\scriptsize case 1 : $\text{r ATE}^* = \text{r ATE}^w = 1.23$, $\text{ASPN}^* = \text{ASPN}^w = 0.2$,\\
    case 2 : $\text{r ATE}^* = \text{r ATE}^w = 1.23$, $\text{ASPN}^* = \text{ASPN}^w = 0.4$,\\
    case 3 : $\text{r ATE}^* = \text{r ATE}^w = 1.33$, $\text{ASPN}^* = \text{ASPN}^w = 0.2$,\\
    case 4 : $\text{r ATE}^* = \text{r ATE}^w = 1.33$, $\text{ASPN}^* = \text{ASPN}^w = 0.4$,\\
    case 5 : $\text{r ATE}^* = \text{r ATE}^w = 1.38$, $\text{ASPN}^* = \text{ASPN}^w = 0.2$,\\
    case 6 : $\text{r ATE}^* = \text{r ATE}^w = 1.38$, $\text{ASPN}^* = \text{ASPN}^w = 0.4$,\\ 
    case 7 : $\text{r ATE}^* = \text{r ATE}^w = 1.5$, $\text{ASPN}^* = \text{ASPN}^w = 0.2$,\\
    case 8 : $\text{r ATE}^* = \text{r ATE}^w = 1.5$, $\text{ASPN}^* = \text{ASPN}^w = 0.4$ \par}
    \caption{Power when (WA-a) is violated in that the time-varying pattern of $\mee_k^\w(t)$ is different from that of $\mee_k^*(t)$ while $\rate^\w = \rate^*$. \textbf{Panels (a) and (b):} the true $\mee_k^*(t)$ is linear in $t$ but the working $\mee_k^\w(t)$ is constant. \textbf{Panel (a):} $\ate_1^* = \ate_1^\w$ and $\ate_2^* = \ate_2^\w$. \textbf{Panel (b):} $\ate_1^* \neq \ate_1^\w$ and $\ate_2^* \neq \ate_2^\w$, yet $ \rate^* =  \rate^\w$. \textbf{Panels (c) and (d):} the true $\mee_k^*(t)$ is constant but the working $\mee_k^\w(t)$ is linear in $t$. \textbf{Panel (c):} $\mee_2^\w(t) - \mee_1^\w(t) = \mee_2^*(t) - \mee_1^*(t)$. \textbf{Panel (d):} $\mee_2^\w(t) - \mee_1^\w(t) \neq \mee_2^*(t) - \mee_1^*(t)$ yet $\ate^\w = \ate^*$. \textbf{Panels (e) and (f):} both the true $\mee_k^*(t)$ and the working $\mee_k^\w(t)$ are linear in $t$. \textbf{Panel (e):} the true $\mee^*_1(t)$ and $\mee^*_2(t)$ are parallel but the working $\mee^\w_1(t)$ and $\mee^\w_2(t)$ are not parallel. \textbf{Panel (f):} the true $\mee^*_1(t)$ and $\mee^*_2(t)$ are not parallel but the working $\mee^\w_1(t)$ and $\mee^\w_2(t)$ are parallel.}
\end{figure}

\subsection{Power when (WA-b) is violated}
\label{subsec:simulation-violate-b}

Using GM-0, we consider two scenarios where (WA-b) can be violated: the magnitude of the average success probability under the null is misspecified ($\aspn^* \neq \aspn^\w$) but the time-varying pattern of the success probability null curve is correct (e.g., both $\spnc^*(t)$ and $\spnc^\w(t)$ are linear in $t$), or the magnitude is correct ($\aspn^* = \aspn^\w$) but the time-varying pattern is incorrect ($\spnc^*(t) \neq \spnc^\w(t)$).

In the first scenario, misspecifying the magnitude of $\aspn$ can lead to either an under-power or over-power MRT (\Cref{fig:mis_aspnc}). Specifically, in the case where the working $\aspn^\w > \aspn^*$, the MRT tends to be under-powered. Conversely, in the case where $\aspn^\w <\aspn^*$ the MRT is over-powered. Therefore, in situations where the value of $\aspn$ is unknown, underestimating $\aspn^\w$ is recommended to yield a more conservative estimate. 

When the magnitude of $\spnc(t)$ is correctly specified $(\aspn^\w = \aspn^*)$, but the pattern of $\spnc(t)$ is misspecified, whether the MRT is adequately powered depends on the type of misspecification. If the true $\spnc^*(t)$ follows a linear or quadratic pattern in $t$ but the working $\spnc^\w(t)$ is constant, the MRT is still adequately powered (\Cref{fig:gt_working_constant}). However, when the true $\spnc^*(t)$ is constant but the working $\spnc^\w(t)$ is linear or quadratic in $t$, the MRT power decreases as $\theta_g$ is farther from 0 (\Cref{fig:gt_true_constant}). In the case when the pattern of $\spnc(t)$ is unknown, the desired power is guaranteed by our sample size calculator as long as the input of $\spnc^\w$ is constant. When the input of $\spnc^\w$ is not constant, it can result in an under-powered MRT. These results differ from those observed in the continuous outcome setting, where the input effect under no treatment ($\spnc$ and $\aspn$) does not affect the power of MRT. 

We further consider a variation of GM-0 with a ``weekend'' effect on the true mean proximal outcome, i.e., $E(Y_{it} \mid A_{it}, H_{it}) = e^{\one(A_{it} = 1) f_t^T \beta_1^* + \one(A_{it} = 2) f_t^T \beta_2^* + g_t^T \alpha^* + \text{weekend}_t \theta_\w }$, where $\text{weekend}_t = 1$ if the decision point is on a weekend and 0 otherwise. $\theta_\w$ captures the strength of the weekend effect, and $\theta_\w = 0$ implies that $\spnc$ is correctly specified ($\spnc^\w(t) = \spnc^*(t)$; note that we use a linear $\spnc^\w(t)$). 

We consider two settings: one where both the weekend effect and the magnitude ($\aspn$) are incorrectly specified, and one where the weekend effect is incorrectly specified, but $\aspn$ is correctly specified. When the $\aspn^\w$ is incorrectly specified, the weekend effect results in either an under-powered or over-powered MRT depending on the direction of the effect (\Cref{fig:gt_weekend true}). When the weekend effect decreases the value of $\spnc(t)$ ($\theta_\w<0$), the MRT is under-powered. In contrast, when the weekend effect increases the value of $\spnc(t)$ ($\theta_\w > 0$), the MRT is over-powered \Cref{fig:gt_weekend true correct aa}. In contrast, when the $\aspn^\w$ is correctly specified, the MRT is adequately powered, even when using a constant $\spnc^\w$.

\begin{figure}[htbp]
    \centering
    \begin{subfigure}[b]{0.475\textwidth}
        \centering
        \includegraphics[width=\textwidth]{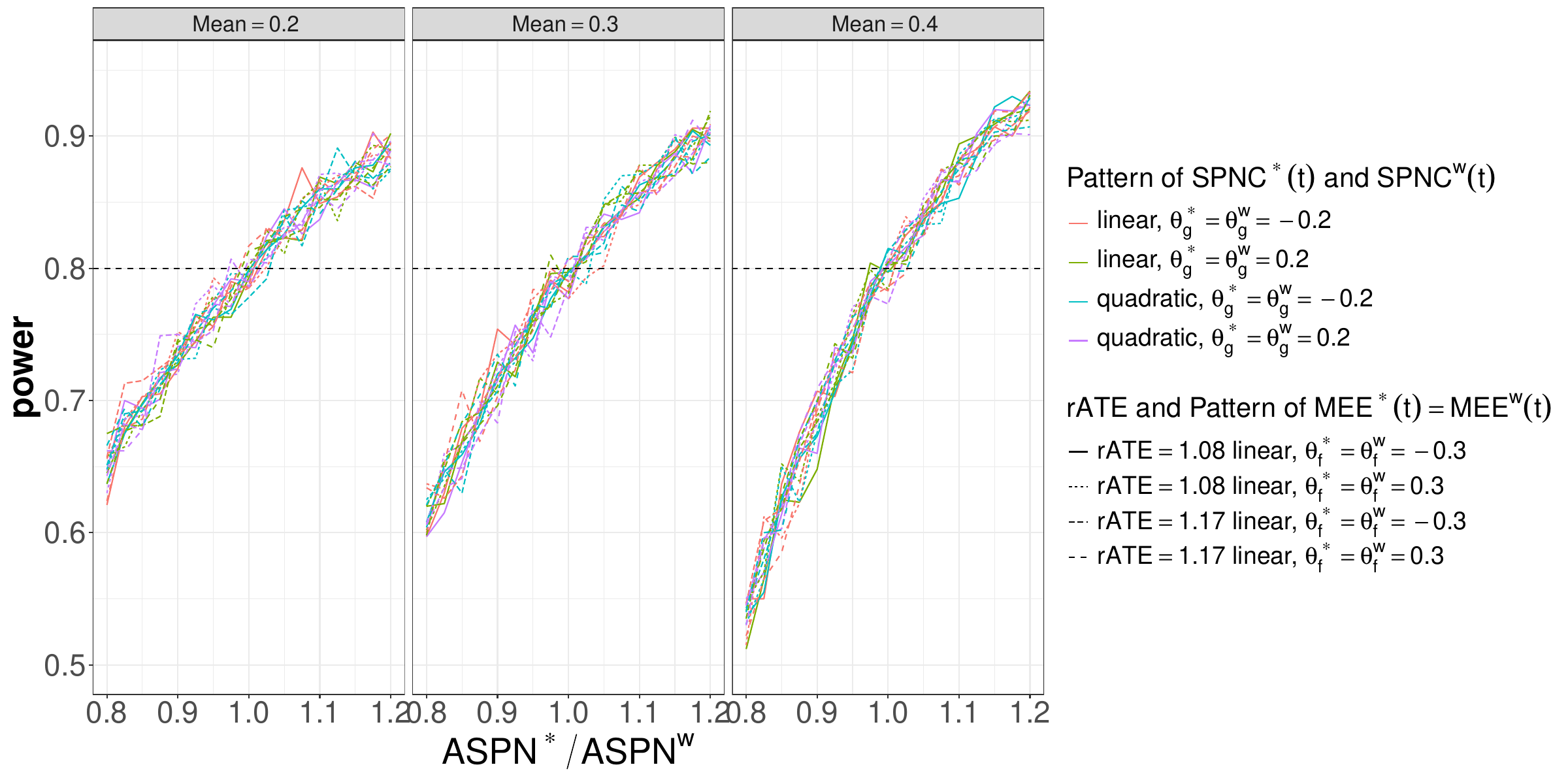}
        \caption{}
        \label{fig:mis_aspnc}
    \end{subfigure}
    \hfill
    \begin{subfigure}[b]{0.475\textwidth}  
        \centering 
        \includegraphics[width=\textwidth]{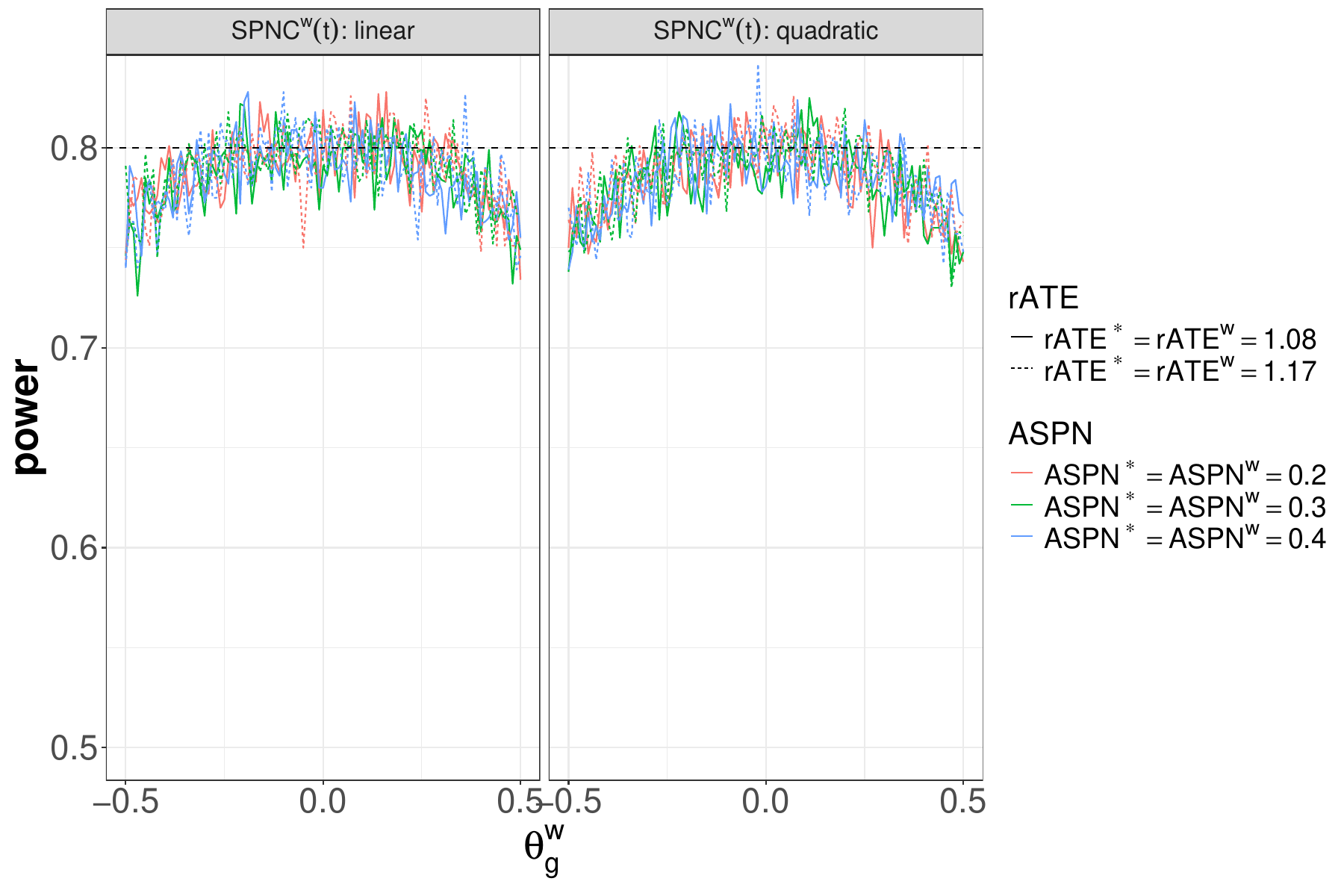}
        \caption{}  
        \label{fig:gt_true_constant}
    \end{subfigure}

    \begin{subfigure}[b]{0.475\textwidth}   
        \centering 
        \includegraphics[width=\textwidth]{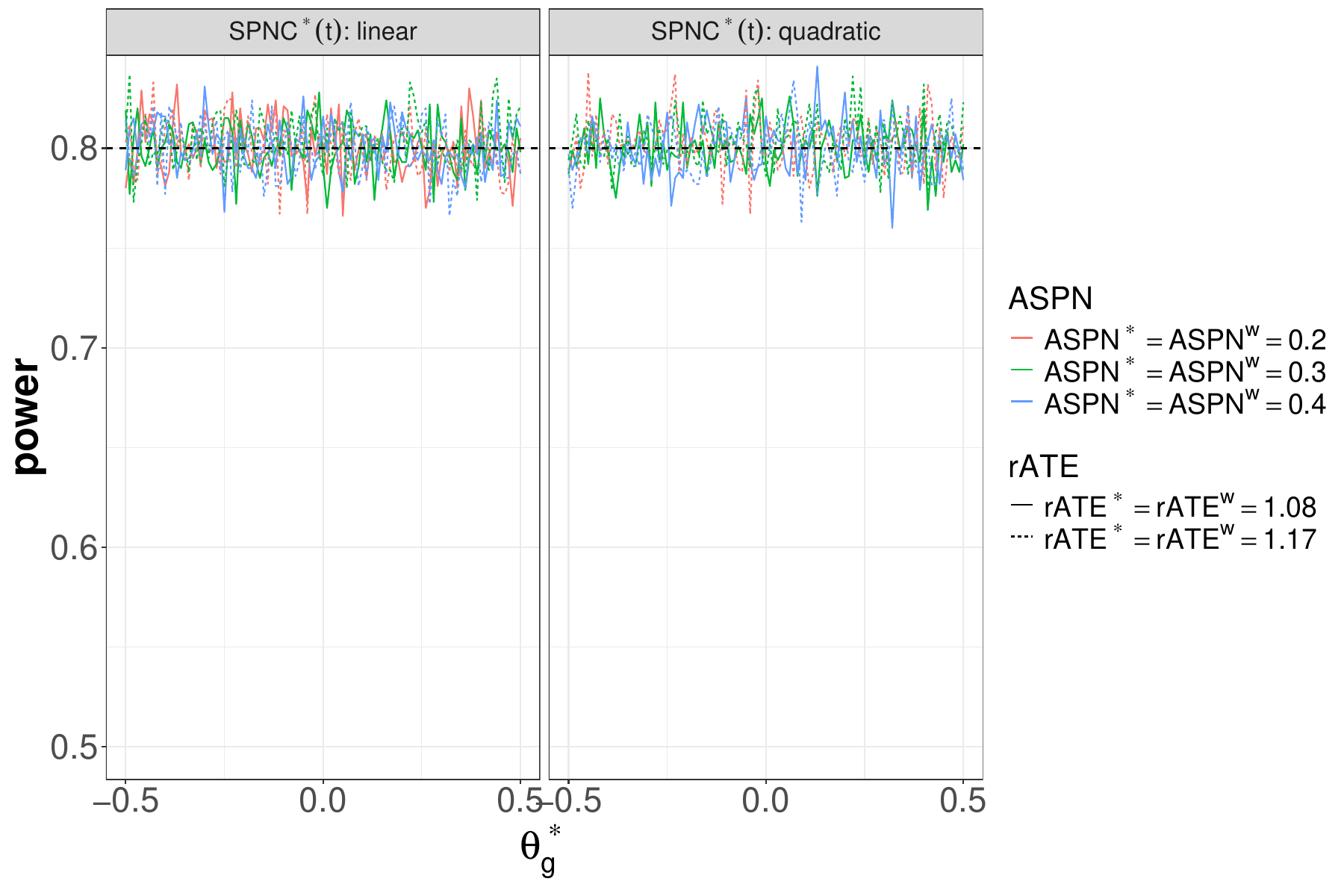}
        \caption{}
        \label{fig:gt_working_constant}
    \end{subfigure}
    \caption{Power when (WA-b) is violated. \textbf{Panels (a):} the magnitude of the average-over-time mean outcome is misspecified ($\aspn^* \neq \aspn^\w$) but the time-varying pattern of the mean outcome is correct (e.g., both $\spnc^*(t)$ and $\spnc^\w(t)$ are linear in $t$). \textbf{Panel (b):} the working $\spnc^\w(t)$ is linear or quadratic in $t$ but the true $\spnc^*(t)$ is constant. \textbf{Panel (c):} the working $\spnc^\w(t)$ is constant but the true $\spnc^*(t)$ is linear or quadratic in $t$. The facet labels ``Mean'' in Panel (a) give the value of $\aspn^*$.}
\end{figure}

\begin{figure}[htbp]
    \centering
    \begin{subfigure}[b]{0.44\textwidth}
        \centering
        \includegraphics[width=\textwidth]{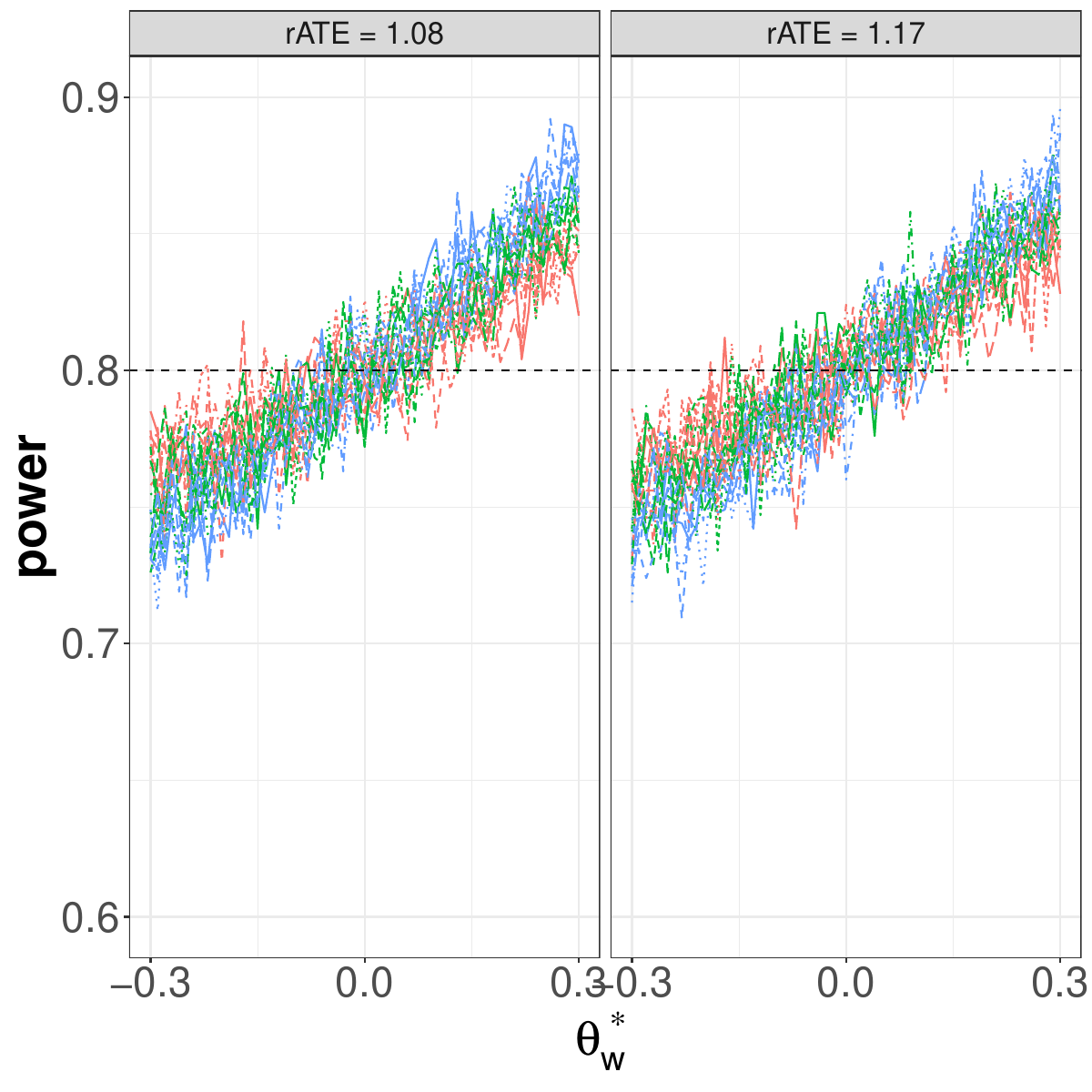}
        \caption{}
        \label{fig:gt_weekend true}
    \end{subfigure}
    \hfill
    \begin{subfigure}[b]{0.54\textwidth}  
        \centering 
        \includegraphics[width=\textwidth]{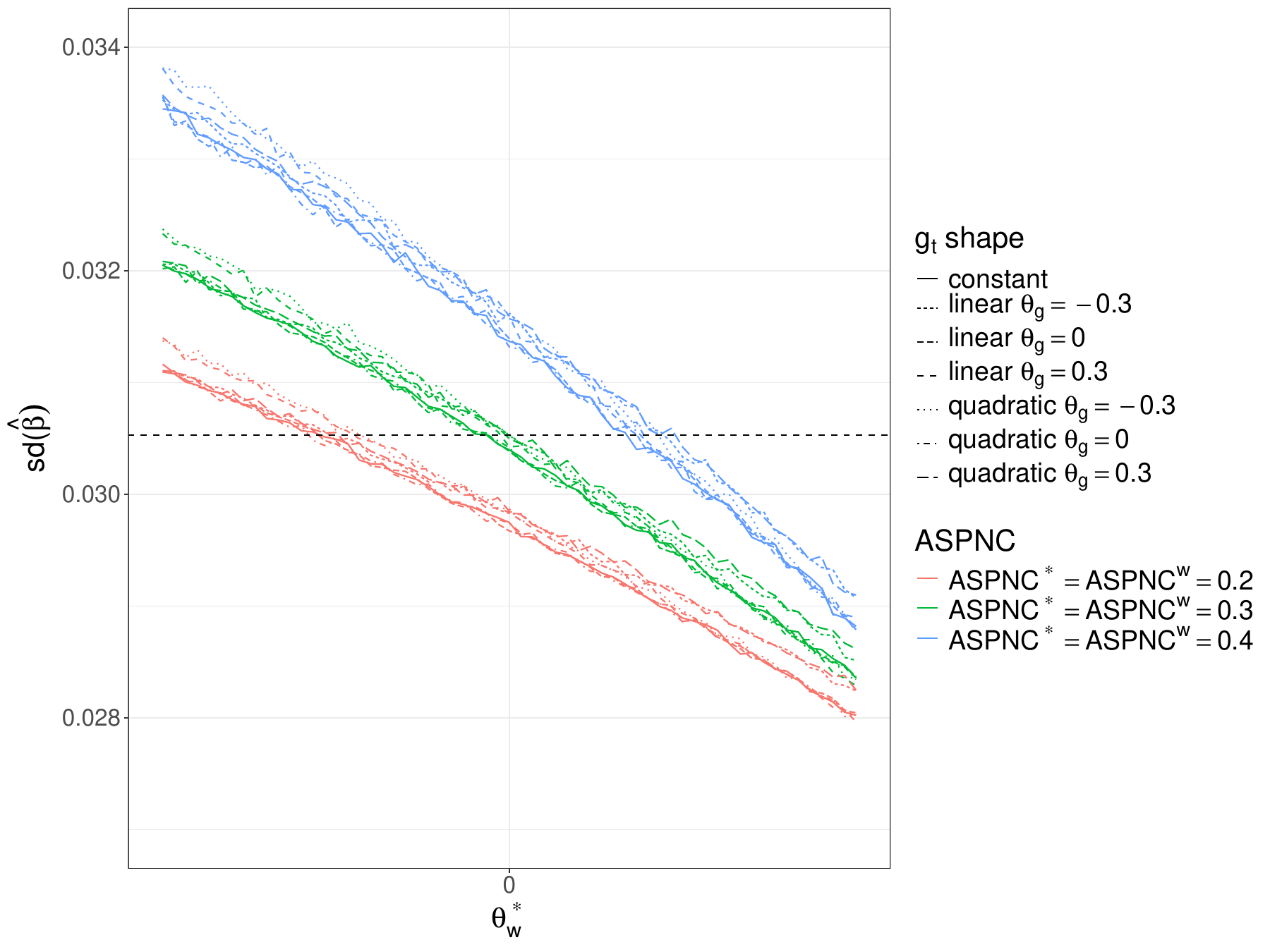}
        \caption{}  
        \label{fig:gt_weekend true se}
    \end{subfigure}
    
    \begin{subfigure}[b]{0.44\textwidth}
        \centering
        \includegraphics[width=\textwidth]{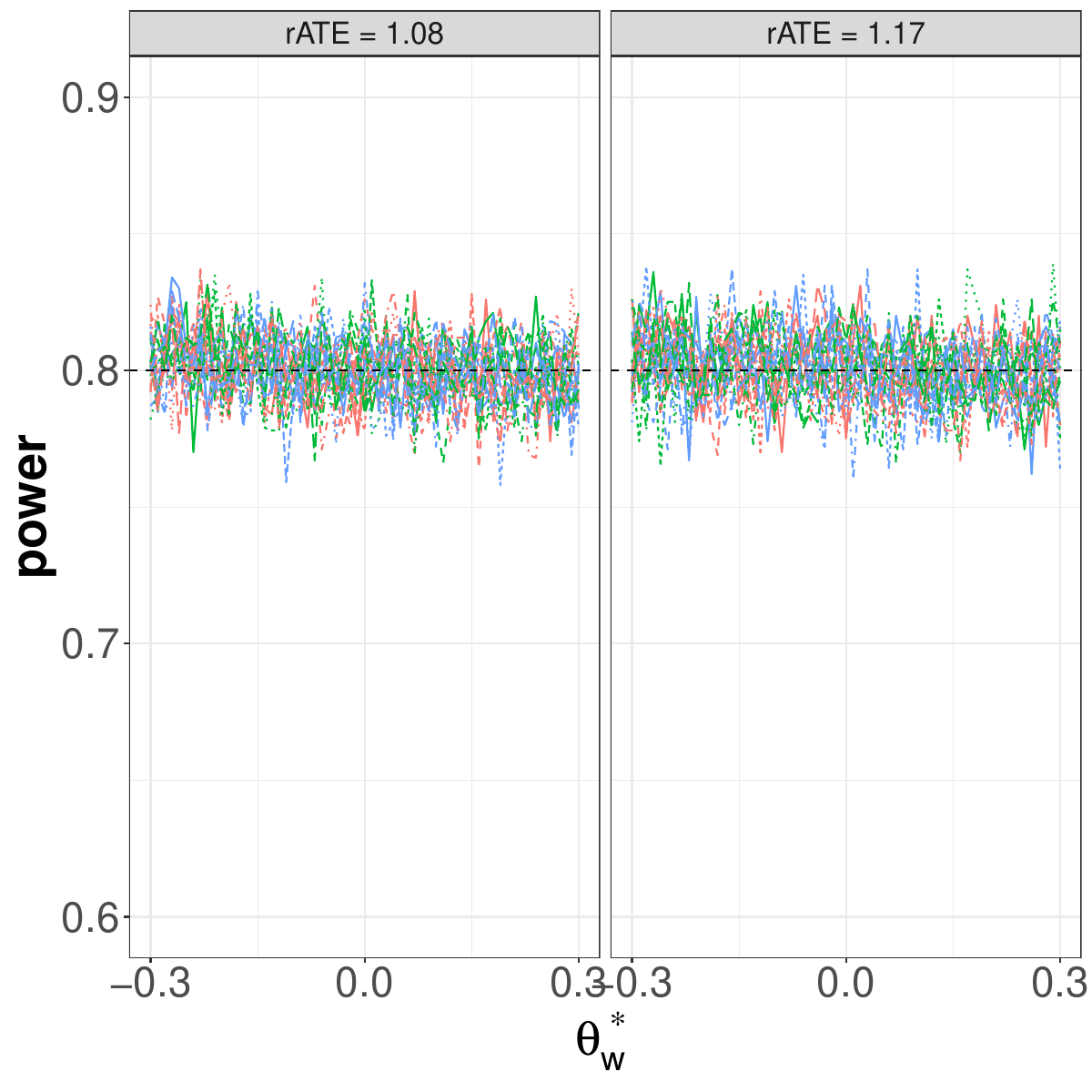}
        \caption{}
        \label{fig:gt_weekend true correct aa}
    \end{subfigure}
    \hfill
    \begin{subfigure}[b]{0.54\textwidth}  
        \centering 
        \includegraphics[width=\textwidth]{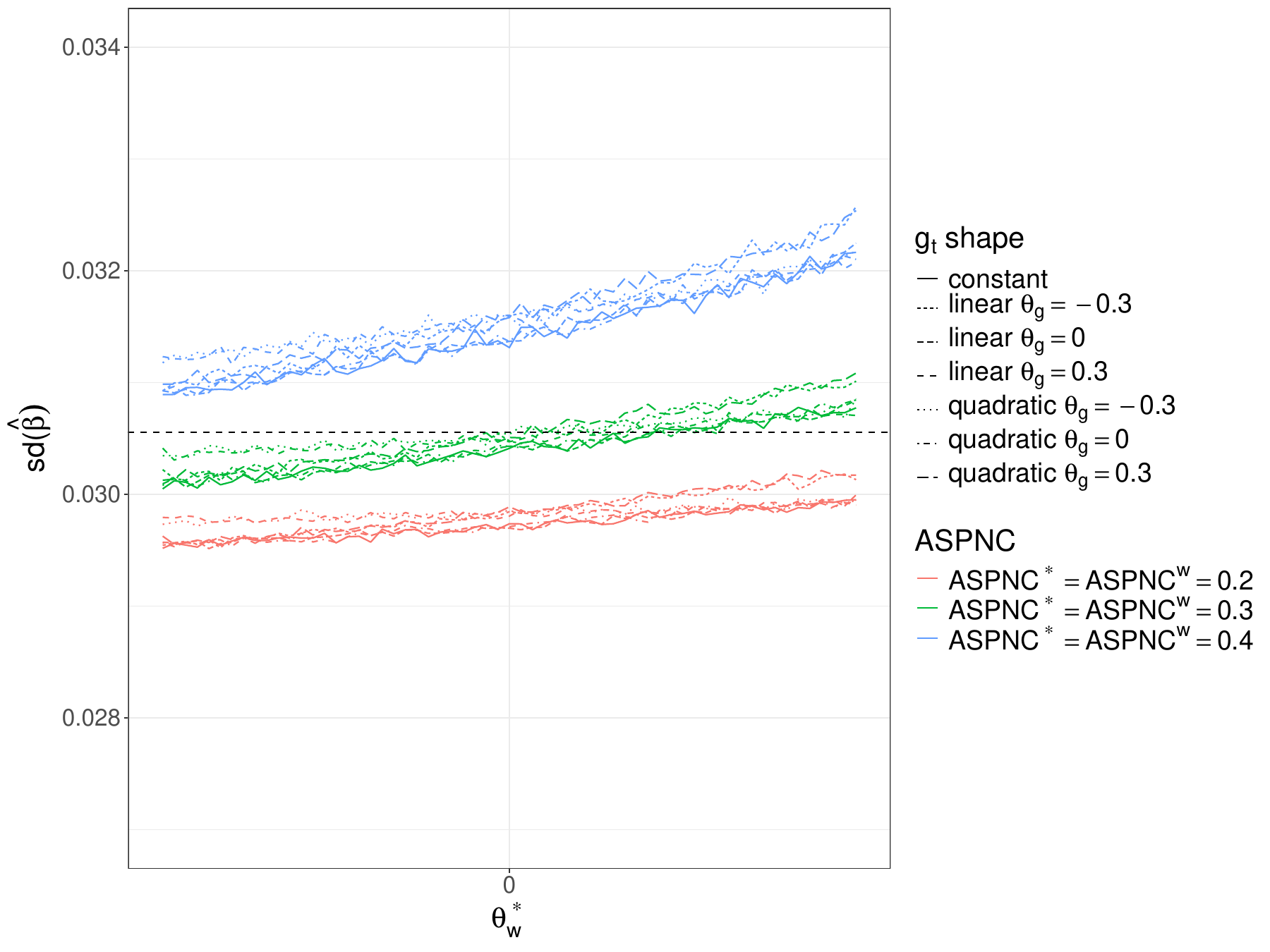}
        \caption{}  
        \label{fig:gt_weekend true se correct aa}
    \end{subfigure}
    \caption{Power when (WA-b) is violated in that there is a ``weekend'' effect on the true mean proximal outcome. \textbf{Panels (a):} the power when $\aspn^\w$ is misspecified. \textbf{Panel (b):} the standard deviation of $\hat\beta$ for the same setting as Panel (a), explaining the low power in Panel (a). \textbf{Panels (c):} the power when $\aspn^\w$ is correctly specifed. \textbf{Panel (d):} the standard deviation of $\hat\beta$ for the same setting as Panel (c), explaining the adequate power in Panel (c).}
\end{figure}

\subsection{Power when (WA-c) is violated}
\label{subsec:simulation-violate-d}

Using GM-0, we consider two scenarios where (WA-c) can be violated: the magnitude of $\aaa$ is misspecified (so that $\aaa^* \neq \aaa^\w$ but the time-varying patterns of $\tau^*(t)$ and $\tau^\w(t)$ are the same) or the pattern of $\tau(t)$ is incorrect (so that $\aaa^* = \aaa^\w$ but the time-varying patterns of $\tau^*(t)$ and $\tau^\w(t)$ are different). Recall that $\aaa$ is the averaged-over-time expected availability defined in \eqref{eq:AA}.

In the first scenario, we set $\tau^\w(t)$ and $\tau^*(t)$ to be constants in $t$, so that $\tau^\w(t) = \aaa^\w$ and $\tau^*(t) = \aaa^*$ for all $t \in [T]$. We set $\aaa^\w = 0.3$ and varied $\aaa^*$ from 0.1 to 1. The MRT is over-powered if $\aaa^\w < \aaa^*$, and the MRT is under-powered if $\aaa^\w > \aaa^*$ (\Cref{fig:viol_avail_mag}).

In the second scenario, we set $\aaa^* = \aaa^\w = 0.3$ and varied the time-varying pattern of $\tau^*(t)$ and $\tau^\w(t)$, so that one is constant in $t$ and the other is either linear or periodic in $t$. Regardless of which of $\tau^*(t)$ or $\tau^\w(t)$ is constant, the MRT is always adequately powered (\Cref{fig:pattern_viol}).

\begin{figure}[htbp]
    \centering
    \includegraphics[width=1\textwidth]{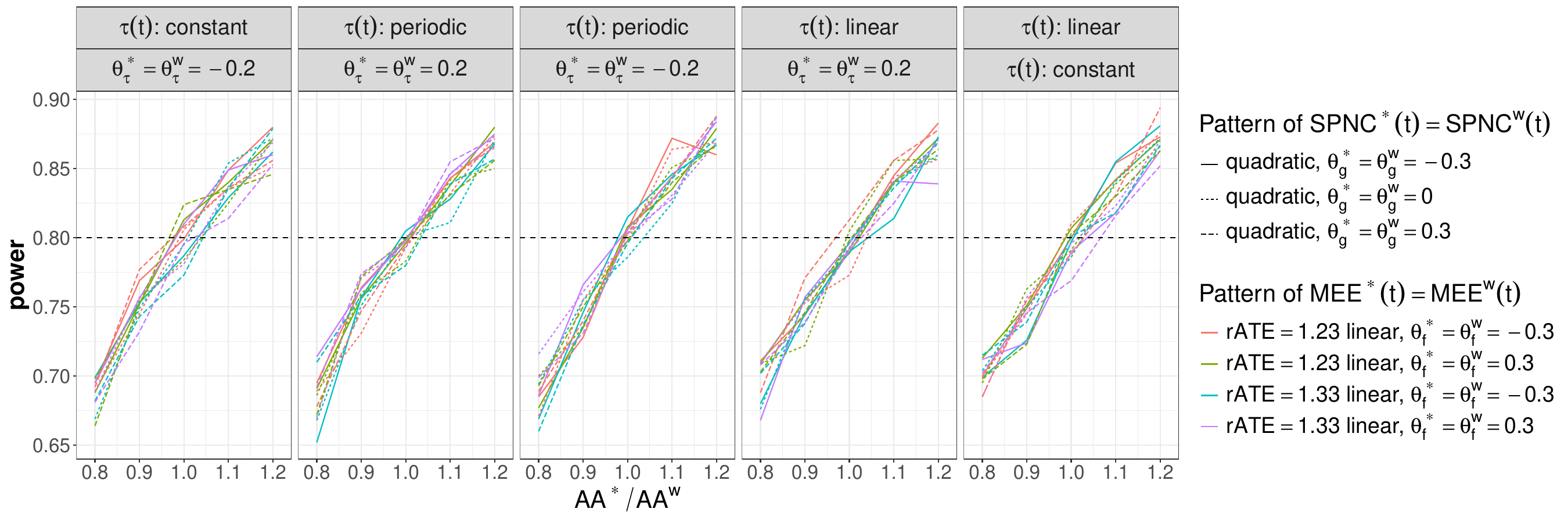}
    \caption{Power when (WA-c) is violated in that $\aaa^* \neq \aaa^\w$ but the time-varying patterns of $\tau^*(t)$ and $\tau^\w(t)$ are the same.}
    \label{fig:viol_avail_mag}
\end{figure}

\begin{figure}[htbp]
    \centering
        \vskip\baselineskip
    \begin{subfigure}[b]{1\textwidth}   
        \includegraphics[width=0.9\textwidth]{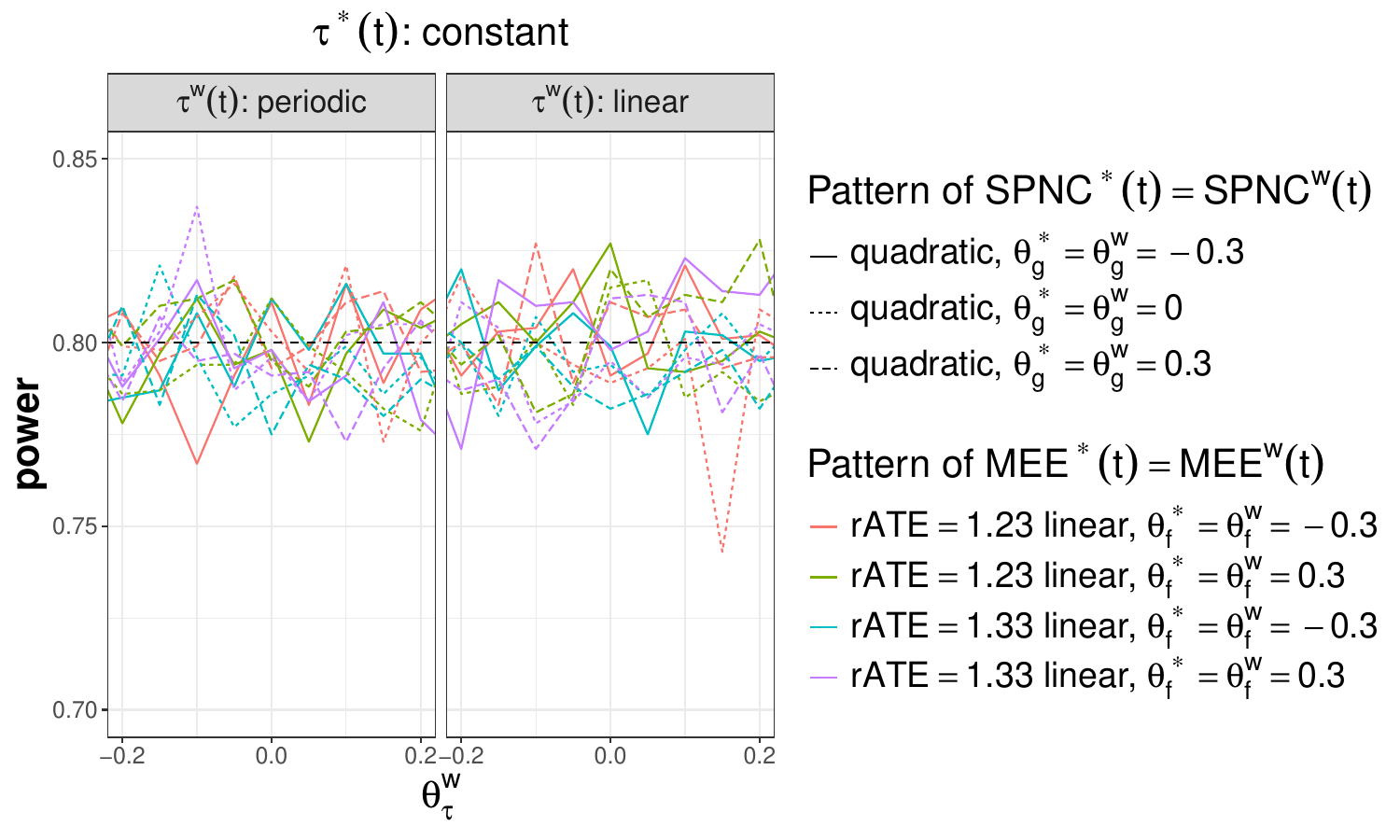}
        \caption{}
    \end{subfigure}
    
    \vskip\baselineskip
    \begin{subfigure}[b]{1\textwidth}   
        \centering
        \includegraphics[width=0.9\textwidth]{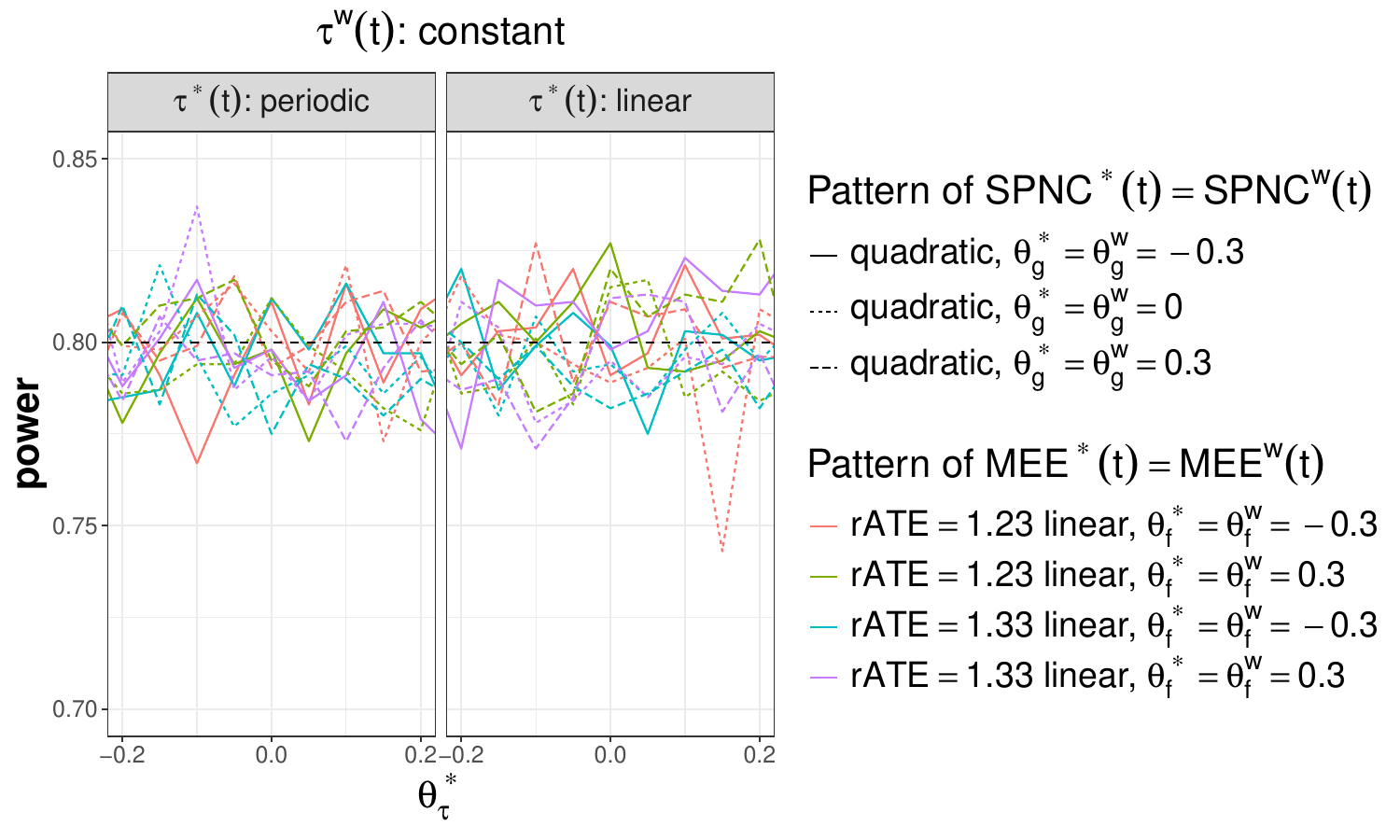}
        \caption{}
    \end{subfigure}
    \hfill
    \caption{Power when (WA-c) is violated in that $\aaa^* = \aaa^\w$ but the time-varying patterns of $\tau^*(t)$ and $\tau^\w(t)$ are different. \textbf{Panels (a):} when $\tau^*(t)$ is constant and $\tau^\w(t)$ is either periodic or linear in $t$. \textbf{Panels (b):} when when $\tau^\w(t)$ is constant and $\tau^*(t)$ is either periodic or linear in $t$.}
    \label{fig:pattern_viol}
\end{figure}

\subsection{Power when (WA-d) is violated}
\label{subsec:simulation-violate-e}

Using GM-SC, (WA-d) is violated as long as $\nu_1 \neq 0$. Recall that $\nu_1$ parameterizes the amount of serial correlation in the proximal outcomes. We consider three ways to specify $\spnc^\w: \spnc^\w(t) = \spnc^*(t)$ (thus $\aspn^\w = \aspn^*$), $\spnc^\w(t) \neq \spnc^*(t)$ but $\aspn^\w = \aspn^*$, and $\spnc^\w(t) \neq \spnc^*(t)$ and $\aspn^\w \neq \aspn^*$. \Cref{fig:sc_viol} shows that as long as $\aspn^\w = \aspn^*$, the MRT will be adequately powered even when the pattern of $\spnc^{\w}(t)$ is misspecified. When both $\spnc^\w$ and $\aspn^\w$ are incorrect, a larger serial correlation effect ($\nu_1$) will increase $\aspn^*$, which results in a higher power.

\begin{figure}[htbp]
    \centering
    \includegraphics[width=.8\textwidth]{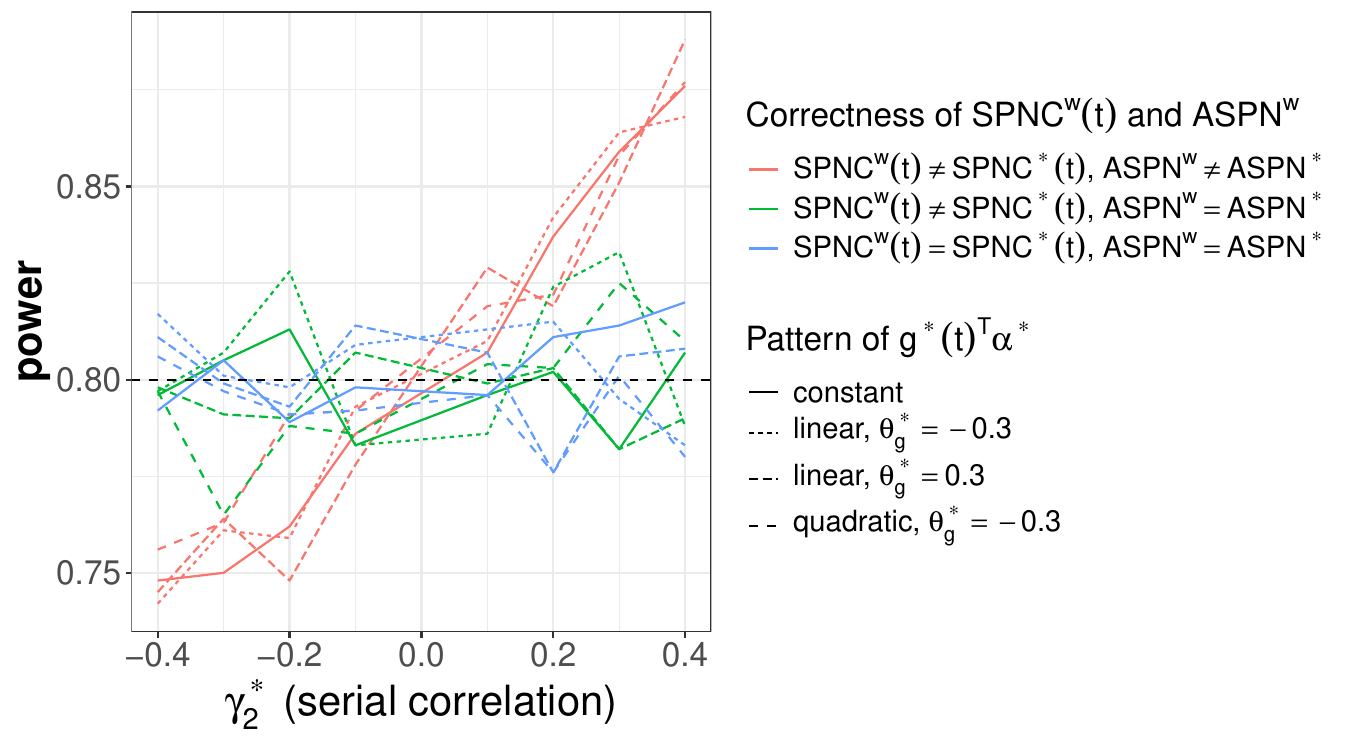}
    \caption{Power when (WA-d) is violated. A larger $\nu_1$ represents a larger magnitude of serial correlation in the proximal outcomes. The horizontal axis of the figure is labeled $\gamma_2^*$, the name of $\nu_1$ in the simulation code.}
    \label{fig:sc_viol}
\end{figure}

\subsection{Power when (WA-e) is violated}
\label{subsec:simulation-violate-f}

Using GM-EA, (WA-e) is violated as long as $\nu_2 \neq 0$, $\nu_3 \neq 0$, or $\nu_4 \neq 0$. Recall that $\nu_2$ and $\nu_3$ parameterize the effect of $A_{t-1}$ on $I_t$, and $\nu_4$ parameterizes the dependence of $I_t$ on $Y_{t-1}$. We considered three ways to specify $\tau^\w(t):\tau^\w(t) = \tau^*(t)$ (thus $\aaa^\w = \aaa^*)$, $\tau^\w(t) \neq \tau^*(t)$ but $\aaa^\w = \aaa^*$, and $\tau^\w(t) \neq \tau^*(t)$ and $\aaa^\w \neq \aaa^*$.
\Cref{subfig:combined_endogeneity} illustrates how endogeneity affected the power of MRT.

In the case where availability depends only on the previous treatment (treatment 1 or treatment 2, represented by $\nu_2$ and $\nu_3$ respectively), as shown in \Cref{subfig:treatment_endogeneity}, the MRT remains adequately powered as long as $\aaa^\w$ is correctly specified, even when the pattern of $\tau^\w(t)$ is misspecified, as evidenced by the nearly constant pattern around 0.8 with no observable gradient shift.

\Cref{subfig:one_treatment_previous_outcome_endogeneity} illustrates how the interactions between the previous outcome ($\nu_4$) and the effect of one of the treatments ($\nu_2$) on the current availability influence the power of MRT. The MRT can be under-powered or over-powered depending on the direction of $\nu_4$. Moreover, when $\aaa^\w$ is incorrect, the power of MRT can also be affected by the direction of $\nu_2$, as indicated by the gradient shift in the plot. However, similar to the result from the previous subfigure, when $\aaa^\w$ is correctly specified, the effect of $\nu_2$ does not affect the power of the MRT as evidenced by the absence of a gradient shift in the plot.

Finally, \Cref{subfig:combined_endogeneity} combines the results of \Cref{subfig:treatment_endogeneity} and \Cref{subfig:one_treatment_previous_outcome_endogeneity} to show the joint effect of previous treatment (color) and previous outcome (slope) on the power of MRT. Consistent with prior findings, endogeneity based on previous treatment ($A_{t-1}$) will not affect the power of MRT as long as $\aaa$ is correctly specified, as evidenced by the lack of gradient shifts. In contrast, endogeneity based on previous outcome can result in MRT being under-powered or over-powered depending on the direction and magnitude of $\nu_4$ as represented by the positive slope trend observed across all three settings.

\begin{figure}[htbp]
    \centering
    \begin{subfigure}[b]{0.44\textwidth}
        \centering
        \includegraphics[width=\textwidth]{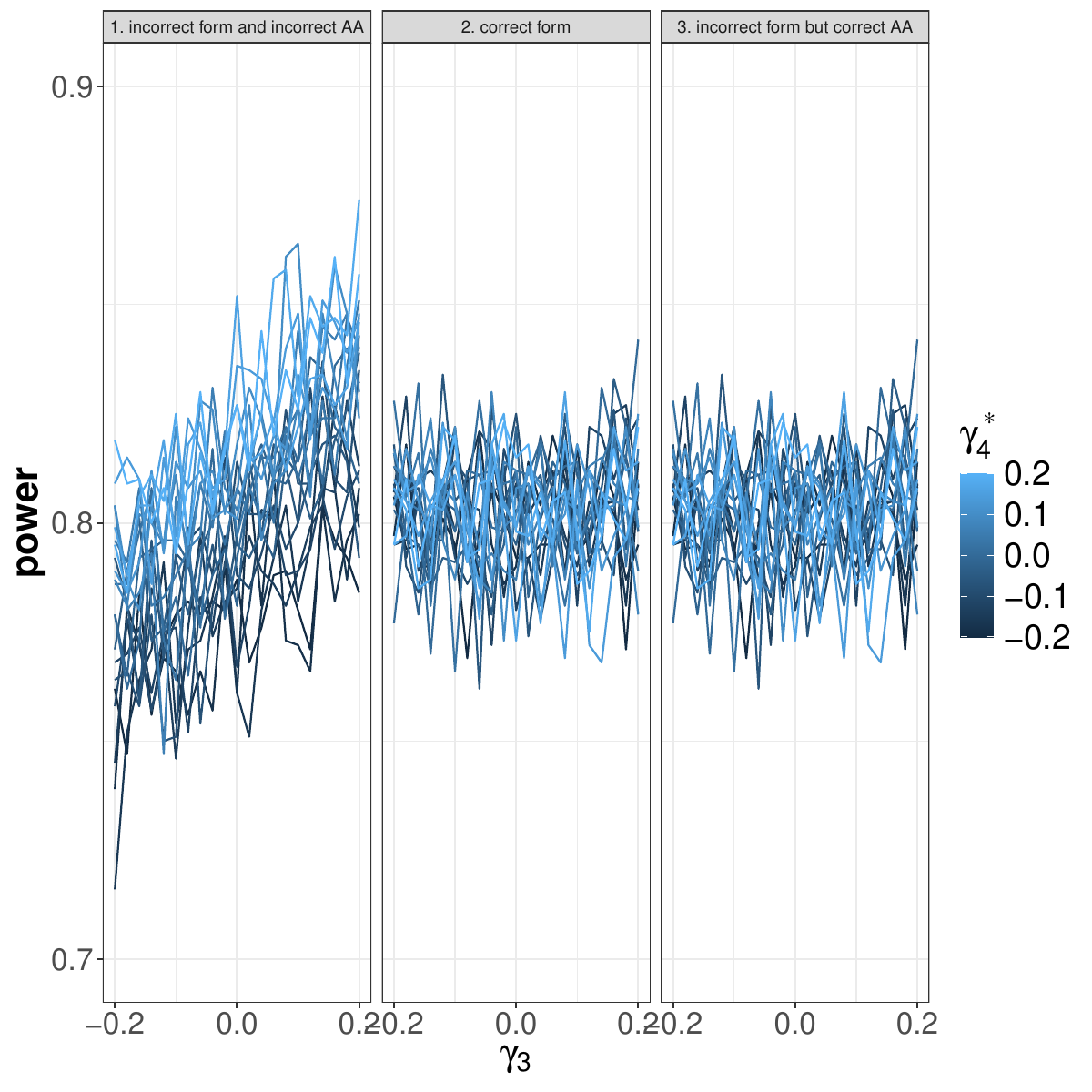}
        \caption{}
        \label{subfig:treatment_endogeneity}
    \end{subfigure}
    \hfill
    \begin{subfigure}[b]{0.44\textwidth}  
        \centering 
        \includegraphics[width=\textwidth]{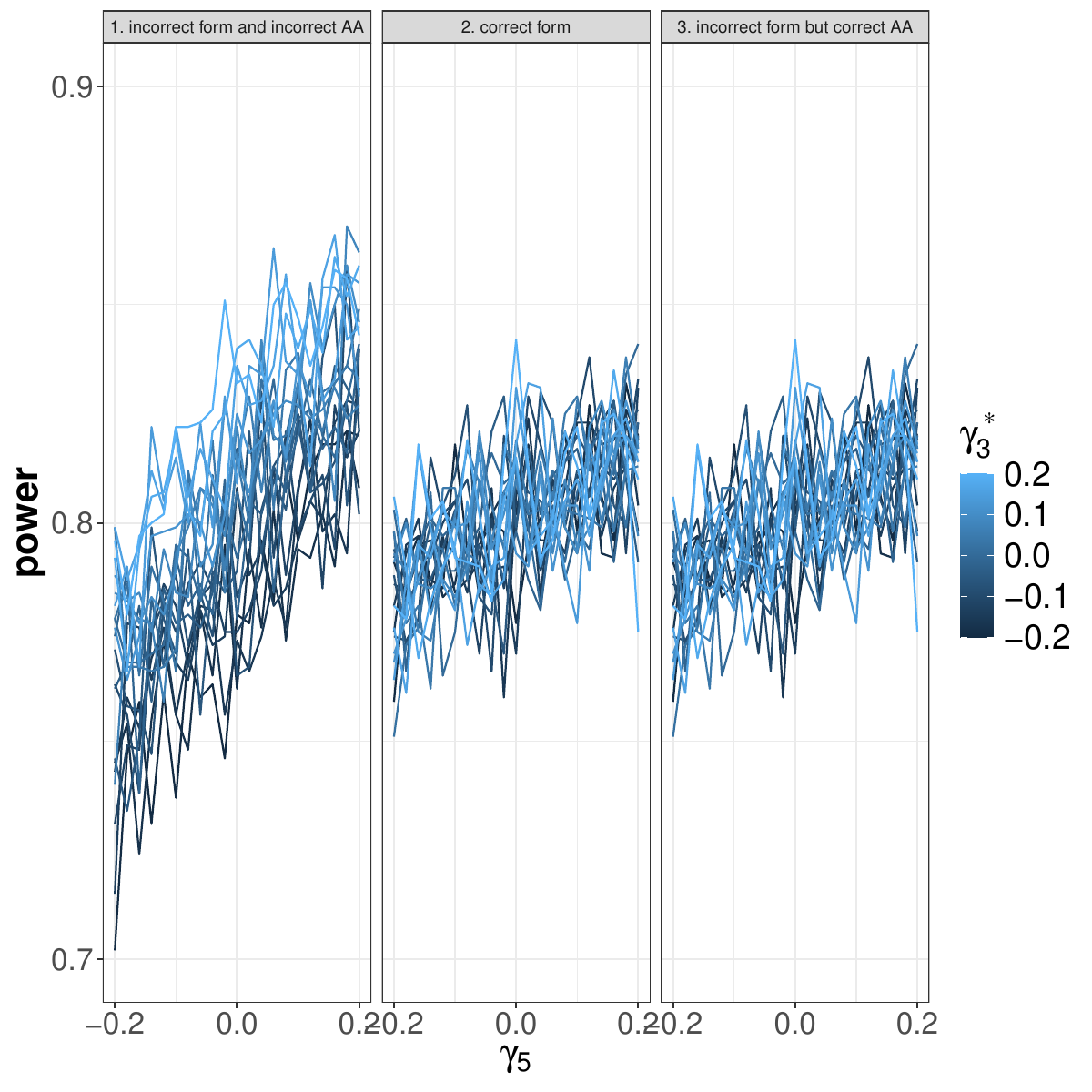}
        \caption{}  
        \label{subfig:one_treatment_previous_outcome_endogeneity}
    \end{subfigure}
    \begin{subfigure}[b]{0.54\textwidth}  
        \centering 
        \includegraphics[width=\textwidth]{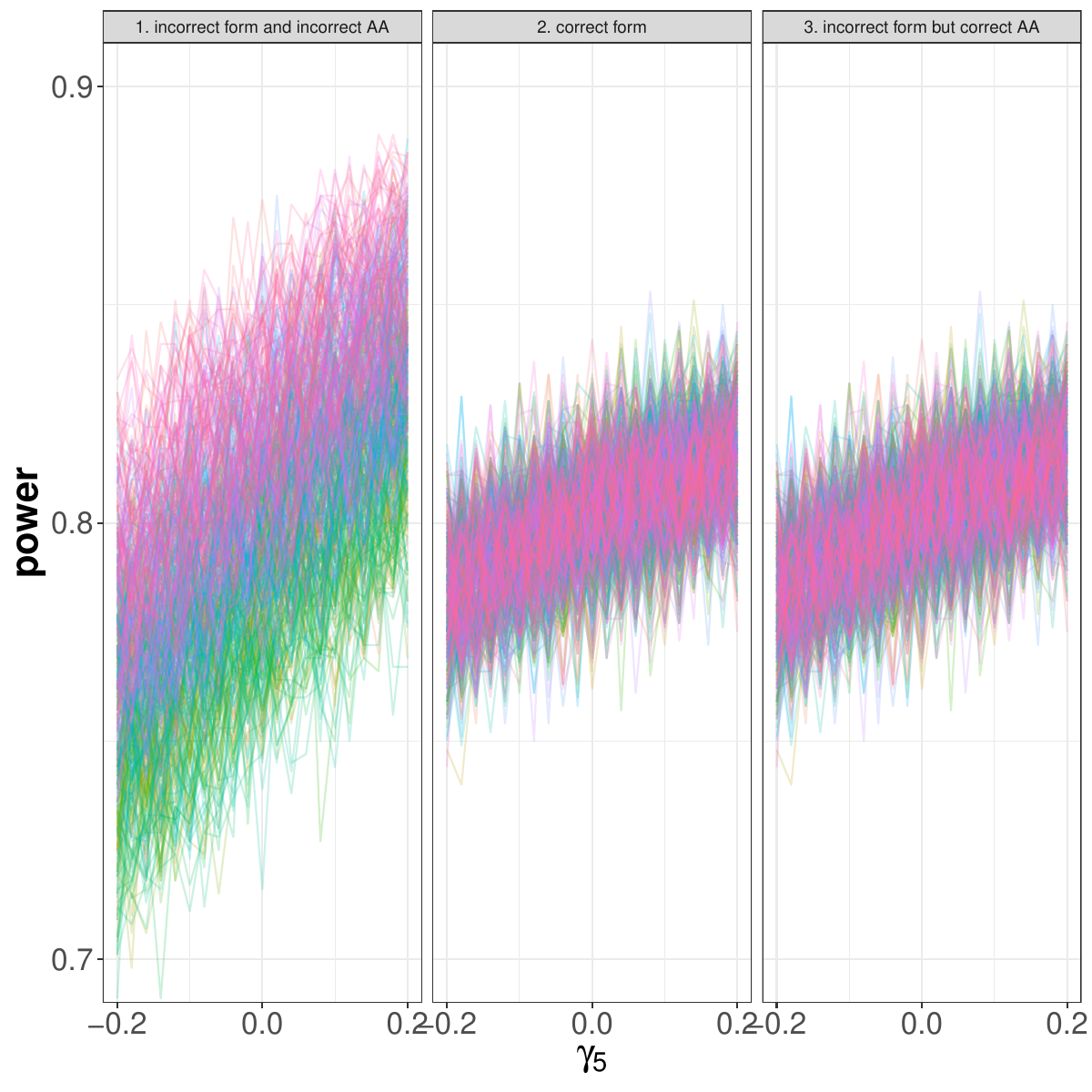}
        \caption{}  
        \label{subfig:combined_endogeneity}
    \end{subfigure}
    \caption{Power when (WA-e) is violated in that there is an ``endogeneity'' effect on the availability. \textbf{Panels (a):} Availability depends on the previous treatment, $\nu_2$ represents the effect of treatment 1 and $\nu_3$ represents the effect of treatment 2. \textbf{Panel (b):} Availability depends on the one of previous treatment and previous outcome, $\nu_2$ represents the effect of treatment 1 and $\nu_4$ represents the effect of the previous outcome. \textbf{Panel (c):} combination of \textbf{Panel (a)} and \textbf{Panel (b)}; the colors distinguish the combinations of $(\nu_2,\nu_3)$ values. In the axis and legend labels of all three panels, $\gamma_3$, $\gamma_4$ and $\gamma_5$ are the names of $\nu_2$, $\nu_3$ and $\nu_4$ in the simulation code.}
    \label{fig:endogeneity}
\end{figure}


\end{document}